\documentclass[journal]{IEEEtran}
\usepackage{macros}

\crefname{appsec}{Appendix}{Appendices}

\begin{document}
\maketitle

\begin{abstract}

We furnish a procedure based on universal hash families (UHFs) that can convert an error correcting coding scheme (ECC) of rate $R$ into a \textit{semantically} secure wiretap coding scheme of rate $R - \xi$ where $\xi$ is a parameter derived from the eavesdropper's point-to-point channel. This conversion is shown to be polynomial time efficient with block length and is applicable to any channel, i.e., both discrete and continuous channels. When an ECC is chosen, our procedure induces a wiretap coding scheme that is concrete and efficient as long as the ECC is also such. To prove this induced wiretap coding scheme is semantically secure, we have constructed bounds on the information leaked to the eavesdropper. Our construction is an upgrade of bounds from recent literature: the novelty here being that our leakage bounds hold for \textit{any} message distribution. Indeed, our wiretap procedure using UHFs and our characterization of its semantic leakage is the first main contribution of this work.

The other main contribution of this work is as follows. We apply the aforementioned procedure to a variety of wiretap channels in order to show the procedure's efficacy, and as a result of such applications, we mirror existing results from literature regarding achievable semantically secure rates. More notably, in some cases our results establish \textit{new} achievable semantically secure rates. For DMC wiretap channels and No-CSIT (instantaneous \underline{c}hannel \underline{s}tate \underline{i}nformation at the \underline{t}ransmitter) fast fading wiretap channels, we show how our wiretap scheme can achieve the secrecy capacity in certain cases, but more generally, can always achieve a non-negative rate of $R - C_E$ under semantic security where $R$ is the rate of the ECC on the main channel and $C_E$ is the capacity of the eavesdropper's point-to-point channel. On partial CSIT fast fading wiretap channels, we show that our wiretap coding scheme can achieve the best known secure achievable rates from literature, even under semantic security. On full CSIT fast fading wiretap channels, we show that our wiretap coding scheme can achieve the secrecy capacity. On AWGN wiretap channels, using a recent ECC from literature, we provide an end-to-end wiretap coding scheme that is concrete, polynomial time efficient in block length, semantically secure, and has both its probability of error and  semantic leakage exponentially diminishing with block length. In fact, we prove that the semantic leakage in each of the previous cases is exponentially decreasing with block length.

\end{abstract}


\section{Introduction}
\IEEEPARstart{P}{hysical} layer security exploits the inherent randomness in a communication environment to derive security; this form of security makes no assumptions on the eavesdropper's capabilities. This is in direct contrast to computational based security which derives security based on the assumption that the eavesdropper has \textit{bounded} computational resources. 

Computational based security has been the de facto security for communication systems since its inception due especially to its ease of implementation; however, the main assumption of computational boundedness has been scrutinized in recent years more than ever. One of the primary reasons for this scrutiny is the potential advent of practical quantum computers in the near future. On the other hand, physical layer security is impervious to advances in computing, in particular quantum computing, because it makes no underlying assumptions on computational resources. Thus, regardless of the technology the eavesdropper possesses, physical layer security maintains its integrity. In this way, physical layer security is \textit{inherent} security.

Given this clear advantage of physical layer security, it is still underutilized in modern communication systems. This is primarily because most proposed schemes to implement physical layer security are too impractical. The schemes are most often only theorized to exist with a tangible construction unknown, i.e., proofs are by existence and not by construction. Moreover, even when a construction is given, it is rarely efficient in block length.

Overcoming these hurdles has been one of the primary aims of the physical layer community for quite some time. But there is yet another reason physical layer security has not found common use in new communication systems; this reason is significantly more subtle. The measure of security provided by most physical layer security schemes is insufficient to be used in a practical setting. 

There is no direct analog of this problem that arises from computational based security because in that case the underlying assumption that certain decision problems are computationally hard is unproven anyway. Here, in physical layer security where security is rigorously proven, the choice of how security is measured needs to be consistent with reality if the proof of security is to hold any merit. 

If a physical layer scheme could be created that is tangible, efficient, utilizes the most realistic measure of security, and achieves an input/output rate near the theoretical maximum, then physical layer security could potentially rival computational based security as the de facto security of modern communication systems, or at the very least could be an indispensable component. Motivated by this, herein we develop a physical layer coding scheme that \textit{aims} to satisfy all of these properties and in some cases even \textit{does}.

\subsection{Background - Security Metrics}
Physical layer security is often modeled by a wiretap channel which was introduced in the 1970's by Wyner \cite{Wyner} and later generalized by Csisz\'ar and K\"orner \cite{csiszarkorner}. The metric used to measure security in these works is now colloquially referred to as the \textit{weak security metric}. For years, this was the primary metric used to measure security on wiretap channels, however, it was asserted in the 1990's in \cite{strongsecurity} that the weak metric provided an inadequate measure of security to be deemed practical. This led to the creation of the \textit{strong security metric}, the unnormalized version of the weak metric.

This metric sufficed for awhile, but in 2012, this metric was again shown to be an inadequate measure of security for realistic communication systems by Bellare, Tessaro, and Vardy \cite{cryptoTreatment}. In addition to showing this, they created three new security metrics provably stronger than the strong security metric and proved them asymptotically equivalent. For the purposes of this paper, due to their equivalence, we will refer to all three of these metrics collectively as the \textit{semantic security metric}, the name given in \cite{cryptoTreatment}. This metric is now held to be the gold standard of security metrics for the wiretap channel. Moreover, it is argued that a stronger security metric than the semantic security metric does not exist. For these reasons, \textit{it is the only measure of security that should be utilized in practice}. Admittedly, proving results with this metric tend to be more arduous, therefore many results in literature still use the strong security metric and even the weak security metric, but in this work we will \textit{exclusively} use the semantic metric to prove security.

\subsection{Background - Fading Channels}
In addition to focusing on physical layer security schemes that are tangible, efficient, and utilize semantic security, we will be primarily concerned with the most realistic of wiretap channel models: the fading wiretap channel. Fading wiretap channels are commonly used to model security of wireless communications. It assumes the input signal is attenuated/amplified then corrupted by some additive noise. The amount of attenuation/amplification is called the channel state. When the channel state changes frequently and independently, we are in the so called \textit{fast fading} regime. This is one of the most practical fading wiretap channel models and is the main focus of our applications.

Due to the nature of wireless systems, fading wiretap channels sometimes assume that the current channel state is fed back from the receiver to the transmitter (this is abbreviated by CSIT - instantaneous channel state information at the transmitter). However, since there are actually two point-to-point channels within a wiretap channel, the transmitter potentially receives both of these channel states, a channel state corresponding to the intended receiver's channel and a channel state corresponding to the eavesdropper's channel.

We denote the case when the transmitter knows neither of these channel states by \bfit{No-CSIT}, although we do assume the transmitter knows the \textit{statistics} of the channel states as random variables. We denote the case when the transmitter knows the intended receiver's current channel state but not the eavesdropper's current channel state (only the statistics) by \bfit{partial CSIT}. Lastly, we denote the case when the transmitter knows both current channel states by \bfit{full CSIT}. 

The level of CSIT drastically changes which secure rates are achievable. For this reason, we will treat No-CSIT, partial CSIT, and full CSIT as separate wiretap channels entirely.

\subsection{Related Work}
In \cite{cryptoTreatment} and also in \cite{semanticallySecure}, a tangible (concrete) and efficient wiretap coding scheme was given that could achieve positive secrecy rates on discrete memoryless wiretap channels under \textit{semantic} security. In certain cases, this wiretap scheme could also achieve the \textit{semantic} secrecy capacity \cite{channelupgrading}. In \cite{explicitGaussianWiretap}, this scheme was extended for use on the AWGN wiretap channel and was shown to achieve the secrecy capacity, however, the wiretap scheme therein was only able to achieve positive secrecy rates under the \textit{strong} security metric. In \cite{UHF}, however, this wiretap scheme was shown to achieve the \textit{strong} secrecy capacity for both continuous and discrete wiretap channels. Their proof is a direct bound on the strong leakage and admits a nice characterization of the secure achievable rates. In \cite{AWGNsemanticRecent}, a wiretap scheme was shown to achieve the \textit{semantic} secrecy capacity of AWGN wiretap channels, albeit in a completely different manner than the previously mentioned five papers. To date, there is currently no universal wiretap scheme that achieves the semantic secrecy capacity for both discrete memoryless \textit{and} AWGN wiretap channels.

Physical layer security for fast fading wiretap channels was arguably started with Liang, Poor, and Shamai in \cite{secureoverfading} where they found the weak secrecy capacity of the fast fading wiretap channel with the assumption of full CSIT. This was later improved by Bloch and Laneman in \cite{chresolv} where they determined the secrecy capacity of this channel under the strong secrecy metric. In a different direction, Bloch and Laneman \cite{PartialCSIT_2013} considered the case of fast fading wiretap channels with partial CSIT; they gave a set of achievable secrecy rates under the strong secrecy metric for this channel. Their solution relies on an optimization problem that has no closed form solution and thus it represents the \textit{best known} secrecy rate on the fast fading channel with partial CSIT. In the case of fast fading channels with No-CSIT, it was only recently shown in \cite{us}, \cite{stochasticdegrade}, \cite{mukherjee_ulukus_2013} that positive rates are actually achievable and an upper bound for the secrecy capacity is also derived. For a special class of fast fading No-CSIT channels, \cite{stochasticdegrade} actually finds the secrecy capacity of these channels under the weak secrecy constraint. In \cite{almostuniversal}, a positive semantically secure achievable rate is obtained for fast fading channels with No-CSIT. To date, there are few results involving semantic security on fast fading wiretap channels. In particular, no one has constructed a wiretap scheme that achieves the best possible semantically secure rates for each case of CSIT. Moreover, hardly any wiretap schemes exist for fast fading channels that are tangle/efficient \textit{and} come close to the best possible rates, even in the lesser weak and strong cases.

\subsection{Summary of Results}
The main purpose of this paper is to amplify results of physical layer security into a more practical setting. We prove all of our results using the semantic security metric, the most demanding security metric in this field. Our wiretap coding scheme developed is modular in the sense that it can immediately be adapted to any existing channel to provide semantic security; furthermore, it is shown to be concrete and efficient\footnote{As will be made clear in \Cref{scheme}, we only prove the \textit{preprocessor} is concrete and efficient; however, if the error correcting code is also such, then so is the entire wiretap coding scheme.}.

To prove our wiretap coding scheme is semantically secure, we bound the semantic leakage asymptotically (\Cref{lem:LHL}). We do this by upgrading the strong leakage bounds found in \cite{UHF}. In particular, we optimize over all message distributions. As in \cite{explicitGaussianWiretap,UHF}, our wiretap scheme is a modular scheme consisting of a preprocessor based on UHFs. However, in order to guarantee that our scheme is semantically secure, we require the UHF to also have \textit{additional} properties (we dub UHFs with these additional properties as semantically secure universal hash families - SS-UHFs). The additional properties are non-restrictive in general and we provide a particular implementation of an SS-UHF based on finite field arithmetic that is concrete and quadratic time efficient. In effect, our SS-UHF based preprocessor is a converter that takes in an off-the-shelf ECC and converts it to a semantically secure wiretap coding scheme (\Cref{thm:LHL3}).

In \Cref{procedure} below, we outline the necessary steps for using our wiretap scheme on an arbitrary wiretap channel. Use of this procedure attains semantic security for any wiretap channel contingent on certain conditions being satisfied which are derived from the wiretap channel. We show that these conditions are indeed satisfied for the DMC, AWGN, and fast fading wiretap channels where we examine the fading channels with various levels of instantaneous channel state information at the transmitter. In other words, we \textit{demonstrate} this procedure, in effect, proving that our wiretap coding scheme can achieve semantically secure rates on these channels.

The following are our specific contributions on each of the aforementioned channels.
\begin{itemize}
\item \textit{DMC} - In \Cref{thm:DMWC}, we reestablish the result given by Tal and Vardy's upgrade \cite{channelupgrading} of Bellare, Tessaro, and Vardy's original result \cite{cryptoTreatment,semanticallySecure}; that is, we show our wiretap coding scheme achieves the semantic secrecy capacity of any symmetric, degraded, discrete memoryless wiretap channel. However, we allow \textit{any} ECC for the main point-to-point channel in our construction. This is in contrast to the previous results that impose certain restrictions on the ECC in order to achieve secrecy capacity. 
\item \textit{AWGN} - In \Cref{thm:awgnrates}, we reestablish \cite{AWGNsemanticRecent} by constructing a concrete, end-to-end efficient wiretap scheme and prove that it can achieve the secrecy capacity on the AWGN wiretap channel under semantic security. However, we prefer our wiretap scheme in the fact that it is modular: the same preprocessor used here can be used on any channel without modification. 
\item \textit{No-CSIT} - In \Cref{thm:ImaxLessCE}, we prove that our wiretap scheme achieves the semantic secrecy capacity here for the case when the eavesdropper's channel is \textit{stochastically degraded} (cf. \cite{stochasticdegrade}) with respect to the main channel. Furthermore, in other cases, we provide a set of semantically secure achievable rates. 
\item \textit{Partial CSIT} - In \Cref{thm:ImaxPartialCSIT}, we prove that our wiretap scheme achieves the best known achievable secrecy rates to date (cf. \cite{bloch_barros_2011}) with semantic security.
\item \textit{Full CSIT} - In \Cref{thm:ImaxFullCSIT}, we prove that our scheme can actually achieve the strong secrecy capacity in this setting with semantic security thereby proving that semantic secrecy capacity is equivalent to the strong secrecy capacity and hence also the weak secrecy capacity. 
\end{itemize}

All of the achievable semantically secure rates on these channels can be attained concretely and efficiently (\Cref{prop:SSUHF} and \Cref{prop:UHFisEfficient}) - since our preprocessor is already such, one only needs to concentrate on finding an error correcting code that is concrete and efficient. Once this is done, the entire wiretap coding scheme is concrete and efficient! In other words, we have converted the problem of finding \textit{good} wiretap coding schemes into a problem of finding \textit{good} error correcting coding schemes where good here means concrete and efficient.

To recap, we give in this paper a procedure for attaining semantically secure rates in a concrete and efficient way for arbitrary wiretap channels. We apply this procedure in particular to the five aforementioned channels. Therefore, if the reader desires to attain semantically secure rates on one of these channels, all that remains is to find an error correcting code. As a special case, we have pointed the reader to an ideal error correcting code for the AWGN wiretap channel, thereby completing the procedure in this case in full. If the reader wants to attain semantic security on a wiretap channel not listed above, then the reader must apply \Cref{procedure} in its entirety. Specifically, the reader must check that the hypothesis of \Cref{thm:LHL3} is satisfied for that channel.

\subsection{Outline}
The remainder of the paper is organized as follows. \Cref{sec:Prelims} introduces notation and gives the preliminary mathematical background necessary to proceed through the rest of the paper. \Cref{scheme} presents our modular wiretap coding scheme and gives a concrete and efficient implementation of the preprocessor based on finite field arithmetic. \Cref{sec:LHL} analyzes both the security and achievable rates of our proposed wiretap scheme and gives a procedure for how to utilize our main results on an arbitrary (discrete or continuous) wiretap channel. In \Cref{sec:App1}, we apply this procedure to the DMC and AWGN wiretap channels as a first application and show how our wiretap scheme \textit{replicates} the best results from literature. \Cref{section:fading} considers fast fading wiretap channels with various levels of CSIT (No-CSIT, partial CSIT, and full CSIT) and gives semantically secure achievable rates for each of these. Moreover, we show how our wiretap scheme in these cases \textit{exceeds} the best results from literature.

In an attempt to give a more polished presentation, we have assigned nearly all of the proofs to the appendices.

\section{Preliminaries}\label{sec:Prelims}

\subsection{Notation and Conventions}
We shall write $a^n$ to denote an $n$-dimensional vector where $a_i$ denotes the $i$-th component, i.e., $a^n = (a_1,\dots,a_n)$. We use the usual notation $\norm{a^n}$ to denote the Euclidean norm. We shall denote the indicator (or characteristic) function by $\indicator_\mathcal{A}(x)$ or $\indicator\left(x \in \mathcal{A} \right)$ and will take all logarithms in this paper to be base $2$ unless we write $\ln$, for which we mean the logarithm of base $e$. We will write $\N$, $\R$ and $\C$ to denote the set of natural, real, and complex numbers respectively. With a slight abuse of notation, we will write $\R_{+}$ to denote the set of \textit{non-negative} reals. We will write $|\mathcal{A}|$ to denote the cardinality of set $\mathcal{A}$.

We will denote random variables by capital letters and will denote the spaces for which a random variable is defined by a respective scripted letter, e.g.,\ $A$ is a random variable with values in $\mathcal{A}$. As usual we write $X \in \U{\left( \mathcal{X} \right)}$ to denote that $X$ is a uniform random variable over some discrete set $\mathcal{X}$; we write $Y \in \mathcal{N}\left(a,b\right)$ to denote that $Y$ is a real Gaussian random variable with mean $a$ and standard deviation $b$; we write $Z \in \mathcal{CN}\left(a,b\right)$ to denote that $Z$ is a circularly symmetric complex Gaussian random variable with mean $a$ and standard deviation $b$.

We shall use the notation of \cite{csiszarkorner,UHF} and let $I(A \sep B)$ denote the usual mutual information between random variables $A$ and $B$. 
We write $A \bot B$ when random variable $A$ is independent of $B$. We write $\P\left[\mathcal{A}\right]$ to denote the probability of event $\mathcal{A}$ and $\E \left[A\right]$ to denote the expected value of random variable $A$. When we want to be explicit about which random variable we are taking the probability (resp. expected value) with respect to, we shall denote the random variable by a subscript. 

We denote all probability densities\footnote{Sometimes when we have a probability mass function, instead we will use the notation $P(\cdot)$ with appropriate subscripts as necessary.} by $\omega(\,\cdot\,)$ defined by the Radon-Nykodym derivative with respect to some implicit reference measure; we will almost always denote this reference measure by $\mu$. We denote the conditional probability density in an analogous way as $\omega(\,\cdot\,|\,\cdot\,)$. As an example of our notation, if $A$ and $B$ are random variables on $\mathcal{A}$ and $\mathcal{B}$ respectively, then $\omega(a)$ denotes the probability density of $A$ and $\omega(b|a)$ denotes the conditional probability density of $B$ given $A = a$.

When algorithms are completed in polynomial time (in the worst case) then we take up the standard convention and call such algorithms \bfit{efficient}.

\subsection{Channels}
Let $\X$ and $\Y$ be sets. We shall denote a \bfit{stochastic map} by $T: \X \leadsto \Y$. Given $x \in \X$, a stochastic map assigns a likelihood that $x$ will map to a certain $y \in \Y$. For each $x \in \X$, this induces the random variable $T(x)$. The support of this random variable, $\text{supp}(T(x)) \subset \Y$, is the elements in $\Y$ that $T$ can map $x$ to with non-zero likelihood.

Let $T: \X \leadsto \Y$ be some stochastic map, $X$ a random variable on $\X$, $\mu$ some reference measure on $\Y$, and $Y = T \circ X = T(X)$. We will call the conditional density $\omega(y|x)$ the \bfit{transition density} of the stochastic map $T$ and we will call the tuple $(\X, \omega(y|x),\Y)$ a \bfit{channel}. We will often abuse language/notation and call $T$ itself a channel. The transition density probabilistically tells us how the channel is mapping $\X$ to $\Y$. Given that some symbol $x \in \X$ was sent across the channel, the probability that $Y$ is in some subset $\mathcal{U} \subset \Y$ is given by $\int_{\mathcal{U}} \omega(y|x) \mu(dy)$.

For the rest of this paper, we will be considering \textit{subnormalized} channels: channels with transition densities such that $\int_{\Y} \omega(y|x) \mu(dy) \leq 1$. This is a technical condition that allows us to define the following. Given a channel $T = (\X, \omega(y|x), \Y)$ and subset $\T \subset \X \times \Y$ denote $\omega_\T(y|x) = \omega(y|x) \indicator\left( (x,y) \in \T \right)$. This induces a restricted channel as follows. Given that $x \in \X$ was sent across the restricted channel, the probability that $Y$ is in some subset $\mathcal{U} \subset \Y$ is given by $\int_{\mathcal{U}} \omega_\T(y|x) \mu(dy) = \int_{\mathcal{U}}\omega(y|x)\indicator\left((x,y) \in \T \right) \mu(dy)$.

\subsection{Error Correcting Codes}
We will always refer to the number of channel uses\footnote{Note that we are only considering discrete-time channels in this work.} as the \textit{block length} (of the code) and denote it by $n$. As usual, we will mainly be considering the $n$-letter extension of channel $T$ notated by $T^n = (\X^n, \omega(y^n|x^n), \Y^n)$.

Let $\M'$ be some finite message set. An $n$-length encoder for $T^n$ is an injective function $e_n: \M' \to \X^n$. The image $e_n(\M') \subset \X^n$ is called the codebook and is denoted $\codebook_n$. Elements of the codebook are referred to as codewords. An $n$-length decoder for $T^n$ is a function $d_n: \Y^n \to \M'$ and an $n$-length code is a tuple $\code_n = (e_n,d_n)$. The rate of the code is given by $R_{\code_n} = \frac{1}{n}\log|\M'|$. Lastly, a family of codes $\{\code_n\}_{n \in \N}$ is called a coding scheme $\code$ with rate given by $R_\code = \lim_{n \to \infty } R_{\code_n}$, where we assume this limit exists.

The maximum probability of error for code $\code_n$ is given by $\P_e (\code_n) = \max_{M' \in \M'} \P[(d_n \circ T^n \circ e_n) (M') \neq M']$. If $\P_e (\code_n)$ is sufficiently small then $\code_n$ is called an error correcting code (ECC). If every code in scheme $\code$ is an ECC, we call $\code$ an ECC scheme. If $\P_e (\code_n) \to 0$ as $n\to \infty$ then we say the scheme $\code$ is reliable. In particular, if $\log(\P_e (\code_n)) \leq -a n^b$ for some constants $a,b >0$ and for every $n$, then we call the ECC scheme $\code$ \bfit{exceptionally reliable}.

\begin{remark}
It was noted in \cite{cryptoTreatment} that ``good'' error correcting coding schemes in practice should satisfy the reliability condition exponentially fast; they called such ECC schemes ``strongly reliable.'' Due to the plethora of definitions containing the wording ``strong'' in the literature, we have instead called such ECC schemes here ``exceptionally reliable.''
\end{remark}

For continuous channels (i.e. $\X = \Y = \R$) we shall always impose the average power constraint as usual. In more detail, for some fixed constant $P$, we shall require the code to satisfy $\frac{1}{n} || x^n||^2 \leq P$ for every $x^n \in \codebook_n$.

The supremum of reliable achievable rates over all ECC schemes is known as the (point-to-point) channel capacity. We shall denote the channel capacity of a channel $T$ by $C_T$.

\subsection{Wiretap Codes}\label{sect:wiretapcodes}
Let $T = (\X, \omega(y|x), \Y)$ be a channel that models the communication between a transmitter Alice and intended receiver Bob. Let $E = (\X, \omega(z|x), \Z)$ be a channel modeling the unintended communication between Alice and a passive eavesdropper Eve. We call the pair of channels $W = (T,E)$ the \bfit{wiretap channel}.

Note that we have chosen the letters $T$, $E$, and $W$ so as to denote the \underline{T}ransmission channel, \underline{E}avesdropper's channel, and \underline{W}iretap channel. We also note that the $n$-letter wiretap channel is given by $W^n = (T^n, E^n)$.

The goal of physical layer security as modeled by a wiretap channel is for Alice to communicate information reliably to Bob while keeping that same information hidden from Eve. Let $M \in \M$ be the random variable representing the message Alice wants to impart to Bob yet keep secret from Eve. Let $Z^n \in \Z^n$ be the $n$-letter random variable representing Eve's output. To measure security, we recall the most common security metrics. 
\begin{itemize}
\item \cite{Wyner} \bfit{Weak}: \[\frac{1}{n} I(M \sep Z^n), \quad M \sim \U(\M).\]
\item \cite{maurer_1994} \bfit{Strong}\footnote{Strong security is sometimes referred to as MIS-R, cf. \cite{cryptoTreatment}.}: \[ I(M \sep Z^n), \quad M \sim \U(\M).\]
\item \cite{cryptoTreatment} \bfit{Semantic}: \[\max\limits_{P_M} I(M \sep Z^n).\]
\end{itemize}
We refer to each of these quantities as \bfit{leakage} and we say that a  coding scheme is secure under a given metric if its respective leakage goes to 0 as $n\rightarrow \infty$. In a similar fashion to exceptional reliability, we say that a coding scheme is \bfit{exceptionally secure} if the leakage is vanishing exponentially fast with $n$. 

\begin{remark}
The expression for semantic security above is technically called mutual information security (MIS) as originally defined in \cite{cryptoTreatment}. Semantic security (in the wiretap context) is actually defined using guessing probabilities. However, therein it was shown for discrete channels (and in \cite{continuousSemantic} for continuous channels) that MIS was equivalent to semantic security asymptotically. Thus, in the asymptotic regime there is no need to differentiate between the two metrics because each implies the other. Hence, our choice of name is technically justified.

However, one may still ask why we call the definition above ``semantic security'' when it is actually the definition of MIS; the reasoning is as follows. The definition of semantic security in \cite{cryptoTreatment} is named such to allude to the gold standard definition from computational based security \cite{goldwasser_micali_1984}. However, the definition of semantic security is considerably less tractable than the definition of MIS. In order to get the best of both worlds, we have chosen our naming convention. We note that it is a convention already followed by other works.
\end{remark}

Let $\wiretap = \{\wiretap_n\}_{n \in \N}$ be a coding scheme for channel $T$ (and inherently channel $E$) using message set $\M$. We say $\wiretap$ is a \bfit{$\mathfrak{X}$-wiretap coding scheme}, where $\mathfrak{X} \in \{\text{weak},\text{strong},\text{semantic}\}$, if it satisfies each of the following.
\begin{itemize}
\item \bfit{Reliability:} $\wiretap$ is a reliable ECC scheme for $T$.
\item \bfit{Security:} $\wiretap$ is secure (relative to $E$) using the $\mathfrak{X}$-metric.
\end{itemize}

If these two conditions are satisfied \textit{exceptionally}, then we say that $\wiretap$ is an \bfit{outstanding $\mathfrak{X}$-wiretap coding scheme}.

If $\Rs=\lim_{n\to\infty} \frac{1}{n} \log|\M|$ is the rate of an $\mathfrak{X}$ wiretap coding scheme, then we say $\Rs$ is an \bfit{$\mathfrak{X}$ achievable secrecy rate.} We call the supremum of all $\mathfrak{X}$ achievable secrecy rates the \bfit{$\mathfrak{X}$ secrecy capacity} denoted by $\Cs\bigr|_\mathfrak{X}$ or simply $\Cs$ when the metric is clear from context.
\begin{fact}\label{fact:CsEqual}
If all secure rates $\Rs$ achievable under the weak secrecy metric are also achievable under the semantic secrecy metric, then:
\[\Cs\bigr|_\text{weak}=\Cs\bigr|_\text{semantic}.\]
\end{fact}

\subsection{Universal Hashing}

Let $\M = \{0,1\}^k$ be the set of binary strings of length $k$, $\M'$ and $\seed$ be finite sets, and $S$ a \textit{uniform} random variable on $\seed$. Consider now a family of a finite number of functions indexed by $\seed$: \[\mathcal{F} = \{ f_s:\Mprime \to \M \,|\, s \in \seed\}.\]
\begin{enumerate}[(i)]
\item $\mathcal{F}$ is called a \bfit{universal hash family} (UHF) if for every $m'_1 \neq m'_2 \in \M'$,
\[|\{s\in \seed\,|\, f_s(m'_1) = f_s(m'_2)\}| \leq \frac{|\seed|}{2^k}.\]
\item $\mathcal{F}$ is called \bfit{uniform} if for every $m' \in \M'$ and for every $m \in \M$,
\[|\{s\in \seed\,|\, f_s(m') = m\}| = \frac{|\seed|}{2^k}.\]
\item $\mathcal{F}$ is called \bfit{$b$-regular} if for every $s \in \seed$ and for every $m \in \M$, \[|\{m' \in \Mprime \,|\, f_s(m') = m\}| = 2^b. \] 
\item $\mathcal{F}$ is called \bfit{invertible} if for each $s \in \seed$ there exists some stochastic mapping $\phi_s : \M \leadsto \M'$ such that for all $m \in \M$ and $y \in \text{supp}(\phi_s(m))$, $f_s(y) = m$. If $\phi_s(m)$ is a \textit{uniform} random variable for every $s \in \seed$ and $m \in \M$ then we call $\mathcal{F}$ \bfit{evenly invertible}.
\item Lastly, we call $\mathcal{F}$ a \bfit{semantically secure universal hash family} (SS-UHF) if it is: (i) universal, (ii) uniform, (iii) $b$-regular, and (iv) evenly invertible.
\end{enumerate}

Many of the definitions here coincide with those found in computer science literature. The conditions of being a \textit{universal hash family} (as introduced in \cite{CARTER}) and \textit{uniform} are found in most textbooks on hash families. The condition of being \textit{$b$-regular} and \textit{invertible} can be found in \cite{semanticallySecure} and \cite{UHF}. That being said, we have invented some terminology. We have dubbed hash families that are universal, uniform, $b$-regular, and evenly invertible as \textit{semantically secure universal hash families} to emphasize that hash families with these four properties are the proper ones for inducing semantic security (see \Cref{sec:LHL}).

\subsection{$\epsilon$-smooth $\alpha$-Mutual Information}
In order to measure the amount of information leaked to the eavesdropper using our wiretap scheme, we will need to employ the use of a different measure of information, known as $\alpha$-mutual information. $\alpha$-mutual information is defined using R\'enyi entropy and is actually a generalization of the usual mutual information defined by Shannon. 

For a discrete random variable $M'$ over $\M'$, the following generalizes Shannon's entropy and is called R\'enyi entropy of order $\alpha \in (1,\infty)$ \cite{renyi1961measures}:
$H_\alpha (M') = \frac{1}{1-\alpha} \log \left( \sum_{m'} \omega(m')^\alpha  \right)$. This can be extended by continuity to the cases of $\alpha = 1$ and $\alpha = \infty$ where $H_1(M')$ is the usual Shannon entropy and $H_\infty(M')$ is the usual min-entropy. In particular, when $M'$ is \textit{uniform}, for any $\alpha \in [1,\infty]$ we have $H_\alpha (M') = \log( |\M'| )$, a fact we will use frequently.

In a similar way, one can define conditional R\'enyi entropy, however, there is no universal notion of such a definition in literature as different definitions can be employed based on the specific properties one desires (cf. \cite{Renyi2017,Berens}). We will be using Arimoto's definition \cite{arimoto1977information,ConditionalRenyi2018} given as follows. 

Let $Z^*$ be an arbitrary random variable over $\Z^*$ (with measure $\mu$ on $\Z^*$) and $M'$ a discrete random variable over $\M'$. Then conditional R\'enyi entropy of order $\alpha \in (1,\infty)$ is given by:
\begin{align*}
&H_\alpha(M'|Z^*)\\
&\quad = \frac{\alpha}{1-\alpha} \log  \int_{\Z^*} \omega(z^*)  \left(\sum_{m'}\omega(m'|z^*)^\alpha \right)^{\frac{1}{\alpha}} \mu(dz^*).
\end{align*}
Just as in the case of (unconditioned) R\'enyi entropy, this definition can be extended to the cases of $\alpha = 1$ and $\alpha = \infty$ by continuity. For $\alpha \to 1$, one easily checks using L'Hospitals rule that $H_\alpha (M' |Z^*)$ becomes $H(M'|Z^*)$, the conditional Shannon entropy. For $\alpha \to \infty$, the definition becomes
\[H_\infty (M'|Z^*) = -\log \int_{\Z^*} \omega(z^*) \max_{m'} \omega(m'|z^*)\mu(dz^*),\] and is often referred to as conditional min-entropy. Another important case for which we would like to emphasize is when $\alpha = 2$: \[H_2(M'|Z^*) = -2\log \int_{\Z^*} \omega(z^*) \left(\sum_{m'}\omega(m'|z^n)^2 \right)^{\frac{1}{2}} \mu(dz^*),\] which is often referred to as conditional collision entropy.

Now let us finally define $\alpha$-mutual information: the R\'enyi extension to Shannon's mutual information. Again, there is no universal definition in literature but we will be using the definition put forth in \cite{ConditionalRenyi2018} for the special case when $M'$ is a uniform random variable.

Let $M'$ and $Z^*$ be random variables as before except now we require $M'$ to be uniform over $\M'$. For $\alpha \in [1,\infty]$ we define the \bfit{$\alpha$-mutual information} between $M'$ and $Z^*$ by \[I_\alpha(M' \sep Z^*) = \log |\M'|  - H_\alpha(M'|Z^*).\]
Notice that $I_1(M' \sep Z^*)$ is exactly Shannon's mutual information $I(M' \sep Z^*)$ so in this case we will drop the subscript. Moreover, for the case of $\alpha = 2$, we will often call $I_2(M' \sep Z^*)$ \bfit{collision-information} and for the case of $\alpha = \infty$, we will often call $I_\infty(M' \sep Z^*)$ \bfit{max-information}. 

\begin{fact}\label{fact:renyiOrdering1} \cite{ConditionalRenyi2018,arimotoProperties}
 For any $\alpha \in [1,\infty]$, $I_\alpha(M'\sep Z^*)$ is monotonically increasing in $\alpha$.
\end{fact}
Note that this fact justifies the name of $I_\infty(M'\sep Z^*)$ as max-information because it measure the \textit{most} amount of information of all of the $\alpha$-mutual informations.

The $\alpha$-mutual information also admits several other desirable properties of an ``information measure'' which can be found in \cite{arimotoProperties}. Note however that this definition of $\alpha$-mutual information is not symmetric in its arguments and does not satisfy the chain rule in general. This of course is in contrast to Shannon's mutual information.

To facilitate our proofs later on we will also need a concept called $\epsilon$-smooth $\alpha$-mutual information. Basically, we will define $\alpha$-mutual information on a portion of the entire space that probabilistically contains enough content up to some $\epsilon$. To make this rigorous we first introduce the concept of a typical set.

For $\epsilon \geq 0$, we call a subset $\T \subset \M' \times \Z^*$ a \bfit{$(1-\epsilon)$-typical set} if \[\P\left[(M',Z^*) \in \T \,|\, M' = m' \right]\geq 1-\epsilon, \qquad \forall m' \in \M'.\]
Furthermore, we will denote the set of all $(1-\epsilon)$-typical sets by $\TT$. Typical sets intuitively contain almost all that there is to know about our space up to some $\epsilon$, hence the name typical. 

For some typical set $\T$, we first define the conditional R\'enyi entropy of order $\alpha$ restricted to $\T$. This is simply given by 
\begin{align*}
&H_\alpha^\T(M'|Z^*)\\
&\quad = \frac{\alpha}{1-\alpha} \log  \int_{\Z^*} \omega(z^*)  \left(\sum_{m'}\omega_\T(m'|z^*)^\alpha \right)^{\frac{1}{\alpha}} \mu(dz^*).
\end{align*}

Given $\epsilon \geq 0$ define \bfit{$\epsilon$-smooth $\alpha$-mutual information} for $M'$ uniform over $\M'$ by
\[I_\alpha^\epsilon(M'\sep Z^*) = \inf\limits_{\T \in \TT} I_\alpha^\T(M' \sep Z^*) ,\] where $\alpha$-mutual information evaluated on $\T$ is given by \[\ I_\alpha^\T(M' \sep Z^*) = \log|\M'| - H_\alpha^\T(M'|Z^*).\]

Given some threshold $\epsilon$, we find the smallest value that $\alpha$-mutual information could possibly be when defined on the \textit{subnormalized} channels corresponding to those sets that contain \textit{enough} probability with respect to our threshold. Later, we will bound the leakage between the transmitter and eavesdropper as an increasing function of this metric; thus, defining $\epsilon$-smooth $\alpha$-mutual information using the infimum provides the tightest bound we should expect when $\epsilon$ is our threshold. 

Note that when $\epsilon = 0$, $\TT$ contains only sets equal to the entire space less a set of measure zero and hence $I_\alpha^0(M' \sep Z^*) = I_\alpha(M'\sep Z^*)$.

Analogous to \Cref{fact:renyiOrdering1} we have the following ordering for $\epsilon$-smooth $\alpha$-mutual information, a result we will use in proving our wiretap scheme is secure.

\begin{lemma}\label{fact:renyiOrdering2} For any $\T \subset \M' \times \Z^*$ and $\alpha \in [1,\infty]$, $I_\alpha^\T(M'\sep Z^*)$ is monotonically increasing in $\alpha$.
\end{lemma}
\begin{proof}
This follows easily from the proof given for \cite[Proposition 1]{ConditionalRenyi2018} replacing the densities $\omega(m'|z^*)$ by $\omega_\T(m'|z^*)$ and noting that all inequalities still hold.
\end{proof}

\section{A Wiretap Coding Scheme} \label{scheme}
In this section we will furnish a wiretap coding scheme $\wiretap$ for an arbitrary\footnote{Here arbitrary indeed means \textit{any} discrete-time wiretap channel; however, a positive secrecy rate may not be attainable on some wiretap channels.} wiretap channel which is \textit{based} on a wiretap scheme put forth in \cite{semanticallySecure}, \cite{explicitGaussianWiretap}, and \cite{UHF}. We will first define each step of this scheme and show that it is reliable (we will show security in the next section). Then we will give a particular implementation and show that this implementation is efficient with respect to the block length $n$.

Over an arbitrary wiretap channel $W = (T,E)$ our wiretap coding scheme $\wiretap$ involves combining an SS-UHF with a reliable ECC already in use over the main point-to-point channel. This modular wiretap scheme is precisely the scheme put forth in \cite{semanticallySecure,explicitGaussianWiretap,UHF} except there the UHF was only required to be $b$-regular and evenly invertible. Here, we are also demanding that our UHF be \textit{uniform}. The necessity of this extra property will be elucidated in the next section when we prove that our scheme is semantically secure. 

Consider \Cref{fig:txscheme}; this describes our wiretap scheme overall. We will now describe in detail each layer.

\begin{figure}[ht]
  \begin{center}
  	\begin{tikzpicture}[auto, thick, node distance=2cm, >=triangle 45,thick,scale=0.81, every node/.style={scale=0.81}]
\draw
	node [name=input1] {\Large $M$}
	node [right of=input1] (Mprime) {\Large $M'$}
	node [right of=Mprime] (Xn) {\Large $X^{n}$}
    node [right of=Xn] (Yn) {\Large $Y^{n}$}
    node [right of=Yn] (Mhat) {\Large $\widehat{M'}$}
    node [right of=Mhat] (Mdec) {\Large $\hat{M}$}
    node [below of=Yn] (Zn) {\Large$Z^{n}$};
	\draw[->](input1) -- node {$\phi_s$}(Mprime);
 	\draw[->](Mprime) -- node {$e_n$} (Xn);
	\draw[->](Xn) -- node {$T^{n}$} (Yn);
    \draw[->](Yn) -- node {$d_n$} (Mhat);
    \draw[->](Mhat) -- node {$f_s$} (Mdec);
	\draw[->](Xn) |- node [near end] {$E^{n}$} (Zn);
\end{tikzpicture}
  \end{center}
  \caption{Wiretap coding scheme.}
  \label{fig:txscheme}
\end{figure}
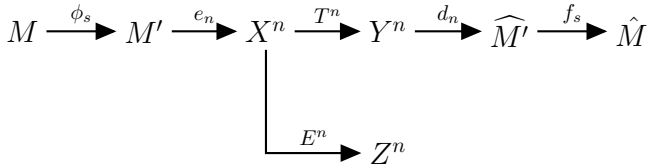

\subsection{Preprocessing Layer}
Consider the finite sets $\M=\{0,1\}^k$ and $\M' = \{0,1\}^l$ with $l > k$. We shall refer to $M \in \M$ as the \bfit{actual message} and $M' \in \M'$ as the \bfit{pseudo-message} because $M$ represents the information the transmitter \textit{actually} wishes to impart to the intended receiver securely, whereas $M'$ is some random variation of the actual message necessary for security. We will \textit{not} assume which distribution the message $M$ takes.

Over a fixed arbitrary finite set $\seed$, the transmitter will first draw a seed $S \sim \U(\seed)$ to be used for the remainder of transmission. We assume the seed is independent of the message $M$ and that the realized seed is \textit{publicly} available to all parties. All communication must take place over the wiretap channel; however, we show in \Cref{appendix:seedremoval} that the transmitter can send the seed before the transmission of an actual message with no asymptotic loss in rate or security. 

The transmitter now chooses an SS-UHF $\F = \{f_s: \Mprime \to \M \,|\,s \in \seed \}$. Suppose each function $f_s$ in the SS-UHF has its invertible stochastic mapping given by $\phi_s$. Choosing our message as $M = m$ and seed as $S = s$, we will choose our pseudo-message to be $M' = \phi_s(m)$. Since $\F$ is \textit{evenly} invertible and $b$-regular, $M'$ is a uniform random variable on $2^b$ elements of $\M'$. In particular, $\omega(m'|m,s) = 2^{-b}\indicator\left(m'\in \text{supp}(\phi_s(m)) \right)= 2^{-b}\indicator\left( f_s(m') = m \right)$.

\subsection{Coding Layer}
The transmitter chooses some reliable ECC scheme\footnote{We always assume that the scheme satisfies the power constraint for the channel if there is one.} $\code = \{\code_n\}_{n \in \N}.$ We will assume (as per standard) that each party has full knowledge of $\code$. Thus, for a given blocklength $n$, each party knows $\codebook_n$ is the codebook and we have inherently induced new channels: $T^n:\codebook_n \to \Y^n$ and $E^n:\codebook_n \to \Z^n$. We will henceforth be considering \textit{these} as the main transmission channel and eavesdropper's channel respectively for the remainder of this work. At this point the transmitter encodes the pseudo-message $M'$ using $e_n$, this will be a random variable $X^n = e_n(M')$ over $\codebook_n$. Now the transmitter sends $X^n$ over the wiretap channel $W = (T,E)$; that is, the channel input $X^n$ is sent across $T^n$ but also across $E^n$ inherently.

\subsection{Intended Receiver's Decoding Layer}
The intended receiver will receive a noisy version of the channel input $Y^n = T^n(e_n(M'))$. The goal of the intended receiver is to correctly guess which realization of the random variable $M'$ was sent given the realization of the random variable $Y^n$. This is accomplished using the estimate $\widehat{M}' = d_n(Y^n)$. Since we have assumed $\code$ to be reliable, each $\code_n$ is an ECC and thus the probability of error $\P_e(\code_n)$ is considerably low. In particular, for some finite $n$ this means there is a high probability that $\widehat{M}'$ will equal $M'$; this equality happens almost surely asymptotically with $n$. In short, the intended receiver will be able to undo the coding layer entirely. 

Next, the intended receiver shall \textit{post-process} $\widehat{M'}$ to an estimate of the actual message $\widehat{M}$ using the hash function corresponding to the public seed $S$. That is, given that $S = s$ the intended receiver's estimate is given as $\widehat{M} = f_s(\widehat{M}')$. Since we assumed our SS-UHF to be invertible, if $\widehat{M}'$ is equal to $M'$ then the UHF is guaranteed to map $\widehat{M}'$ to $M$ (the original message); however, we showed that this happens almost surely asymptotically with $n$. In this sense, the pre/post processing layers do not subtract anything from our reliability. In more detail, if $\code$ is reliable to begin with then our entire wiretap scheme will also satisfy reliability. Furthermore, if $\code$ is exceptionally reliable, then our wiretap scheme is exceptionally reliable as well. 

\subsection{Eavesdropper's Decoding Layer}
Once the eavesdropper receives her channel output $Z^n = E^n(e_n(M'))$ she will attempt to decode it in a similar fashion to that of the intended receiver; however, we will not assume \textit{how} she decodes her output since that could affect our measure of security. As a side note, in contrast to computational based security methods, we also do \textit{not} assume the boundedness of resources at the eavesdropper.

\subsection{Discussion}
As in \cite{UHF}, we call the preceding scheme \textit{modular} since the pre/post processing layers are not intrusive to the main channel in any way in terms of either reliability or constructibility. That is, our preprocessing layer could be added to any already existing communication system without changing any core components of the original system.

\subsection{Explicit Construction}\label{explicit}
Does such a wiretap scheme exist? By extensions of Shannon's channel coding theorem we know that if $R_\code < C_T$ then a reliable ECC scheme $\code$ exists. Since our wiretap scheme is a concatenation of a pre/post processing layer with a reliable ECC, we now only need to be concerned if such a pre/post processing scheme exists; in particular, if an SS-UHF exists. 

In this subsection we give an explicit construction of an SS-UHF. Our construction is inspired by those given in \cite{semanticallySecure,explicitGaussianWiretap,UHF}; however, there, the UHF's can be shown to \textit{not} satisfy uniformity which is \textit{essential} to our proof showing our wiretap scheme is semantically secure in the next section.

Consider the following family of functions \[\F^* = \{f_{s,t}: \M' \to \M \;|\; s \in \{0,1\}^l \setminus 0^l,\ t \in \{0,1\}^l \}\] where $f_{s,t}(m') =  \left[ \left( s \odot m'\right) \oplus t\right]_k$ and $\M = \{0,1\}^k$ and $\M' = \{0,1\}^l$ as before. Here, all $l$-length bit strings correspond to their respective elements in the finite field $GF(2^l,\oplus,\odot)$ (where $\oplus$ and $\odot$ denote addition and multiplication in the field respectively), $[\cdot]_k$ selects the $k$ most significant bits, and $0^l$ denotes the all-$0$ bit string of length $l$ (which is correspondent to the additive identity in $GF(2^l)$). As a remark, we note that $\oplus$ here is equivalent to modulo-2/bitwise/XOR addition and $\seed = \{0,1\}^l \setminus 0^l \times \{0,1\}^l$ where $|\seed| = (2^l - 1) 2^l$.

For some random variable $R \sim \U(\{0,1\}^{l-k})$ and $(s,t) \in \seed$ consider the inverses of $f_{s,t}$ given by \[\phi_{s,t,R}(m)= s^{-1} \odot \left((m||R) \oplus t \right).\] Here $s^{-1}$ is the inverse element of $s$ in $GF(2^l)$ (which always exists because $s$ is non-zero) and $(\cdot|| \cdot)$ represents usual bit-string \textit{concatenation}.

\begin{prop}\label{prop:SSUHF}
The family of functions $\F^*$ is an SS-UHF.
\end{prop}
\begin{proof}
See \Cref{appendix:UHFconstruction}.
\end{proof}

With this we have constructed a concrete (algorithmic) implementation of an SS-UHF: this means that our wiretap coding scheme $\wiretap$ of the previous subsection always exists. Specifically, our pre/post processing layers are given \textit{concretely} so that if the reliable ECC $\code$ is also given concretely, then so is the entire wiretap scheme. Let us emphasize again that this is in contrast to much of the literature wherein wiretap schemes are implicitly defined through proofs by existence.

The fact that our wiretap scheme is explicitly given is necessary for realistic wiretap schemes but not quite enough in terms of practicality. We would also like our scheme to be \textit{efficient} with block length $n$. Fortunately, our pre/post processing scheme \textit{is} efficient as proven in the next proposition. In other words, when the reliable ECC scheme $\code$ is efficient, so is the entire wiretap scheme.

\begin{prop}\label{prop:UHFisEfficient}\
\begin{enumerate}
\item Given $m \in \M$, $(s,t) \in \seed$, and $r \in \{0,1\}^{l-k}$, the inverse $\phi_{s,t,r}(m)$ can be computed in quadratic-time with respect to $n$.
\item Given $(s,t) \in \seed$ and $m' \in \M'$, the function $f_{s,t}(m')$ can be computed in quadratic time with respect to $n$.
\end{enumerate}
\end{prop}
\begin{proof}
See \Cref{appendix:UHFconstruction}.
\end{proof}

In conclusion of this section, we have constructed a concrete and efficient  wiretap scheme that is polynomially time computable with block length $n$. We note that the construction given here is by no means unique and one could use \textit{any} concrete and efficient SS-UHF as the pre/post processing layers of our wiretap scheme $\wiretap$.

\section{Achievable Semantically Secure Rates}\label{sec:LHL}
We have already seen that the wiretap scheme we constructed in \Cref{scheme} satisfies the reliability property of a wiretap scheme as long as the ECC $\code$ is reliable (and does so \textit{exceptionally} when $\code$ is chosen to be \textit{exceptionally} reliable). Now we need to show that the scheme satisfies the security property of a wiretap scheme as well. In this section we will do just that by constructing leakage bounds for the semantic metric. It will turn out that under certain conditions our leakage bounds asymptotically go to 0 implying that our scheme is a semantically secure wiretap scheme. In particular, under further restrictions, our wiretap scheme is shown to be \textit{outstanding}. 

It is noted that leakage bounds for arbitrary wiretap channels using evenly invertible, $b$-regular UHFs are already given in \cite{UHF}; however, the leakage there assumes the secret message $M$ follows a uniform distribution and hence will only lead to strong security at best. As a reminder, strong security is not a sufficient measure of security in real world applications because often times messages are not uniformly distributed. We therefore need to generalize the \textit{leftover hash lemma} (channel version) in \cite{UHF} to overcome this obstacle. What becomes obvious upon proof is that considering UHF's that are only evenly invertible and $b$-regular is not quite restrictive enough to lead to semantic security; this explains why in our wiretap coding scheme of \Cref{scheme} we chose our UHF to also be \textit{uniform}.

For the remainder of this section, we will write $\wiretap =\{\wiretap_n\}$ to be the modular wiretap coding scheme described in \Cref{scheme}.

\subsection{Semantic Leakage Bounds}

\begin{thm}\label{lem:LHL}
Using $\wiretap$ on any wiretap channel $W = (T,E)$, for $\epsilon \geq 0$ we have
\begin{align*}
\max_{P_M} I(M \sep Z^n) &\leq \frac{1}{\ln 2}2^{ \frac{1}{2} \left(-b + I_2^\epsilon(M' \sep Z^n) \right) } + \epsilon k\\
&\leq \frac{1}{\ln 2}2^{\frac{1}{2} \left(-b + I_\infty^\epsilon(M' \sep Z^n)\right)} + \epsilon k.
\end{align*}
\end{thm}
\begin{proof}
See \Cref{appendix:LHL}.
\end{proof}

Note the striking resemblance of our first inequality to \cite[Theorem 3]{privacyamplification} for secret key agreements. There, they also used universal hashing to amplify privacy. Also note that our bounds generalize those provided in \cite{UHF}. Therein, the message was assumed to be uniform, whereas here we make no a priori assumptions on $P_M$. Admittedly, we require an SS-UHF for the pre/post processors whereas they require an SS-UHF less the uniform requirement for the pre/post processors. However, we have provided in \Cref{explicit} an efficient and concrete construction, thereby alleviating any doubts that such a hash family exists. Lastly, we note that an \textit{attempt} to generalize the bounds of \cite{UHF} to the ones given here has already been given in literature but was redacted due to an error. Our approach is noticeably different allowing our proof to overcome said error.

We will only be concerned with the second inequality of \Cref{lem:LHL} for the remainder of this paper. It is considerably more tractable computationally and has already been studied in \cite{UHF}.

Recall the wiretap scheme $\wiretap$ consists of a pre/post processor and an ECC scheme $\code$. \Cref{lem:LHL} makes no a priori assumptions on what that ECC scheme is. Once we actually pick the ECC, however, we can characterize \Cref{lem:LHL} more appropriately. In particular, suppose we choose a reliable ECC $\code = \{\code_n\}$ with each $\code_n$ having rate $R_{\code_n} = \frac{l}{n}$ and with the overall rate of each $\wiretap_n$ given by $R_n = \frac{k}{n}$. Now since the ECC has been chosen, there is a \textit{deterministic bijective} mapping between $\M'$ and $\codebook_n = e_n(\M')$. Thus if $M'$ is a random variable on $\M'$ then $e_n(M')$ is a random variable on $\codebook_n$ with the same distribution as $M'$. For convenience, define the random variable $X^n = e_n(M')$ and note that it is defined only over $\codebook_n$ \textit{not} $\X^n$. With these observations in mind, we can reformulate \Cref{lem:LHL} as follows.

\begin{cor}\label{lem:LHL2}
Using $\wiretap$ with reliable deterministic ECC $\code$ on any wiretap channel $W = (T,E)$, for $\epsilon \geq 0$ we have
\[\max_{P_M} I(M \sep Z^n) \leq \frac{1}{\ln 2}2^{- \frac{n}{2} \left(R_{\code_n} - R_n - \frac{1}{n}I_\infty^\epsilon(X^n \sep Z^n)\right)} + \epsilon n R_n.\]
\end{cor}

\subsection{Semantically Secure Rates}
With the previous two bounds on the semantic leakage in mind, we would like to know under what conditions they asymptotically (with respect to $n$) approach 0. In this way, those conditions will tell us precisely when our wiretap coding scheme $\wiretap$ is semantically secure. It is fortunate that these conditions can be described in terms of $\Rs$ (the asymptotic achievable secrecy rate), $R_\code$ (the rate of the ECC scheme), and $\frac{1}{n} I_\infty^\epsilon(X^n \sep Z^n)$ ($\epsilon$-smooth max-information per channel symbol). 

Let $(\cdot)^+ = \max(\,\cdot\,,0)$. The following theorem characterizes which secure rates are achievable under semantic security and we will be using its conclusions frequently throughout the rest of this paper.

\begin{thm} \label{thm:LHL3} 
Using $\wiretap$ with a reliable deterministic ECC $\code$ on any wiretap channel $W = (T,E)$, if $\epsilon$ is chosen such that $\epsilon n \to 0$ as $n \to \infty$ then we have the following.
\begin{enumerate}[(1)]
\item \[\Rs < \left( R_\code - \lim\limits_{n \to \infty} \frac{1}{n} \Iemax \right)^+\] using \textit{semantic} security.
\item If $\lim_{n \to \infty} \frac{1}{n}\Iemax \leq \xi$ then \[\Rs < \big( R_\code - \xi \big)^+\] using \textit{semantic} security.
\item If $\epsilon$ is \textit{exponentially} diminishing to 0 with $n$, then for any secure rates as in (1) and (2), $\wiretap$ is exceptionally semantically secure. Moreover, if $\code$ is exceptionally reliable, then $\wiretap$ is an outstanding semantically secure \textit{wiretap} scheme.
\end{enumerate}
\end{thm}
\begin{proof}
See \Cref{appendix:LHL}.
\end{proof}
\begin{remark}
In \Cref{section:fading}, we will apply this theorem to channels with side information (extra information available to Alice that may help her deduce better security or reliability, of which fading channels are a special case). In that case, we will restate this theorem in a more suitable form (see \Cref{cor:Thm1Redux}).
\end{remark}

Given \textit{any} wiretap channel, \Cref{thm:LHL3}.1 says all one needs to do is calculate $\lim_{n \to \infty} \frac{1}{n} \Iemax$, then use of the wiretap scheme $\wiretap$ will guarantee that rates $R_s$ are achievable with semantic security. However, finding this limit is probably not feasible. For fixed $n$, the $\epsilon$-smooth max-information is basically an $n$-dimensional integral where each point of the integral is a maximization over a set with roughly $2^{n}$ elements. This problem is exponentially hard unless one can exactly characterize the ``regions'' of the integrand that have the same maximum. Characterization of these regions is an interesting line of future work but we do not explore that any further here. 

Luckily, we do not need to calculate $\lim_{n \to \infty} \frac{1}{n} \Iemax$ exactly. \Cref{thm:LHL3}.2 says an upper bound to this limit suffices. We will primarily be using this result for the remainder of this paper due to its tractability. In forthcoming sections we will see that this still yields surprisingly favorable results. 

The leakage bound, $\xi$, in \Cref{thm:LHL3}.2 can be thought of as a parameter of the eavesdropper's channel. Moreover, it can be thought of as the loss we incur when converting an ECC into a semantically secure wiretap code by our procedure. That is, given an ECC of rate $R_\code$, our procedure converts that ECC into a semantically secure wiretap code of rate $R_\code - \xi$.

\Cref{thm:LHL3}.3 says that in order to control the speed by which the semantic leakage diminishes with $n$, we only need to control the speed by which $\epsilon$ diminishes with $n$ where we recall that $\epsilon$ is a parameter that controls how much of the total space (with respect to probability) we are considering. We note that when $\epsilon = 0$ we are always considering the entire space for every $n$ so that the condition of \Cref{thm:LHL3}.3 is trivially satisfied and we have exceptional semantic security. We will not pursue such an approach in this paper as the $\epsilon > 0$ case is much more manageable. However, in all of our applications, $\epsilon$ will be exponentially diminishing with $n$ so that we will get exceptional semantic security.

Recall that $C_T$ is the point-to-point channel capacity of Alice and Bob's channel and $C_E$ is the point-to-point channel capacity of Alice and Eve's channel. The following is a special case of \Cref{thm:LHL3}.2. 
\ \\ 
\begin{cor} \label{cor:LHL4}
If $\xi = C_E$, $\epsilon n \to 0$ as $n \to \infty$, and we pick a reliable deterministic ECC $\code$ with rate arbitrarily close to $C_T$ then \[\Rs < \left( C_T - C_E \right)^+\] with \textit{semantic} security.
\end{cor}
This corollary is particularly satisfying considering that many channels have their weak secrecy capacity given by $C_T - C_E$. Thus in those cases, if we can satisfy the conditions of \Cref{cor:LHL4}, we can achieve the secrecy capacity using our wiretap scheme $\wiretap$ and moreover, we immediately have proven that the semantic secrecy capacity is \textit{equivalent} to the weak secrecy capacity by using \Cref{fact:CsEqual}, a result not known in general.

\subsection{Summary of Wiretap Coding Scheme}
Let us end this section by summarizing what we have shown for our wiretap scheme so far and explain how this can be applied in practice and in theory.

Our wiretap scheme outlined precisely in \Cref{scheme} is a combination of a pre/post processor based on an SS-UHF together with a reliable ECC scheme. We constructed a concrete and efficient SS-UHF in \Cref{scheme} and showed that it did not affect the reliability of the ECC scheme. Hence, since we always assume the ECC scheme is chosen to be reliable, our entire wiretap scheme is always reliable. Moreover, when the ECC is exceptionally reliable the entire scheme is also exceptionally reliable.

In this section, we showed that over a truly arbitrary wiretap channel, our wiretap scheme's semantic leakage can be bounded using \Cref{lem:LHL} or \Cref{lem:LHL2}. Moreover, if the threshold probability $\epsilon$ of our space (a parameter solely designed to aide in the proof) is chosen so that $\epsilon n \to 0$ as $n \to \infty$ then \Cref{thm:LHL3} gives us a precise characterization of when our wiretap coding scheme is semantically secure over \textit{any} wiretap channel. 

To this end, we find it beneficial to outline the steps one shall take in applying our wiretap scheme to a wiretap channel of their choice.

\begin{proc}\label{procedure} The following is the general procedure one should take when using our wiretap scheme over an arbitrary wiretap channel $W = (T, E)$.
\begin{enumerate}[1.]
\item Find which achievable rates $\Rs$ are supported on $W$.
\begin{itemize}
\item For each $n$, construct a $(1-\epsilon)$-typical set $\T$ where $\epsilon n \to 0$ as $n\to \infty$.
\begin{itemize}
\item \bfit{Preferably} construct $\T$ so that $\epsilon$ is exponentially diminishing to 0 so that we induce \textit{exceptional} semantic security.
\end{itemize}
\item Find an upper bound $\xi$ such that \[\lim_{n \to \infty} \frac{1}{n} \Iemax \leq \xi.\]
\begin{itemize}
\item Ideally, one should find the smallest possible $\xi$ as to guarantee higher achievable rates.
\item One could also compute $\lim_{n \to \infty} \frac{1}{n} \Iemax$ directly as mentioned previously, but currently this is seemingly intractable.
\end{itemize}
\end{itemize}

\item Choose operating point $\Rs$.
\begin{itemize}
\item We can achieve all rates $\Rs < \big( R_\code - \xi \big)^+$ with semantic security (\Cref{thm:LHL3}.2).
\begin{itemize}
\item We must choose $R_\code > \xi$ in order to have positive secrecy rates using our wiretap scheme over $W$. However, if this is not possible then either $\xi$ was chosen poorly or the channel does not allow a positive semantic secrecy rate with our wiretap scheme.
\end{itemize}
\end{itemize}

\item Build the wiretap coding scheme $\wiretap$.
\begin{itemize}
\item Find a reliable ECC scheme $\code$ of rate $R_\code$ for use over the main point-to-point channel.
\begin{itemize}
\item \bfit{Preferably} choose $\code$ as follows:
\begin{itemize}
\item Concrete, so that the entire wiretap scheme is concrete.
\item Efficient, so that the entire wiretap scheme is efficient.
\item Exceptionally reliable, so that the entire wiretap scheme is exceptionally reliable. 
\end{itemize}
\end{itemize}
\item Use the finite field SS-UHF of \Cref{prop:SSUHF} as the pre/post processor of this wiretap scheme.
\begin{itemize}
\item One could use any SS-UHF in practice but it is preferable to use one like ours that is concrete and efficient.
\end{itemize}
\end{itemize}
\end{enumerate}
\end{proc}
\begin{remark}
Note that if $\epsilon$ is exponentially diminishing to 0 and $\code$ is chosen exceptionally reliable then $\wiretap$ is an \bfit{outstanding wiretap scheme}.
\end{remark}

\section{Applications I}\label{sec:App1}
In this section, we show how to actually use \Cref{procedure}. We apply \Cref{procedure} to both the discrete memoryless wiretap channel (DMWC) and the memoryless additive white Gaussian noise wiretap channel (AWGN). In particular, on the AWGN and symmetric, degraded DMWCs, we achieve the semantic secrecy capacity. Lastly, we explain how our scheme can be applied in theory in the finite regime; i.e.\ we explain results for finite $n$.

Before we begin, we will write max-information in a more convenient form. This is both so that we can use the supporting results of \cite{UHF}, but also because this alternative form will have a better interpretation here.

\begin{lemma}\label{lem:awgnDMCmaxinfo}
The $\epsilon$-smooth max-information $\Iemax$ can alternatively be written as the infimum of
\[ \log \int\limits_{\Z^n} \max\limits_{x^n \in \codebook_n} \omega_\T (z^n|x^n) \mu(dz^n)\] over all $(1-\epsilon)$-typical sets $\T$.
\end{lemma}
\begin{proof}
See \Cref{appendix:awgnMaxInfoBound}.
\end{proof}

\subsection{Semantic security on a DMWC}
For our first application of \Cref{thm:LHL3} and \Cref{procedure}, we consider DMWCs. This is the case when both the intended receiver's channel and the eavesdropper's channel are given by distinct point-to-point discrete memoryless channels (DMC). We represent the input signal by the discrete random variable $X$, Bob's output signal by the discrete random variable $Y$, and Eve's output signal by the discrete random variable $Z$.

\begin{fact} (cf. \cite{cover_thomas}) The point-to-point capacity of a DMC with input $X$ and output $Y$ is given as
\[C = \max_{P_X} I(X \sep Y).\]
\end{fact}
In particular we denote $C_T = \max_{P_X} I(X \sep Y)$ and $C_E = \max_{P_X} I(X \sep Z)$. As described in \Cref{procedure}, in order to characterize a set of semantically secure rates, we need to asymptotically bound the max-information per channel symbol of Eve's channel. The following lemma provides this bound.

\begin{lemma}\label{lem:DMCBound}
Using a reliable ECC scheme $\code$, the max-information per channel symbol of the DMC $E$ is asymptotically bounded as
\[\lim\limits_{n \to \infty} \frac{\Iemax}{n} \leq C_E,\] where $\epsilon$ is exponentially decreasing to 0 with $n$.
\end{lemma}
\begin{proof}
Using \Cref{lem:awgnDMCmaxinfo} (where $\mu$ is the counting measure) we can write
\[\frac{1}{n} I_\infty^\epsilon(X^n \sep Z^n) \leq \frac{1}{n}  \log \sum\limits_{z^n} \max\limits_{x^n \in \codebook_n} \omega_\T (z^n|x^n),\] for any $(1-\epsilon)$-typical set $\T$.

Luckily, \cite[Lemma 5]{UHF} proved a bound on this right hand term for the same modular pre/post processing scheme less our uniform requirement. Thus, by their result we immediately have that there exists a constant $c >0$ such that for $\epsilon = e^{-nc}$:
\[\frac{1}{n} I_\infty^\epsilon(X^n \sep Z^n) \leq C_E + \frac{1}{n}o(n)\]
where $\frac{1}{n}o(n)$ is a term diminishing to $0$ as $n \rightarrow \infty$. This completes the claim asymptotically with $n$.
\end{proof}

With this bound we can apply \Cref{thm:LHL3}.2 immediately to describe the semantically secure rates our wiretap scheme can achieve. However, in certain cases we can achieve the secrecy capacity (with semantic security), i.e.\ the best possible semantically secure rate. In order to describe this, let us recall the following fact.

\begin{fact}\label{fact:vanDijk}\cite{vanDijk}
The secrecy capacity of a DMWC where the eavesdropper's channel is noisier than the main channel and both channels are weakly symmetric is given by 
\[\Cs = C_T - C_E.\]  
\end{fact}

With this fact, we can state our main result of this subsection, a characterization of semantically secure achievable rates for the DMWC. Note that this result was already proven in \cite{semanticallySecure} and \cite{channelupgrading}, but we restate this here to show the efficacy of our proposed wiretap coding scheme and the fact that our proof differs significantly.

\begin{thm}\ \label{thm:DMWC}
\begin{enumerate}
\item On any DMWC, our wiretap scheme $\wiretap$ can achieve all secure rates, \[\Rs < (R_\code - C_E)^+,\] with exceptional\footnote{Recall that exceptional here means that the semantic leakage diminishes to 0 exponentially fast with $n$.} semantic security.
\item On a DMWC where both channels are weakly symmetric and the eavesdropper's channel is noisier than the main channel we can achieve the secrecy capacity under exceptional semantic security when $R_\code$ achieves the main channel capacity $C_T$.
\end{enumerate}
\end{thm}
\begin{proof}
This follows from \Cref{thm:LHL3}.2 and \Cref{cor:LHL4} combined with \Cref{lem:DMCBound} and \Cref{fact:vanDijk}.
\end{proof}

The first part of this proposition emphasizes that our wiretap scheme $\wiretap$ acts as a converter. If we input an ECC scheme for the DMC $(\X,\omega(y|x),\Y)$ of rate $\Rc > C_E$, then our procedure converts that ECC scheme into an exceptionally semantically secure wiretap code for a DMWC of rate $\Rs$. 

The second part of this proposition says that on degraded symmetric DMWCs, our conversion respects the optimality of rates. Specifically, it says that given an optimal ECC scheme, i.e.\ an ECC scheme achieving the point-to-point main channel capacity, our procedure converts this ECC scheme into an exceptionally semantically secure wiretap code of optimal rate, i.e.\ a wiretap scheme achieving the secrecy capacity.

With this, we again emphasize that our conversion is concrete and efficient. Thus, if the ECC scheme is such, so is the entire wiretap scheme. Moreover, if the ECC scheme is exceptionally reliable, the wiretap scheme is \textit{outstanding}\footnote{Recall the definition of an \textit{outstanding} wiretap coding scheme from \Cref{sect:wiretapcodes}.}.

\subsection{Semantic security for AWGN wiretap channels}
We consider now the additive white Gaussian noise (AWGN) memoryless wiretap channel where both the intended receiver's channel and eavesdropper's channel are given by distinct AWGN memoryless channels. We represent the input signal by the real random variable $X$ (where we suppose it satisfies the average power constraint $P$) and the additive white Gaussian noise by the real random variable $U$. The channels $T$ and $E$ can be described by their outputs given respectively as
\begin{align*}
Y &= X + U_T\\
Z &= X + U_E.
\end{align*}
The random variables $U_T$ and $U_E$ are assumed mutually independent and sampled i.i.d.\ according to $\mathcal{N}(0,\sigma_T^2)$ and $\mathcal{N}(0,\sigma_E^2)$ respectively. 

\begin{fact}\label{gaussianCapacity} (cf. \cite{cover_thomas})
The capacity of an AWGN channel with average input power constraint $P$ and additive noise variance $\sigma^2$ is given by \[C = \frac{1}{2}\log\left( 1+ \frac{P}{\sigma^2}\right).\]
\end{fact}
In particular, this means the capacity of the intended receiver's point-to-point channel is given by $C_T = \frac{1}{2}\log( 1 + \frac{P}{\sigma_T^2})$ and the capacity of the eavesdropper's point-to-point channel is given by $C_E = \frac{1}{2}\log ( 1 + \frac{P}{\sigma_E^2})$.

Our goal of this subsection is to describe the semantically secure achievable rates that our wiretap scheme $\wiretap$ can achieve. Using \Cref{procedure} we already have a prescription of how to do this by bounding the max-information per channel symbol.

\begin{lemma}\label{lem:AWGNbound} Using a reliable ECC scheme $\code$, the max-information per channel symbol of an AWGN eavesdropper channel $E$ is asymptotically bounded as
\[\lim\limits_{n \to \infty} \frac{\Iemax}{n} \leq C_E,\]
where $\epsilon$ is exponentially decreasing to 0 with $n$.
\end{lemma}
\begin{proof}
Using \Cref{lem:awgnDMCmaxinfo} (where $\mu$ is the Lebesgue measure) we can write
\[\frac{1}{n} I_\infty^\epsilon(X^n \sep Z^n) \leq \frac{1}{n}  \log \int\limits_{\R^n} \max\limits_{x^n \in \codebook_n} \omega_\T (z^n|x^n) dz^n,\] for any $(1-\epsilon)$-typical set $\T$.

Again, \cite[Lemma 6]{UHF} proved a bound on this right hand term for the same modular pre/post processing scheme less our uniform requirement. Thus, by their result we immediately have the following bound for every $\delta > 0$ small:
\[\frac{1}{n} I_\infty^\epsilon(X^n \sep Z^n) \leq C_E + \delta \log e + \frac{1}{n}o(n).\] Here $\epsilon = \exp(-n\delta^n/8)$ and $\frac{1}{n}o(n)$ is a term diminishing to 0 as $n \to \infty$.

Since this holds for every $\delta >0$ this completes the claim asymptotically with $n$.
\end{proof}
\begin{remark}
A reworked proof of \cite[Lemma 6]{UHF} can be found in our \Cref{appendix:awgnMaxInfoBound} (\Cref{lem:awgnmaxinfobound}). We feel it is worthwhile to see the proof of this statement for the AWGN wiretap channel, since later (specifically in \Cref{thm:ImaxLessCE} and \Cref{thm:ImaxPartialCSIT}), we prove a more complicated analogous result for the No-CSIT and partial CSIT wiretap channels.
\end{remark}

Again, now that we have this bound in hand, we can apply \Cref{thm:LHL3}.2 to describe the semantically secure rates our wiretap scheme can achieve. However, we notice that we can actually achieve the best possible rates after considering the following fact.

\begin{fact}\label{fact:AWGN}\cite{AWGNwiretap}
On an AWGN wiretap channel $W = (T,E)$, the weak secrecy capacity is given as:
\begin{align*}
\Cs = \begin{cases}
C_T-C_E, &\text{if }\sigma_T^2 < \sigma_E^2\\
0, &\text{Otherwise.}
\end{cases}
\end{align*}
\end{fact}
\begin{remark}
This fact can be upgraded to strong secrecy using the usual technique (cf. \cite{bloch_barros_2011}). Only recently was this fact upgraded to semantic secrecy \cite{AWGNsemanticRecent}.
\end{remark}

Using this fact, we have the following main result of this subsection. 
\ \\ \ \\ \ \\ 
\begin{thm}\label{thm:awgnrates}\ 
\begin{enumerate}
    \item On an AWGN wiretap channel, our wiretap scheme $\wiretap$ can achieve all secure rates \[\Rs < R_\code - C_E\] with exceptional semantic security as long as $R_\code > C_E$.
    \item In particular, when $R_\code$ achieves the main channel capacity $C_T$, then our wiretap scheme achieves the secrecy capacity under exceptional semantic security.
\end{enumerate}
\end{thm}
\begin{proof}
This follows from \Cref{thm:LHL3}.2 and \Cref{cor:LHL4} combined with \Cref{lem:AWGNbound} and \Cref{fact:AWGN}.
\end{proof}
\begin{remark}
In an independent way from \cite{AWGNsemanticRecent}, \Cref{thm:awgnrates}.2 shows that the semantic secrecy capacity is equivalent to the weak secrecy capacity for the AWGN wiretap channel using \Cref{fact:CsEqual}.
\end{remark}

Note that $\wiretap$ is exceptionally semantically secure so that if $\code$ is also chosen to be exceptionally reliable, then our entire wiretap coding scheme is \textit{outstanding}\footnote{Again recall the definition of an \textit{outstanding} wiretap coding scheme from \Cref{sect:wiretapcodes}.}.

Indeed an ECC scheme is given in \cite{capacityAchAWGNECC} that is concrete, reliable, and has quadratic time complexity with respect to block length $n$ in both encoding and decoding. Moreover, it has probability of error exponentially decreasing to 0 so that it is exceptionally reliable. Thus using this particular ECC scheme with our SS-UHF implementation given in \Cref{prop:SSUHF} gives an end-to-end wiretap coding scheme for the AWGN wiretap channel that is \textit{concrete, efficient, outstanding, semantically secure, and can achieve the secrecy capacity}.

Note that the wiretap scheme used in \cite{AWGNsemanticRecent} has every single one of these properties as well. However, their wiretap coding scheme is based on polar lattices and is not modular. In contrast, our scheme \textit{is} modular: the exact same pre/post processor used here (that is, the SS-UHF of \Cref{prop:SSUHF}) can be used on any channel (discrete or continuous); one just needs to find a reliable ECC scheme for the given point-to-point channel.

\subsection{Finite Analysis}
Thus far we have exclusively focused on asymptotic analysis of our wiretap scheme. Despite this, \Cref{lem:LHL2} gives an extremely useful bound of security and rates in the finite regime, that is, for a fixed finite coding blocklength $n$. We do not pursue this line any further here, but for an interesting look into finite block length analysis see Yang, Schaefer, and Poor's result \cite{Poor_finite} which also uses a UHF based scheme to derive upper and lower bounds on the achievable rates in the finite regime. 

\section{Applications II - Fading}\label{section:fading}
In this section, we will consider even more applications of \Cref{thm:LHL3} and \Cref{procedure}, specifically, applications to \textit{fading} wiretap channels. Fading wiretap channels are the prototypical physical layer security models of wireless communication.

It is standard to assume some feedback of channel state information to Alice that will help her deduce the current fade and increase her overall secure transmission rate. In this sense it is obvious that fading wiretap channels are only a particular instance of a much more general case of wiretap channels: wiretap channels with side information. Side information is any information in the form of a random variable available to Alice before transmission that may be advantageous. In this way, it may help her induce more reliability or security, which in turn may help her ascertain a higher secure achievable rate. Hence, by studying wiretap channels with side information, we are inherently considering fading wiretap channels by inclusion.

To study wiretap channels with side information we will first need to manipulate the language we have introduced thus far. Let $\Lambda^n \in \H^n$ represent the $n$ pieces of \bfit{side information} that may be advantageous to the transmitter. Because we always deal in the worst case for security, it is necessary to assume that the eavesdropper also knows $\Lambda^n$, thus we will need to convert the previously defined security metrics in the obvious way to account for this. However, as is a common trick in fading, we can consider the entire tuple $(Z^n, \Lambda^n)$ to be the eavesdropper's output instead of only $Z^n$ as before. Thus, for wiretap channels with side information, the semantic security metric has its leakage given by $\max_{P_M} I(M \sep Z^n, \Lambda^n)$.

With this trick, we can also consider our main result, \Cref{thm:LHL3}, redone for side information, however, we will only need part 2 and part 3 of that theorem.
\ \\
\begin{cor}[\Cref{thm:LHL3} redux]\label{cor:Thm1Redux}
Using $\wiretap$ with a reliable deterministic ECC $\code$ on any wiretap channel $W = (T,E)$, if $\epsilon$ is chosen such that $\epsilon n \to 0$ as $n \to \infty$ we have the following.
\begin{itemize}
\item If $\lim_{n \to \infty} \frac{1}{n}I_\infty^\epsilon (X^n \sep Z^n,\Lambda^n)\leq \xi$ then\footnote{As a reminder, $(\cdot)^+ = \max(\,\cdot\,,0)$.} \[\Rs < \big( R_\code - \xi \big)^+\] using semantic security.
\item If $\epsilon$ is \textit{exponentially} diminishing to 0 with $n$, then for any secure rates above, $\wiretap$ is exceptionally semantically secure. In particular, if $\code$ is exceptionally reliable, then $\wiretap$ is an outstanding semantically secure \textit{wiretap} scheme.
\end{itemize}
\end{cor}
\vspace*{0.15cm}
\begin{remark}\
\begin{itemize}
\item We call this a corollary due to the numerous references hereafter; however, it is in itself \textit{just} \Cref{thm:LHL3} in the case where side information is present.
\item Recall that $I_\infty^\epsilon(X^n \sep Z^n,\Lambda^n)$ is defined as the infimum of $I_\infty^\T(X^n \sep Z^n, \Lambda^n)$ over all $(1-\epsilon)$-typical sets $\T$. To be precise, we note that now $\T \subset \codebook_n \times \H^n \times \Z^n$. 
\end{itemize}
\end{remark}

It will be beneficial in the sequel to characterize $I_\infty^\T(X^n \sep Z^n,\Lambda^n)$ in the following way.

\begin{lemma}\label{lem:maxinfoFading} Let $X^n$ be a random variable over $\codebook_n$ and $\Lambda^n$ be some side information. If $X^n \bot \Lambda^n$ then
\begin{align*}
&I_\infty^\T (X^n \sep Z^n,\Lambda^n)\\
&\hspace*{1cm}= \log \left( \E_{\Lambda^n} \int_{\Z^n} \max_{x^n \in \codebook_n} \omega_\T(z^n|\Lambda^n,x^n) \mu(dz^n) \right).
\end{align*}
\end{lemma}
\begin{proof}
See \Cref{appendix:fading}.
\end{proof}
\begin{remark}
Indeed $X^n \bot \Lambda^n$ seems to be a restrictive assumption, however, it is not, as the forthcoming proofs will make clear.
\end{remark}

Until this point, we have been general with respect to side information. We really do allow \textit{any} extra information available to the transmitter that could be used to aide in a higher secure rate. However, we will now be focusing on fading wiretap channels, that is, when side information is a tuple of fading coefficients. 

\subsection{Fading Preliminaries}\label{sub:FadingDescription}
The general channel model used to model wireless communication environments is that of the \textit{fading channel}, where the output signal is an attenuation of the input signal layered with additive white Gaussian noise. The attenuation, input, and noise are represented using the complex random variables $H$, $X$, and $U$ respectively. The output of this channel at time $i$ is then given as \[Y_i = H_iX_i + U_i\] where $X_i \in \C$, $H_i \in \C$, and $U_i \sim \mathcal{CN}(0,\sigma^2) $. Here, $\mathcal{CN}(0,\sigma^2)$ is a circularly-symmetric normal distribution with $0$ mean and variance $\sigma^2$. We shall refer to the random variable representing attenuation, $H$, as the \bfit{channel coefficient}.

For the purposes of this paper, we will only be considering fast fading channels, that is, the fading coefficient is sampled i.i.d.\ for each use of the channel (cf. \cite{tse}). In particular, we will consider the case of fast fading wiretap channels, i.e., channels $T$ and $E$ are both taken to be fast fading channels. More specifically, during the $i$-th symbol of the codeword, the output at Bob from channel $T$ and the output at Eve from channel $E$ are given respectively by
\begin{align*}
Y_i &= H_{T,i}X_i + U_{T,i}\\
Z_i &= H_{E,i}X_i + U_{E,i},
\end{align*}
where $U_{T,i}$ and $U_{E,i}$ are i.i.d.\ $ \mathcal{CN}(0,\sigma_T^2)$ and $ \mathcal{CN}(0,\sigma_E^2)$ additive noise respectively, $X_i \in \C$ is subject to the power constraint $\E\left[|X|^2 \right] \leq P'$, and the coefficients $H_{T,i},H_{E,i}\in \C$ are also i.i.d. and $H_{T,i} \bot H_{E,j}$ for all $i,j$. For technical reasons we assume that the second order moment of $|H_E|$ exists, i.e., $\E[|H_E|^2] < \infty$. We note that this is not a very limiting constraint since it can be interpreted as the channel having an attenuation with finite energy. Apart from this, we do not assume \textit{which} distribution the channel coefficients follow so as to remain as general as possible. Note that this is in contrast to much of the fast fading literature that a priori assumes a distribution on both $H_T$ and $H_E$.

Achievability results for fading channels depend on which parties have instantaneous access to the realizations of $H_{T,i}$ and $H_{E,i}$, or rather, which parties have \textit{full channel state information}. If a party only has access to the statistics of $H_{T,i}$ or $H_{E,i}$ we say that party has \textit{no channel state information}.

\begin{fact}\label{fact:ComplexFadingtoRealFading}
On a complex fast fading channel, if the receiver has full channel state information (CSIR) then the channel can be decomposed into two real parallel channels.
\end{fact}
\begin{proof}
See \Cref{appendix:fading} for the usual proof.
\end{proof}

For the remainder of this paper, we will assume both the intended receiver and the eavesdropper have full channel state information (CSIR) about their respective channels. In particular, this means that we will only be considering the \textit{real} fast fading channels given at time $i$ as $Y_i = |H_{T,i}| X_i + U_{T,i}$ and $Z_i = |H_{E,i}| X_i + U_{E,i}$ due to \Cref{fact:ComplexFadingtoRealFading}. Since carrying around the modulus on the channel coefficients is cumbersome, we shall simply write $H_T$ and $H_E$ for the remainder of the paper where it will be clear that both are \textit{non-negative real} random variables instead of complex as previously mentioned. An illustration of our setup is given in \Cref{fig:fadingch}.

\begin{figure}[h]
  \begin{center}
  	\scalebox{0.75}{\begin{tikzpicture}[auto,scale=0.70, thick, node distance=2cm, >=triangle 45]
\draw
	node [name=Xn] {\Large $X_i$};
    
	\draw[->](Xn) -- ++(3cm,1.25cm) node [sum,right] (cross1) {\crossp};
    \draw[->](cross1) -- ++(1.75cm,0) node [sum, right] (sum1) {\suma};
    \draw[->](sum1) -- ++(3cm,0) node [right] (Y) {\Large$Y_i$};
    \draw[<-](cross1) -- ++(0, 1.25cm) node [above] (Hm) {\Large$H_{T,i}$};
    \draw[<-](sum1) -- ++(0, 1.25cm) node [above] (Um) {\Large$U_{T,i}$};
    
    \draw[->](Xn) -- ++(3cm,-1.25cm) node [sum,right] (cross2) {\crossp};
    \draw[->](cross2) -- ++(1.75cm,0) node [sum, right] (sum2) {\suma};
    \draw[->](sum2) -- ++(3cm,0) node [right] (Z) {\Large$Z_i$};
    \draw[<-](cross2) -- ++(0, -1.25cm) node [below] (He) {\Large$H_{E,i}$};
    \draw[<-](sum2) -- ++(0, -1.25cm) node [below] (Ue) {\Large$U_{E,i}$};
    
\end{tikzpicture}}
  \end{center}
  \caption{Fast fading wiretap channel model.}
  \label{fig:fadingch}
\end{figure}
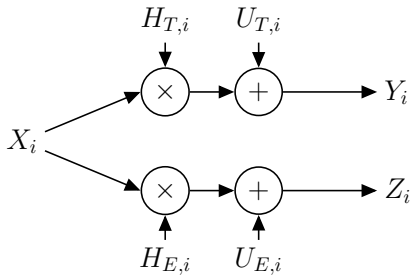

Thus far, we have made no assumptions as to what information the \textit{transmitter} has about the channel coefficients $H_T$ and $H_E$. We shall notate full channel state information at the transmitter by CSIT and will focus on three separate cases. The first case we will consider is \bfit{No-CSIT} where the transmitter has knowledge only of the main channel and eavesdropper channel statistics. Next we will consider \bfit{partial CSIT}, where the transmitter has instantaneous knowledge of the main channel's realizations of $H_T$ at each time $i$ but no knowledge of the eavesdropper's instantaneous channel coefficient - only its statistics. Finally, we will consider \bfit{full CSIT}, where the transmitter has instantaneous knowledge of both the main channel's and eavesdropper channel's realizations of $H_T$ and $H_E$ respectively.

For each of these scenarios, we wish to characterize a set of achievable secure rates. To do so, we utilize \Cref{cor:Thm1Redux} with \Cref{lem:maxinfoFading} where we take the side information to be $\Lambda^n=(H_T^n,H_E^n)$.

\subsection{Fading: No-CSIT}\label{sub:NoCSIT}
The case of No-CSIT, where the transmitter knows only the channel statistics of both the main and eavesdropper channels, is arguably the most realistic scenario of a modern wireless communication environment. It requires no special real-time feedback implementation for the main channel and assumes that the eavesdropper is purely a malicious party (although still passive). Under this assumption, in this subsection we give a set of \textit{semantically secure} achievable rates for the fast fading wiretap channel. To the best of the authors' knowledge, this is the first time semantic security has been characterized on the fast fading wiretap channel with No-CSIT in general. To do so, we find an asymptotic upper bound, $\xi$, to the leakage max-information per channel symbol, i.e., $\frac{1}{n} \Iemaxf$, for any choice of code so as to use \Cref{cor:Thm1Redux} and \Cref{procedure}. In particular, we will be focused on $\xi = C_E$, where $C_E$ denotes the point-to-point channel capacity of the eavesdropper's channel.

We start by first simplifying the expression for max-information of \Cref{lem:maxinfoFading} in the case of No-CSIT.

\begin{lemma}\label{lem:noCSIT}
On the No-CSIT real fast fading channel, max-information can be simplified as
\begin{align*}
&I_\infty^\T(X^n \sep Z^n, H_T^n,H_E^n)\\
&\hspace*{1cm}= \log \left(\E_{H_E^n}  \int\limits_{\R^n} \max_{x^n \in \codebook_n}\omega_{\T}(z^n|x^n,H_E^n) dz^n \right),
\end{align*}
 where $X^n$ is a random variable over $\codebook_n$.
\begin{proof}
See \Cref{appendix:fading}.
\end{proof}
\end{lemma}

With codeword power constraint $P$ and noise variance $\sigma^2$, we note the following fact. 

\begin{fact} \cite{tse} The point-to-point capacity of a real fast fading channel with No-CSIT is given by \[C = \frac{1}{2} \E_H\left[\log\left(1+H^2 \frac{P}{\sigma^2}\right)\right],\]
where $H$ is the random variable representing the channel coefficient.
\end{fact}
To this end, our goal for the remainder of this section will be to show \[\lim\limits_{\substack{n \to \infty\\ \epsilon \to 0}} \frac{\Iemaxf}{n} \leq  \frac{1}{2}\E_{H_E}\left[\log(1+H_E^2 \SNR)\right]\] such that $\SNR$ denotes the eavesdropper's average signal to noise ratio $P/\sigma_E^2$. In particular, we need to show the above holds for some $(1-\epsilon)$-typical set $\T$ such that $\epsilon$ is \textit{exponentially} decreasing to $0$ as $n \to \infty$. 

We begin by constructing such a set $\T$ and showing that it is typical in an exponential fashion with respect to $n$. The set is made up of three constituent sets; one each concerning the output power, noise power, and eavesdropper channel coefficient power.

We define\footnote{Motivation for defining these typical sets is based on a sphere packing argument and can be found in \Cref{appendix:SpherePacking}.} the following sets for $\delta_n,\delta_n',\delta_n'' > 0$ small: 
\begin{itemize}
\item $\Pout_n$ as the set of tuples $(h_E^n,z^n) \in \Rpn \times \Rn$ such that \[\frac{1}{n} \sum\limits_{i=1}^n \frac{z_i^2}{\sigma_E^2 + h_{E,i}^2 P} - 1 \leq \delta_n,\]
\item $\Pnoise_n$ as the set of $z^n \in \Rn$ that satisfy \[\norm{z^n - x^nh_E^n}^2 \geq n \sigma_E^2 (1 - \delta'_n)\] given a fixed $x^n \in \codebook_n$ and $h_E^n \in \Rpn$,
\item $\Perg_n$ as the set of $h_E^n \in \Rpn$ that satisfy \[\left|  \frac{1}{n} \sum\limits_{i=1}^n \log\left(1+h_{E,i}^2 \SNR \right) - \E_{H_E}\left[1 + H_E^2 \SNR \right] \right| \leq \delta''_n.\]
\end{itemize}

Intuitively, $\Pout_n$ corresponds to the set of eavesdropper output powers and channel coefficients most likely to occur in conjunction. $\Pnoise_n$ corresponds to the least amount of noise added to $h_E^nx^n$ during transmission. $\Perg_n$ corresponds to the set of eavesdropper channel coefficients that we expect to occur and is needed for technical reasons. The following lemma proves that events from each of these sets occur with sufficiently high probability. 

\begin{lemma} \label{lem:Psets} \
Consider\footnote{$K^*$ is a parameter of the channel defined in \Cref{appendix:NOCSIT}, \Cref{lem:sunglemma}.} the constant $c = 1/(4 K^*) > 0$.
\begin{enumerate}
\item Let $\epsilon_n^1 = 2e^{-n c \delta_n^2 }$. For any $x^n \in \codebook_n$,  \[\P \left[\left( H_E^n,Z^n\right) \in \Pout_n \, \biggr|\, X^n = x^n \right] \geq 1 - \epsilon_n^1.\]
\item Let $\epsilon_n^2 = e^{-\frac{n }{4}{\delta'_n}^{2}}$. For any $x^n \in \codebook_n$ and $h_E^n \in \Rpn$, \[\P\left[ Z^n \in \Pnoise_n \, \biggr|\, X^n = x^n, H_E^n = h_E^n \right] \geq 1- \epsilon_n^2.\]
\item Let $\epsilon_n^3 = 2e^{-n c {\delta''_n}^2 }$. Then, \[\P \left[H_E^n \in \Perg_n \right] \geq 1-\epsilon_n^3 .\]
\end{enumerate}
\end{lemma}
\begin{proof}
See \Cref{appendix:NOCSIT}.
\end{proof}

\newcommand{\bound}{\gamma}
We now use the sets constructed above to create our typical set. Define each of the following sets:
\begin{align*}
\Tout_n &= \{(x^n,h_E^n,z^n): x^n \in \codebook_n \text{ and } (h_E^n,z^n) \in \Pout_n \},\\
\Tnoise_n &= \{(x^n,h_E^n,z^n): x^n \in \codebook_n, h_E^n \in \Rpn,  \text{ and } z^n \in \Pnoise_n \},\\
\Terg_n &= \{(x^n,h_E^n,z^n)\in \codebook_n \times \Perg_n \times \R^n \}.
\end{align*}
We can think of each of these three sets as the \textit{expansion} set that corresponds to each of the previous three sets $\Pout_n$, $\Pnoise_n$, and $\Perg_n$ but lives in the space $\codebook_n \times \Rpn \times \Rn$, the tuple of all codewords, eavesdropper channel coefficients, and eavesdropper outputs.

We now take the intersection of these sets to construct one final set 
\[\T_n = \Tout_n \cap \Tnoise_n \cap \Terg_n.\]

The following lemma shows that the tuple of main channel coefficients and the previous set, $\Rp^n \times \T_n$, is typical for any $n$. The main channel coefficients must be taken into account as well since we are on a fast fading wiretap channel but we will see shortly that in the case of No-CSIT, it plays little part.
\begin{lemma}\label{lem:typicaldefinition} Let $\epsilon_n = \epsilon_n^1 + \epsilon_n^2 + \epsilon_n^3$ then
\[\P\left[(H_T^n,X^n,H_E^n,Z^n) \in \Rp^n\times\T_n |X^n = x^n \right] \geq 1- \epsilon_n,\] for any $x^n \in \codebook_n$. That is, $\Rp^n\times\T_n$ is a $(1-\epsilon_n)$-typical set where $\epsilon_n$ is exponentially decreasing to $0$ as $n \to \infty$.
\begin{proof}
See \Cref{appendix:NOCSIT}.
\end{proof}
\end{lemma}

With our typical set $\Rp^n \times \T_n$ in hand, we are ready to prove the main result of this section and determine a characterization for semantically secure achievable rates for the fast fading wiretap channel with No-CSIT. 

\begin{thm}\label{thm:ImaxLessCE}
Consider the fast fading wiretap channel with No-CSIT and let $\T_n$ and $\epsilon_n$ be defined as in \Cref{lem:typicaldefinition}. It follows that:
\[\lim\limits_{\substack{n \to \infty\\\epsilon \to 0}} \frac{\Iemaxf}{n} \leq \frac{1}{2}\E_{H_E}\left[\log(1+H_E^2 \SNR)\right].\]
\end{thm}
\begin{proof}
See \Cref{appendix:NOCSIT}.
\end{proof}

The following corollary then tells us what semantically secure rates we can achieve given this bound.

\begin{cor}\label{cor:PositiveRate}
The wiretap coding scheme of \Cref{scheme} can achieve an overall semantic secrecy rate of $C_T - C_E$ on the No-CSIT fast fading wiretap channel when $C_T > C_E$ and $\Rc$ is chosen arbitrarily close to $C_T$.
\end{cor}
\begin{proof}
We can combine the previous theorem with \Cref{cor:Thm1Redux} and note that $\delta_n,\delta'_n,\delta''_n$ can be chosen in such a way that $\epsilon_n \to 0$ exponentially as $n\to \infty$.
\end{proof}

Note that to the best of the authors' knowledge, this is the best semantically secure achievable rate on a No-CSIT fast fading wiretap channel to date. Going further, we actually have achieved the secrecy capacity for a specific class of wiretap channels.

\begin{fact}\label{stochasticfact} \cite{stochasticdegrade}
The weak secrecy capacity of a \textit{stochastically degraded} fast fading wiretap channel with No-CSIT is given by $$C_S = C_T - C_E.$$
\end{fact}

Immediately this fact with the previous corollary implies that we can achieve the secrecy capacity with our wiretap coding scheme of \Cref{scheme} on stochastically degraded fast fading channels with No-CSIT.

\begin{cor}
Using the wiretap coding scheme of \Cref{scheme} on any fast fading stochastically degraded wiretap channel with No-CSIT we have the following:
\begin{enumerate}
\item It is possible to achieve the semantic secrecy capacity.
\item $\Cs \big|_\text{weak} = \Cs \big|_\text{semantic}$.
\end{enumerate}
\end{cor}

\subsection{Fading: Partial CSIT} \label{sub:Partial}
We now turn to the case of partial CSIT, where the transmitter has access to full CSI about the main channel but knows only the statistics of Eve's channel. Our goal in this subsection is the same as in the previous subsection - we wish to characterize a set of semantically secure rates for the wiretap channel at hand and we use \Cref{cor:Thm1Redux} to do so.

Since the transmitter has access to CSI about the main channel, every party can demultiplex the fast fading wiretap channel into a set of $d$ parallel channels by partitioning the channel coefficients of the main channel into $d$ intervals as done in \cite{bloch_barros_2011,goldsmith_varaiya_1997}. Each parallel wiretap channel is then composed of a time-invariant, constant gain Gaussian main channel with a fast fading eavesdropper channel characterized by $H_E$ as depicted in \Cref{fig:PartialBreakup}.

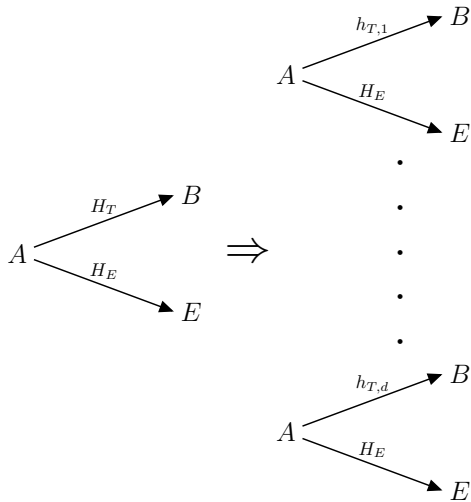
\begin{figure}[h]
  \begin{center}
  	\scalebox{0.70}{\begin{tikzpicture}[auto,scale=0.85, thick, node distance=2cm, >=triangle 45]
\draw node [name=A] {\Large $A$};
	\draw[->](A) -- node [above] {$H_T$} ++(3.5cm,1.3cm) node [right] (B) {\Large$B$};
    \draw[->](A) -- node [above] {$H_E$} ++(3.5cm,-1.3cm)  node [right] (E) {\Large$E$};
\node[text width=3cm] at (6.4,0) {\Huge$\Rightarrow$};
    \draw node [name=A1] at (6,4) {\Large $A$};
	\draw[->](A1) -- node [above] {$h_{T,1}$} ++(3.5cm,1.3cm) node [right] (B1) {\Large$B$};
    \draw[->](A1) -- node [above] {$H_E$} ++(3.5cm,-1.3cm)  node [right] (E1) {\Large$E$};
    \node[text width=1cm] at (9,2) {\Huge$\cdot$};
    \node[text width=1cm] at (9,1) {\Huge$\cdot$};
    \node[text width=1cm] at (9,0) {\Huge$\cdot$};
    \node[text width=1cm] at (9,-1) {\Huge$\cdot$};
    \node[text width=1cm] at (9,-2) {\Huge$\cdot$};
    
    \draw node [name=A1] at (6,-4) {\Large $A$};
	\draw[->](A1) -- node [above] {$h_{T,d}$} ++(3.5cm,1.3cm) node [right] (B1) {\Large$B$};
    \draw[->](A1) -- node [above] {$H_E$} ++(3.5cm,-1.3cm)  node [right] (E1) {\Large$E$};
\end{tikzpicture}}
  \end{center}
  \caption{Decomposition of the fast fading wiretap channel with partial CSIT.}
  \label{fig:PartialBreakup}
\end{figure}

More specifically, we assume the fading gain of the main channel is bounded as usual and divide the possible realizations of $H_T$ into intervals $\left[h_{T,i}, h_{T,i+1}\right)$ with $i\in \{1,\ldots,d \}$. Let 
\[p_i = \P\left[H_T \in \left[h_{T,i}, h_{T,i+1}\right.)\right].\]
Let $N_i$ be the random variable representing the number of times channel $i$ is \textit{actually} used, i.e., the number of times $h_{T,i}$ belongs to the $i$-th interval over all $n$ channel uses. Let $n_i = p_in - \varepsilon_i$ be a real number, where $\varepsilon_i$ is chosen sufficiently large so that $N_i$ is greater than $n_i$ with high probability and $\varepsilon_i \rightarrow 0$ as $n \rightarrow \infty$. In short, $n_i$ represents the number of times we \textit{plan} on the channel coefficients being realized in the $i$-th interval, whereas the realization of $N_i$ is how many times the the channel coefficients \textit{actually do occur} in the $i$-th interval. For every index $i$, the transmitter and legitimate receiver will publicly agree on a transmit power $\gamma_i(H_T)$ where $\{\gamma_i\}_{i=1}^d$ is chosen such that 
\[\sum_{i=1}^{d} p_i\gamma_{i} \leq P.\]

For $1\leq i \leq d$, the transmitter and legitimate receiver also publicly agree upon an ECC $\code_{n_{i}}^{i}$ (with codebook $\codebook_{n_{i}}^{i}$) designed to operate on the Gaussian point-to-point channel with constant channel gain $h_{T,i}$. We denote by $R_i$ the rate of $\code_{n_{i}}^{i}$ and the overall rate over the main channel to be \[R_{\code_n} = \sum_{i=1}^d p_iR_i.\] The full coding scheme is then outlined as follows: a message $m \in \M$ is chosen which passes through the preprocessing layer to produce an $l$-length pseudo-message $m' \in \M'$. These $l$ bits are then divided into sets of $n_{i}R_i$ bits such that \[l = \sum_{i} n_{i}R_{i}.\] A codeword is then generated for each of these sets by their respective $\code_{n_{i}}^{i}$ and the multiplexing strategy outlined in \cite{bloch_barros_2011,goldsmith_varaiya_1997} is then employed to transmit the $i$th codeword when the channel state is in the $i$th interval. In more detail, at each time instant $i$ the multiplexer will determine what the channel state is and send \textit{one} symbol from the codeword associated with that channel gain.

The reliability of this scheme comes from the aggregate reliability of all the ECC's being employed on the $d$ parallel channels and the fact that we are choosing $n_i <N_i$ with high probability. Since we are assuming an ECC $\code_{n_{i}}^i$ is chosen to be reliable over the $i$th point-to-point main channel, we know that the probability of error will be negligible:
\[\P_e(\code_{n_i}^i) \rightarrow 0 \text{ as } n_i \rightarrow \infty.\]
In other words, the receiver will be able to recover each $n_i$-length codeword with high probability. Thus the probability of error for the entire $n$-length transmission is just probability of error for each individual $n_i$-length codeword weighted by the probability that that code is used:
\[\P_e(\code_n) = \sum_i^d p_i \P_e(\code_{n_i}^i) \rightarrow \sum_i^d p_i \cdot 0 = 0 \text{ as } n \rightarrow \infty\]
since $n_i$ grows with $n$. Now that this scheme has been shown to be reliable, we now address its security.

We wish to asymptotically bound $\frac{1}{n}\Iemaxf$ of this fast fading channel by considering the set of $d$ parallel wiretap channels outlined above and each of \textit{their} individual associated max-information terms for which we already know the bound found in \Cref{thm:ImaxLessCE}. This is due to the fact that \Cref{thm:ImaxLessCE} did not impose any restrictions on the main channel distribution, it only required Eve's channel to be given arbitrarily as $H_E$. Thus having a constant gain main channel and no CSIT of Eve's channel is a special case of No-CSIT. The only way this differs from that of \Cref{sub:NoCSIT} is that in the case of No-CSIT, we are not allowed to vary the power we are transmitting at due to our lack of knowledge of instantaneous CSIT, whereas in the case of partial CSIT, we can vary our power to align with what the current main channel gain is.

Similarly to the case of No-CSIT, we wish to create a typical set which will contain enough content about our space of inputs, outputs, and channel coefficients. We accomplish this by creating typical sets for each of the $d$ subchannels and taking the Cartesian product of these to generate the typical set for the entire wiretap channel.

Define the following sets:
\begin{align*}
&\T'_{n_i} = \hspace{-0.05cm}\{(x^{n_i},h_T^{n_i},h_E^{n_i},z^{n_i})\hspace{-0.05cm}: h_T^{n_i} \in \Rp^{n_i}, (x^{n_i},h_E^{n_i},z^{n_i})\hspace{-0.05cm} \in \T_{n_i}\}\\
&\T'_{n} \hspace{0.105cm}= \bigotimes_{i}\T'_{n_i}
\end{align*}
where $\T_{n_i}$ is defined in \Cref{sub:NoCSIT}.

\begin{lemma}\label{lem:PartialTypical}
$\T'_n$ as defined above is a $(1-\epsilon_{n})$ typical set where $\epsilon_n$ is exponentially decreasing with $n$.
\end{lemma}
\begin{proof}
See \Cref{appendix:partial}.
\end{proof}

With the typical set $\T_n'$ in hand, we now aim to find an asymptotic bound $\xi$ for the average max-information for the entire $n$ uses of the wiretap channel $W$.

\begin{thm}\label{thm:ImaxPartialCSIT}
Consider a fast fading wiretap channel where the transmitter has partial CSIT with $\T'_n$ and $\epsilon_n$ as defined in \Cref{lem:PartialTypical}. Using the multiplexing scheme above, it follows that:
\begin{align*}
&\lim\limits_{\substack{n \to \infty\\\epsilon \to 0}} \left( \frac{\Iemaxf}{n}\right) \\
&\qquad\qquad\qquad\leq \frac{1}{2}\E_{H_EH_T}\left[\log\left(1+\frac{\gamma(H_T)H_E^2}{\sigma_E^2}\right)\right].
\end{align*}
\end{thm}

\begin{proof}
See \Cref{appendix:partial}.
\end{proof}

Now that we have found $\xi$, \Cref{cor:Thm1Redux} immediately tells us that by using the SS-UHF based preprocessing scheme we can achieve any positive rate, $R_s$, with semantic security satisfying \[R_s < R_\code - \frac{1}{2}\E_{H_EH_T}\left[\log\left(1+\frac{\gamma(H_T)H_E^2}{\sigma_E^2}\right)\right].\] Let's see how this compares to previous results.

\begin{fact}\label{PartialAchievable}
\cite{PartialCSIT_2013, bloch_barros_2011} For the fast fading wiretap channel where the CSI of the main channel but not the CSI of the eavesdropper channel is known at the transmitter, all rates $R_s$ such that
\begin{align*}
R_s &< \max_{\gamma} \left(\frac{1}{2}\E_{H_T}\left[\log\left(1+\frac{\gamma(H_T)H_T^2}{\sigma_T^2}\right)\right]\right. +\cdots\\
&\qquad\cdots- \left.\frac{1}{2}\E_{H_TH_E}\left[\log\left(1+\frac{\gamma(H_T)H_E^2}{\sigma_E^2}\right)\right]\right)
\end{align*}
where $\gamma:\Rp \rightarrow \Rp$ obeys the constraint $\E\left[\gamma(H_T)\right]\leq P$ are achievable secrecy rates under the strong (and weak) secrecy metric.
\end{fact}

To the extent of the authors' knowledge, the secure achievable rates given in \Cref{PartialAchievable} have never been extended to semantic security. However, the next corollary remedies this. 

\begin{cor}\label{PartialMatching}
The wiretap coding scheme of \Cref{scheme} can achieve all rates given in \Cref{PartialAchievable} with semantic security on the partial CSIT fast fading wiretap channel when the rate of the ECC, $R_\code$, is taken arbitrarily close to \[\frac{1}{2}\E_{H_T}\left[\log\left(1+\frac{\gamma(H_T)H_T^2}{\sigma_T^2}\right)\right]\] for any power allocation $\gamma(H_T)$. Moreover, these rates are achieved with exceptional semantic security.
\begin{proof}
The result follows immediately after combining \Cref{cor:Thm1Redux} with \Cref{thm:ImaxPartialCSIT} and noting that there does exist some ECC which can achieve this rate due to the fact that the above expression is less than or equal to the point-to-point capacity of the fast fading channel.
\end{proof}
\end{cor}

\subsection{Fading: Full CSIT}
In this subsection, we shall assume full CSIT; that is, we assume the  transmitter knows instantaneously the realizations at time instance $i$ of both the main and eavesdropper channel coefficients. The strategy used to find a set of semantically secure rates in this scenario is almost identical to that used in \Cref{sub:Partial} thus we omit most of the redundant explanations and proofs here. We now demultiplex the wiretap channel into $d^2$ parallel constant gain Gaussian wiretap channels determined by the channel coefficients of both channel $T$ and channel $E$. Since each of the parallel wiretap channels are now \textit{Gaussian} wiretap channels, we no longer use the bounds found in \Cref{thm:ImaxLessCE}, but rather use the bounds from \Cref{lem:AWGNbound} to bound the max-information of each of the parallel wiretap channels.

As before, we define a typical set for this channel as the Cartesian product of simpler sets:
\begin{align*}
\prescript{\star}{}{\Tout_{n_{ij}}} &= \{(x^{n_{ij}},h_T^{n_{ij}},h_E^{n_{ij}},z^{n_{ij}}):\\
&\qquad\quad h_T^{n_{ij}} \in \Rp^{n_{ij}}, (x^{n_{ij}},h_E^{n_{ij}},z^{n_{ij}}) \in \Tout_{n_{ij}}\},\\
\prescript{\star}{}{\Tnoise_{n_{ij}}} &= \{(x^{n_{ij}},h_T^{n_{ij}},h_E^{n_{ij}},z^{n_{ij}}):\\
&\qquad\quad h_T^{n_{ij}} \in \Rp^{n_{ij}}, (x^{n_{ij}},h_E^{n_{ij}},z^{n_{ij}}) \in \Tnoise_{n_{ij}}\},\\
\T'_{n_{ij}} &= \prescript{\star}{}{\Tout_{n_{ij}}} \cap \prescript{\star}{}{\Tnoise_{n_{ij}}},\\
\T'_{n}\hspace{0.21cm} &= \bigotimes_{i,j}\T'_{n_{ij}}.
\end{align*}

Note that $\T_{n_{ij}}^1$ and $\T_{n_{ij}}^2$ are defined in \Cref{sub:NoCSIT}. The following lemma, which is analogous to \Cref{lem:PartialTypical} from \Cref{sub:Partial}, shows that $T'_n$ is a typical set.
\begin{lemma}\label{lem:FullTypical}
$\T'_n$ as defined above is a $(1-\epsilon_n)$ typical set where $\epsilon_n$ is exponentially decreasing with $n$.
\end{lemma}

Now in an analogous way to \Cref{thm:ImaxLessCE} and \Cref{thm:ImaxPartialCSIT}, we have the following theorem for the full CSIT scenario.
\begin{thm}\label{thm:ImaxFullCSIT}
Consider the fast fading wiretap channel with full CSIT at the transmitter with $\T'_n$ and $\epsilon_n$ as defined in \Cref{lem:FullTypical}. Using the multiplexing scheme above, it follows that:
\begin{align*}
&\lim\limits_{\substack{n \to \infty\\\epsilon \to 0}} \left( \frac{\Iemaxf}{n}\right) \\
& \qquad \leq \frac{1}{2}\E_{H_E,H_T}\left[\log\left(1+\frac{\gamma(H_T,H_E)H_E^2}{\sigma_E^2}\right)\right].
\end{align*}
\end{thm}

Now that we have found the bound $\xi$, \Cref{cor:Thm1Redux} again tells us that by using the SS-UHF based preprocessing scheme we can achieve any positive rate, $R_s$, with semantic security satisfying \[R_s < R_\code - \frac{1}{2}\E_{H_E,H_T}\left[\log\left(1+\frac{\gamma(H_T,H_E)H_E^2}{\sigma_E^2}\right)\right].\] Once again, let's see how this compares to previous results.

\begin{fact}\label{FullAchievable}
With full CSI for both the main channel and the eavesdropper channels available at the transmitter, the strong secrecy capacity of the fast fading  wiretap channel is:
\begin{align*}
C_s &= \max_{\gamma}\left(\frac{1}{2}\E_{H_T H_E}\left[\log\left(1+\frac{\gamma(H_T,H_E)H_T^2}{\sigma_T^2}\right)  \right]+\right.\cdots\\
&\qquad\qquad\cdots-\left. \frac{1}{2}\E_{H_T H_E}\left[\log\left(1+\frac{\gamma(H_T,H_E)H_E^2}{\sigma_E^2}\right) \right]\right)
\end{align*}
where $\gamma:\Rp^2 \rightarrow \Rp$ obeys the power constraint $\E\left[\gamma(H_T,H_E)\right] \leq P$.
\end{fact}
This fact was originally given in \cite{secureoverfading} under the weak security metric but was upgraded to the strong security metric in \cite{chresolv}. However, to the extent of the authors' knowledge, this result has never been upgraded to semantic security. We provide such a generalization in the next corollary.

\begin{cor}\label{MatchingFullAchievable}
The semantic secrecy capacity of the fast fading wiretap channel with full CSIT is given by:
\begin{align*}
C_s =& \max_{\gamma}\left(\frac{1}{2}\E_{H_T H_E}\left[\log\left(1+\frac{\gamma(H_T,H_E)H_T^2}{\sigma_T^2}\right)  \right] \right.\cdots\\
&\qquad\cdots - \left.\frac{1}{2}\E_{H_T H_E}\left[\log\left(1+\frac{\gamma(H_T,H_E)H_E^2}{\sigma_E^2}\right) \right]\right).
\end{align*}
Furthermore, the transmission scheme of \Cref{scheme} can achieve the semantic secrecy capacity of the fast fading wiretap channel with full CSIT exceptionally fast.
\end{cor}
\begin{proof}
Let $\gamma^*$ be the power allocation function that maximizes the expression in \Cref{FullAchievable} as found in \cite{secureoverfading}. Let the rate of the ECC, $R_\code$, be taken arbitrarily close to
\[\frac{1}{2}\E_{H_T H_E}\left[\log\left(1+\frac{\gamma^*(H_T,H_E)H_T^2}{\sigma_T^2}\right)\right].\]
We know by Shannon's noisy channel coding theorem that \textit{some} ECC will exist which satisfies this rate due to the above expression being less than or equal to the point-to-point capacity of the main fast fading channel. Since the bound found in \Cref{thm:ImaxFullCSIT} holds for any power allocation function $\gamma$, it holds for $\gamma^*$ in particular. In \Cref{thm:ImaxFullCSIT} we found an upper bound to the right hand term of the difference in \Cref{FullAchievable}, thus invoking \Cref{cor:Thm1Redux} we know we can 
achieve any rate arbitrarily close to the secrecy capacity given in \Cref{FullAchievable}. Therefore the semantic secrecy capacity is equal to the weak secrecy capacity by \Cref{fact:CsEqual} in the case of full CSIT and the given wiretap coding scheme achieves it.
\end{proof}

\section{Future Work}
For wiretap channels that do not fall into the purview of the previously listed channels, one must apply \Cref{procedure} in its entirety. Hopefully however, the proof techniques employed here will help guide those pursuits.

As another interesting line of future work, one may try to find a tighter upper bound $\xi$ to the max-information per channel symbol on the fast fading wiretap channel with No-CSIT. Indeed, we proved the case when $\xi = C_E$ (the capacity of the eavesdropper's point-to-point channel), but perhaps this can be improved by clever power allocation techniques.

\section{Conclusion}\label{sec:Conclusion}
The main purpose of this paper has been to amplify the results of physical layer security into a more practical setting. In particular, we have developed a concrete and efficient converter that takes as input an error correcting code and outputs a semantically secure wiretap code. We have addressed five separate wiretap channels that are arguably the most popular in literature and have shown for each which semantically secure rates are achievable.

\section*{Acknowledgment}
The authors would like to thank Himanshu Tyagi of the Indian Institute of Science and Alexander Vardy of UC San Diego upon whose work this paper is largely motivated.

\appendices
\crefalias{section}{appsec}

\renewcommand{\thelemma}{\Alph{section}\arabic{lemma}}

\setcounter{lemma}{0}
\section{Our Construction is an efficient SS-UHF: Proof of \Cref{prop:SSUHF} and \Cref{prop:UHFisEfficient}} \label{appendix:UHFconstruction}
In this beginning appendix, we will prove that our UHF construction based on finite field arithmetic is an SS-UHF (\Cref{prop:SSUHF}) and that it is efficient (\Cref{prop:UHFisEfficient}).

\subsection*{\textbf{Proof of \Cref{prop:SSUHF}}}
\begin{proof}[\unskip\nopunct]
We will show $\F^*$ is universal, uniform, $(l-k)$-regular, and evenly invertible.

\begin{itemize}
\item \textit{Universality:}
Fix $m'_1 \neq m'_2 \in \M'$. We wish to count how many $(s,t)$ satisfy: \[\left[ \left( s \odot m'_1\right) \oplus t \right]_k = \left[ \left( s \odot m'_2\right)\oplus t \right]_k.\]

Since $\oplus$ is equivalent to bitwise addition, we can distribute $[\cdot]_k $ and reduce the equation to: \[ [s \odot m'_1]_k \oplus_k [t]_k = [s \odot m'_2]_k \oplus_k [t]_k\] where $\oplus_k$ is addition over $GF(2^k)$. This reduces even further to $[s \odot m'_1]_k = [s \odot m'_2]_k$, however, this is an equation that does not involve $t$ so that indeed, any choice of $t$ satisfies the original equation. This equation can be rewritten as 
\begin{align*}
0^k &=[s \odot m'_1]_k \oplus_k [s \odot m'_2]_k\\
&= [(s \odot m'_1) \oplus (s \odot m'_2)]_k\\
&= [s \odot m'']_k
\end{align*}
where we have defined $m'' = m'_1 \oplus m'_2$. Now since $m'_1 \neq m'_2$ then $m'' = m'_1 \oplus m'_2 \neq 0^l$. Moreover, by assumption $s \neq 0^l$ so that for each choice of $s$, the multiplication $s \odot m''$ is a unique element in $\{0,1\}^l\setminus 0^l$. Note that since there are $2^{l-k}-1$ elements in $\{0,1\}^l\setminus 0^l$ where the first $k$ bits set to $0$, then there are $2^{l-k}-1$ choices of $s$ that satisfy $0^k = [s \odot m'']_k$. In summary, we have $2^l$ choices for $t$ and $2^{-k}(2^l -2^k)$ choices for $s$, thus we have $2^{-k} 2^l (2^l -2^k)$ choices for $(s,t)$ that satisfy $[( s \odot m'_1) \oplus t ]_k=[ \left( s \odot m'_2\right) \oplus t ]_k$. However, $2^{-k} 2^l (2^l -2^k) \leq 2^{-k} 2^l (2^l - 1)$ since $k \geq 1$ so that (noting $|\seed| = 2^l (2^l-1)$) we have proved that $\F^*$ is a universal hash family.

\ \\

\item \textit{Uniformity:}
Fix $m' \in \M'$ and $m \in \M$. We wish to count how many $(s,t)$ satisfy: \[ [ \left( s \odot m'\right) \oplus t ]_k = m.\] 

We can distribute $[\cdot]_k$ and view this as the equation $[t]_k = m \oplus_k [s \odot m']_k$. For each choice of $s$ the first $k$ bits of $t$ are fixed and the last $l-k$ bits are free; thus there are $2^{l-k}$ choices for $t$. Since there are no restrictions at all on $s$, we can choose any of the $2^{l}-1$ $l$-length bits strings (excluding $0^l$) for $s$. In aggregate there are $2^{l-k}(2^l-1)$ choices of $(s,t)$ that satisfy $[ \left( s \odot m'\right) \oplus t ]_k = m$. Noting again that $|\seed| = 2^l (2^l -1)$ we have proven that our family $\F^*$ is uniform.
\ \\

\item \textit{Regularity:} Fix some $m \in \M$, $s \in \{0,1\}^l \setminus 0^l$, and $t \in \{0,1\}^l$. We wish to count how many $m'$ satisfy: \[ [ \left( s \odot m'\right) \oplus t ]_k = m.\]

As usual, break up this equation to $[s \odot m']_k = m \oplus_k [t]_k$. Since we are working in $GF(2^l)$ and $s \neq 0^l$, for each choice of $m' \in \{0,1\}^l$ the product $s \odot m'$ will be a \textit{unique} element in $\{0,1\}^l$. But by the previous equation, the first $k$ bits of $ s \odot m'$ are fixed at $m \oplus_k [t]_k$ while the last $l-k$ bits are completely free. Hence there will be $2^{l-k}$ choices of $m'$ that satisfy the original equation.  

Therefore, $\F^*$ is $(l-k)$-regular.

\ \\
\item \textit{Invertibility:} Let $m \in \M$, $s \in \{0,1\}^l \setminus 0^l$, and $t \in \{0,1\}^l$. Then,
\begin{align*}
f_{s,t}(\phi_{s,t,R}(m)) &= [ s \odot \left( s^{-1} \odot \left((m || R)  \oplus t \right) \right) \oplus t ]_k\\
&= [ (m||R) \oplus t \oplus t ]_k\\
&= [ m||R ]_k \\
&= m.
\end{align*}
Hence, $\F^*$ is invertible.
\ \\
\item \textit{Even Invertibility:} Suppose we are given a $m \in \M$, $s \in \{0,1\}^l \setminus 0^l$, and $t \in \{0,1\}^l$. Then $\phi_{s,t,r}(m) = \left( s^{-1} \odot (m || r) \oplus t \right)$ is a unique element for every choice of $r$. Since $R \sim \U(\{0,1\}^{l-k})$ the pseudo-message $M' = \phi_{s,t,R}(M)$ will be uniform.
\end{itemize}
\ \\
In summary, we have proven that $\F^*$ is an SS-UHF, thus concluding the proof of \Cref{prop:SSUHF}.
\end{proof}

\subsection*{\textbf{Proof of \Cref{prop:UHFisEfficient}}}
\begin{proof}[\unskip\nopunct]\ 
First recall that $l$ and $k$ are functions of the block length $n$.
\begin{enumerate}
\item Concatenation has time complexity $\bigO(k + (l-k))$ and thus is linear with $n$: $\bigO(k + (l-k)) = \bigO(l) = \bigO(n R_{\code_n}) = \bigO(n)$. Addition in $GF(2^l)$ operates as bitwise addition (or XOR) and thus the time complexity is also linear with $n$: $\bigO(l) = \bigO(n R_{\code_n}) = \bigO(n)$. Therefore, the operation $(m || r) \oplus t $ has linear time complexity.

Now inversion and multiplication in $GF(2^l)$ is known to be computed in at worst quadratic time in $l$ (cf. \cite[Chapter 2]{finitefieldTimeComplexity}). Thus computing the entire inverse $s^{-1} \odot ((m || r) \oplus t)$ is $\bigO(n^2)$.
\ \\
\item 
Using the same arguments as above, the operation $m' \odot s$ can be implemented in quadratic time and addition can be implemented in linear time. Clearly, $[\cdot]_k$ can be implemented in $\bigO(k) = \bigO(n)$: linear time with $n$. Thus, the entire post-processing scheme also can be implemented in quadratic time in $n$.
\end{enumerate}
This concludes the proof of \Cref{prop:UHFisEfficient}.
\end{proof}

\setcounter{lemma}{0}
\section{Security and Rates: Proof of \Cref{lem:LHL} and \Cref{thm:LHL3}} \label{appendix:LHL}
In this appendix, we will prove the two main statements related to the security and achievable rates of our wiretap coding scheme of \Cref{scheme}. Before we begin, we will need the following lemma. Not only do we use it several times in the proofs of the aforementioned results, but also, this lemma justifies our definition of $\alpha$-mutual information as we required $M'$ to be uniform there.

\begin{lemma}\label{lem:mUniform}
The pseudo-message $M'$ is a uniform random variable over the set $\M'$ where $|\M'| = 2^l$, i.e., $\omega(m') = 2^{-l}$.
\end{lemma}
\begin{proof}
We claim that $M'$ is a uniform random variable over the set $\M' = \{0,1\}^l$. We already argued in \Cref{scheme} that given $m$ and $s$, $M'$ is a uniform random variable over $\text{supp}(\phi_s(m))$, hence, we simply need to show that $\omega(m') = 2^{-l}$. Consider the following string of equalities:
\begin{align*}
\omega(m') &\explainequals{1} \sum_{m \in \M} \sum_{s \in \seed} \omega(m',m,s)\\
&\explainequals{2} \sum_{m \in \M} \sum_{s \in \seed} \omega(m'|m,s) P_M(m) P_S(s)\\
&\explainequals{3} 2^{-b}\sum_{m \in \M}P_M(m) \left(\frac{1}{|\seed|} \sum_{s \in \seed}  \indicator\left(m = f_s(m') \right)\right)\\
&\explainequals{4} 2^{-l} \sum_{m \in \M}P_M(m)\\
&= 2^{-l}.
\end{align*}
\begin{explain}\
\begin{enumerate}[1)]
\item Marginal density properties.
\item $M\bot S$ by assumption.
\item $\omega(m'|m,s) = 2^{-b}\indicator\left(m = f_s(m') \right)$ as mentioned in \Cref{scheme}.
\item The term $\left(\frac{1}{|\seed|} \sum_{s \in \seed}  \indicator\left(m = f_s(m') \right)\right) = 2^{-k}$ for any $m \in \M$ and $m' \in \M'$ by the uniform property of our SS-UHF. Moreover $b = l-k$. 
\end{enumerate}
\end{explain}
This concludes the proof of \Cref{lem:mUniform}.
\end{proof}

\newcounter{explainenum}		
\setcounter{explainenum}{0}		
\newcommand{\up}{\arabic{explainenum}\stepcounter{explainenum}}	

\subsection*{\textbf{Proof of \Cref{lem:LHL}}}
\begin{proof}[\unskip\nopunct]
This theorem is the primary tool of this paper. The proof is similar at times to the proof given in \cite{UHF} (for the analogous result for strong security only) and is a very straightforward application of our SS-UHF to the definition of mutual information. Notwithstanding, the proof is rather long and as a point of convenience we note that the proof ends on page \pageref{endOfLemma}.

We first need the following fact which follows immediately from the chain rule of mutual information:
\[\max\limits_{P_M} I(M \sep Z^n) \leq \max\limits_{P_M} I(M \sep Z^n,S).\]
Thus, it is sufficient to bound $\max_{P_M} I(M \sep Z^n,S)$.

We will split the proof into two parts, $\epsilon > 0$ and $\epsilon = 0$, starting with the $\epsilon = 0$ case. As mentioned previously, $1$-typical sets $\T$ are equal to the entire space $\M' \times \Z^n$ less a set of measure $0$, so that $I_\alpha^0(M' \sep Z^n)= I_\alpha(M' \sep Z^n)$. To show our claim is valid, it is therefore sufficient in the case of $\epsilon = 0$ to show: \[\max_{P_M} I(M \sep Z^n,S) \leq \frac{1}{\ln 2}2^{ \frac{1}{2} \left( -b + I_2(M' \sep Z^n)\right)}.\] 

To begin, suppose $M$ has some arbitrary distribution. Since $\seed$ and $\M$ are finite the definition of conditional mutual information $I(M \sep Z^n|S)$ is given by \[\int\limits_{\Z^n} \sum\limits_{m \in \M} \sum\limits_{s \in \seed} \omega(m,z^n,s)\log\left(\dfrac{\omega(m,z^n|s)}{\omega(m|s) \omega(z^n|s)}\right) \mu(dz^n),\] where $\mu$ is some measure on $\Z^n$.

From the \textit{chain rule of mutual information}, since $M \bot S$ by assumption, we have $I(M \sep Z^n,S) = I(M \wedge Z^n|S)$. It then follows that
\begin{align*}
&I(M \sep Z^n,S)\\
&=  I(M \sep Z^n|S) \\
&= \int\limits_{\Z^n} \sum\limits_{m \in \M} \sum\limits_{s \in \seed} \omega(m,z^n,s)\log\left(\dfrac{\omega(m,z^n|s)}{\omega(m|s) \omega(z^n|s)}\right) \mu(dz^n)\\
&= \int\limits_{\Z^n} \sum\limits_{m \in \M} \sum\limits_{s \in \seed} \omega(m,z^n,s)\log\left(\dfrac{\omega(z^n|m,s)}{\omega(z^n|s)}\right) \mu(dz^n) \numberthis \label{eq:thm_LHL_denom}.
\end{align*}

Let us now expand each conditional density of the numerator and denominator of the logarithm in \Cref{eq:thm_LHL_denom}. Starting with the numerator we have:

\begin{align*}
\omega(z^n | m,s) &= \sum\limits_{m' \in \text{supp}(\phi_s(m))} \omega(z^n|m')  \omega(m'|m,s) \numberthis \label{eq:thm_LHL_support}\\
&= 2^{-b} \sum\limits_{m' \in \text{supp}(\phi_s(m))} \omega(z^n|m') \numberthis \label{eq:thm_LHL_sneakers}\\
&= 2^{-b} \sum\limits_{m' \in \M'} \omega(z^n|m') \indicator\left(f_s(m') = m \right) \numberthis \label{eq:thm_LHL_sneakerslast}.
\end{align*}
\Cref{eq:thm_LHL_support} follows from the fact that we can take $M'$ as an intermediate node and sum over all possible realizations of $M'$; by assumption, since we are given $m$ and $s$, then $M'$ can only be found in $\text{supp}(\phi_s(m))$ where $\phi_s$ is the even-inverse of $f_s$. \Cref{eq:thm_LHL_sneakers} follows from the fact that once given $m$ and $s$, the density of $M'$ is uniform on a set with $2^b$ elements which follows from the fact that our SS-UHF is $b$-regular and evenly invertible.

The expansion of the conditional density in the denominator of the logarithm of \Cref{eq:thm_LHL_denom} is given by:
\begin{align*}
\omega(z^n|s) &= \frac{\omega(z^n,s)}{P_S(s)}\\
&\explainequals{\up} \sum\limits_{m \in \M} \frac{\omega(z^n,m,s)P_M(m)}{P_S(s)P_M(m)}\\
&\explainequals{\up} \sum\limits_{m \in \M} \omega(z^n|m,s)P_M(m) \\
&= 2^{-b}  \sum\limits_{m' \in \Mprime} \omega(z^n|m')  \sum\limits_{m \in \M} P_M(m)\indicator{(f_s(m')=m)} \\
&\explainequals{\up} 2^{-b} \sum\limits_{m'\in \M'} \omega(z^n|m') P_M( f_s(m')) \numberthis \label{eq:thm_LHL_banana}.
\end{align*}
\begin{explain}\
\begin{enumerate}\setcounter{enumi}{0}
\item Marginal density property.
\item By assumption, $M \bot S$.
\item When $s$ is fixed, $f_s$ is a well defined function. Thus, inside the sum over $\M'$, $f_s(m')$ can map to only a single $m \in \M$. Therefore, the indicator is $1$ only for a single value of $m$; namely, when $m = f_s(m')$.
\end{enumerate}
\end{explain}

We now continue expanding the leakage (\Cref{eq:thm_LHL_denom}) using these two conditional densities.
\begin{align*}
&I(M \sep Z^n,S)\\
&\explainequals{\up} \int\limits_{\Z^n} \sum\limits_{m \in \M} \sum\limits_{s \in \seed} \omega(m,z^n,s)\cdots \\
&\quad\cdots  \log \left(\frac{ 2^{-b} \sum_{u' \in \Mprime} \omega(z^n|u')\indicator{(f_s(u') = m)}}{ 2^{-b} \sum\limits_{u'' \in \Mprime} \omega(z^n|u'')  P_M(f_s(u'')) }\right) \mu(dz^n) \\
&\explainequals{\up} \int\limits_{\Z^n}  \sum\limits_{m \in \M} \sum\limits_{s \in \seed} \omega(z^n|m,s)P_M(m)P_S(s) \cdots \\
&\qquad\cdots\log \left(\frac{ \sum_{u' \in \Mprime} \omega(z^n|u')\indicator{(f_s(u') = m)}}{ \sum\limits_{u'' \in \Mprime} \omega(z^n|u'')  P_M(f_s(u'')) }\right)  \mu (dz^n) \\
&\explainequals{\up} \frac{1}{|\seed|} \int\limits_{\Z^n} \sum\limits_{m \in \M} \sum\limits_{s \in \seed} \omega(z^n|m,s)P_M(m) \cdots \\
&\qquad\cdots \log \left(\frac{ \sum_{u' \in \Mprime} \omega(z^n|u')\indicator{(f_s(u') = m)}}{ \sum\limits_{u'' \in \Mprime} \omega(z^n|u'')  P_M(f_s(u'')) }\right)\mu(dz^n).
\end{align*}

\begin{explain}\
\begin{enumerate}\setcounter{enumi}{3}
\item We will break with our convention slightly. Here we have written $\omega(z^n|u')$ as shorthand for $\omega_{Z^n|M'}(z^n|u')$; analogously for $\omega(z^n|u'')$. We will stick with this new convention for the remainder of the proof; i.e.\ $\omega(\cdot|u^*)$ and $\omega(u^*)$ will be shorthand for densities with respect to $M'$.
\item By assumption, $M \bot S$.
\item By assumption, $S \sim \U(\seed)$.
\end{enumerate}
\end{explain}

At this point we can expand the conditional density $\omega(z^n|m,s)$ (from \Cref{eq:thm_LHL_sneakerslast}) and continue:
\begin{flalign*}
&= \frac{2^{-b}}{|\seed|} \int\limits_{\Z^n} \sum\limits_{\substack{m \in \M \\ s \in \seed \\m' \in \M'}} \omega(z^n|m') P_M(m) \indicator{(f_s(m') = m)}\cdots \\
& \quad \cdots\log \left(\frac{ \sum_{u' \in \Mprime} \omega(z^n|u')\indicator{(f_s(u') = m)}}{ \sum\limits_{u'' \in \Mprime} \omega(z^n|u'')  P_M(f_s(u'')) }\right) \mu (dz^n) \numberthis \label{eq:thm_LHL_star}\\
&\explainequals{\up}  \frac{2^{-b}}{|\seed|} \int\limits_{\Z^n} \sum\limits_{\substack{m \in \M \\ s \in \seed \\m' \in \M'}} \omega(z^n|m')P_M(m)\indicator{(f_s(m') = m)}\cdots \\
&\quad\cdots \log \left(\frac{ \sum_{u' \in \Mprime} \omega(z^n|u')\indicator{(f_s(u') = f_s(m'))}}{ \sum\limits_{u'' \in \Mprime} \omega(z^n|u'')  P_M(f_s(u'')) }\right) \mu(dz^n) \\
&\explainequals{\up}  \frac{2^{-b}}{|\seed|} \int\limits_{\Z^n} \sum\limits_{\substack{s \in \seed \\m' \in \M'}} \omega(z^n|m')P_M(f_s(m'))\cdots \\
&\quad\cdots \log \left(\frac{ \sum_{u' \in \Mprime} \omega(z^n|u')\indicator{(f_s(u') = f_s(m'))}}{ \sum\limits_{u'' \in \Mprime} \omega(z^n|u'')  P_M(f_s(u'')) }\right) \mu(dz^n)  \\
&\explainequals{\up} 2^{-b}  \int\limits_{\Z^n}  \Biggr[\frac{1}{|\seed|} \sum\limits_{\substack{s \in \seed \\m' \in \M'}} \omega(z^n|m')P_M(f_s(m')) \cdots \\
&\qquad\cdots  \log \left( \sum_{u' \in \Mprime} \omega(z^n|u')\indicator{(f_s(u') = f_s(m'))} \right)\Biggr] +\cdots \\
&\hspace*{1.06cm} \cdots+ \Biggr[-\frac{1}{|\seed|} \sum\limits_{\substack{s' \in \seed \\m'' \in \M'}} \omega(z^n|m'')P_M(f_{s'}(m'')) \cdots \\
&\qquad\cdots \log \left( \sum\limits_{u'' \in \Mprime} \omega(z^n|u'')  P_M(f_{s'}(u'')) \right) \Biggr]  \mu(dz^n) \numberthis \label{eq:thm_LHL_wine}.
\end{flalign*}

\begin{explain}\
\begin{enumerate}\setcounter{enumi}{6}
\item The entire summand is 0 unless $m = f_s(m')$, so we can replace the $m$ in the indicator function of the log as such as long as we stick with the convention that $0 \log 0 = 0$ as the limit suggests.
\item As in \Cref{eq:thm_LHL_banana}, the indicator will filter all but a single $m$; namely, when $m = f_s(m')$.
\item We can break up the logarithm into a subtraction where we change indices of the summation so as not to become confused.
\end{enumerate}
\end{explain}

We will now consider each of expressions within the square brackets of \Cref{eq:thm_LHL_wine} separately, starting with the first. The first square bracket can be written (after multiplying by the unit $2^{-k}2^k$) as 
\begin{align*}
&2^{-k} \sum_{m' \in \Mprime} \omega(z^n|m') \left[ \sum_{s \in \seed} \frac{ 2^k}{|\seed|} P_M(f_s(m'))  \right. \cdots \\
&\qquad\cdots \left. \log \left( \sum_{u' \in \Mprime} \omega(z^n|u')\indicator{(f_s(u') = f_s(m'))}\right) \right] \numberthis \label{thm_LHL_jensen}.
\end{align*}
Our goal now will be to move the sum over $s$ inside of the logarithm via Jensen's inequality. However, Jensen's incurs a multiplicative penalty if the weights do not sum to 1. Fortunately, our weights \textit{do} sum to 1 as shown next. Our preprocessor is an SS-UHF and hence it is \textit{uniform}. Thus for any $m' \in \M'$ we have: 
\begin{align*}
&\sum_{s \in \seed} \left(\frac{ 2^k}{|\seed|} P_M(f_s(m')) \right) \\  =&\sum_{m \in \M} P_M(m) \frac{2^k}{|\seed|} \sum_{s \in \seed} \indicator\left(f_s(m') = m \right) \\
= & 1.
\end{align*}
Thus we can aptly apply Jensen's inequality (without carrying around any extra factors) and move the preceding term inside of the logarithm at the expense of an inequality. This yields:
\begin{align*}
&(\ref{thm_LHL_jensen}) \leq 2^{-k} \sum_{m' \in \Mprime} \omega(z^n|m') \log\left( 2^k \sum_{u' \in \Mprime} \omega(z^n|u') \right. \cdots \\
&\hspace*{2cm} \cdots \left. \sum_{s \in \seed}\frac{P_M(f_s(m'))}{|\seed|} \indicator{(f_s(u') = f_s(m'))} \right) \numberthis \label{eq:thm_LHL_uprimemprime}.
\end{align*}

If $u' = m'$ in \Cref{eq:thm_LHL_uprimemprime}, it is clear that the indicator will always return 1 regardless of $s\in \seed$ so that the argument of the logarithm becomes \[ \sum_{u' \in \Mprime} \omega(z^n|u') \indicator(u'=m') = \omega(z^n|m'),\] where we have again used the fact that our preprocessor is a SS-UHF and is hence uniform.

On the contrary, if $u' \neq m'$ in \Cref{eq:thm_LHL_uprimemprime}, the indicator will only return 1 some of the time, and a nice simplification of the expression is not obvious at this time; we will address this in a bit.

Combining these cases together, the entire first square bracket of \Cref{eq:thm_LHL_wine} is less than or equal to:
\begin{align*}
& 2^{-k} \sum_{m' \in \Mprime} \omega(z^n|m') \log \left[  \vphantom{\sum_{m' \in \Mprime}} \omega(z^n|m') +2^k \sum_{u' \in \Mprime} \omega(z^n|u')\cdots \right.\\
&\left.\qquad \cdots \sum_{s \in \seed}\frac{P_M(f_s(m'))}{|\seed|} \indicator{(f_s(u') = f_s(m'))} \indicator{(u' \neq m')}  \vphantom{\sum_{m' \in \Mprime}} \right].
\end{align*}

Let us now move onto the second square bracket of \Cref{eq:thm_LHL_wine} above. We can write this term as
\begin{align*}
& -\frac{1}{|\seed|} \sum_{s' \in \seed} \left( \sum_{m'' \in \Mprime} \omega(z^n|m'')  P_M(f_{s'}(m''))\right)\cdots \\
&\hspace*{1.8cm} \cdots\log \left( \sum\limits_{u'' \in \Mprime} \omega(z^n|u'') P_M(f_{s'}(u'')) \right) \\
&\leq -\frac{1}{|\seed|} \left(\sum_{s' \in \seed}  \sum\limits_{m'' \in \Mprime} \omega(z^n|m'')  P_M(f_{s'}(m''))\right) \cdots \\
&\hspace*{1.8cm} \cdots \log \left( \frac{\sum\limits_{s'' \in \seed} \sum\limits_{u'' \in \Mprime} \omega(z^n|u'') P_M(f_{s''}(u''))}{\sum\limits_{s''' \in \seed} 1} \right),
\end{align*}
where the inequality follows from the log-sum inequality. Now again using the fact that our preprocessor is an SS-UHF and hence uniform we have the formula $\sum_{s \in \seed} \left(\frac{ 1}{|\seed|} P_M(f_s(m')) \right) = 2^{-k}$ for any $m' \in \M'$. Using this, the entire second square bracket of \Cref{eq:thm_LHL_wine} becomes less than or equal to
\begin{align*}
-2^{-k} \sum\limits_{m'' \in \Mprime} \omega(z^n|m'') \log \left( 2^{-k} \sum\limits_{u'' \in \Mprime} \omega(z^n|u'')\right).
\end{align*}

We are now at a point where each square bracket of \Cref{eq:thm_LHL_wine} is properly simplified. Thus:
\begin{align*}
&(\ref{eq:thm_LHL_wine}) \leq 2^{-b-k}  \int\limits_{\Z^n} \sum_{m' \in \Mprime} \omega(z^n|m')\log\Bigg( \frac{2^k\omega(z^n|m')}{\sum\limits_{u'' \in \Mprime} \omega(z^n|u'')} +\cdots \\
& \cdots+ \frac{2^k2^k}{\sum\limits_{u'' \in \Mprime} \omega(z^n|u'')} \sum\limits_{u' \in \Mprime} \omega(z^n|u') \sum\limits_{s \in \seed} \frac{P_M(f_s(m'))}{|\seed|}   \cdots \\
&\qquad \qquad\qquad \cdots \indicator{(f_s(m') = f_s(u'))} \indicator{(u' \neq m')} \Bigg) \mu(dz^n) \numberthis \label{eq:thm_LHL_blast}.
\end{align*}

We will now simplify the inside of the logarithm. Consider the first summand given by \[\frac{2^k\omega(z^n|m')}{\sum\limits_{u'' \in \Mprime} \omega(z^n|u'')}.\] Conditional densities are defined as $\omega(z^n|u'') = \frac{\omega(z^n,u'')}{\omega(u'')}$. By \Cref{lem:mUniform}, $\omega(u'') = 2^{-l}$ for every $u'' \in \M'$ so that $\omega(z^n|u'') = 2^l \omega(z^n,u'')$. Then by the marginal property of densities, $\sum_{u'' \in \M'} \omega(z^n|u'') = 2^l \sum_{u'' \in \M'} \omega(z^n,u'') = 2^l \omega(z^n)$. Moreover, using Bayes theorem and \Cref{lem:mUniform} again we can write 
\begin{align*}
   \omega(z^n|m') = \frac{\omega(m'|z^n)\omega(z^n)}{\omega(m')} = 2^l \omega(m'|z^n)\omega(z^n) \numberthis \label{eq:thm_LHL_bayes}. 
\end{align*}
The term $2^l \omega(z^n)$ appears both in the numerator and denominator and thus cancels out. Hence the entire first summand of the logarithm in \Cref{eq:thm_LHL_blast} becomes \[2^k \omega(m'|z^n).\]

Now the second summand of the logarithm of \Cref{eq:thm_LHL_blast} is given by
\begin{align*}
&\frac{2^k2^k}{\sum\limits_{u'' \in \Mprime} \omega(z^n|u'')} \sum\limits_{u' \in \Mprime} \omega(z^n|u') \sum\limits_{s \in \seed} \frac{P_M(f_s(m'))}{|\seed|}   \cdots \\
&\qquad \qquad\qquad \cdots \indicator{(f_s(m') = f_s(u'))} \indicator{(u' \neq m')}
\end{align*}
Using the same argument as in the preceding paragraph we have $\sum_{u'' \in \M'} \omega(z^n|u'')  = 2^l \omega(z^n)$ and $\omega(z^n|u') = 2^l \omega(z^n) \omega(u'|z^n)$. Again, the term $2^l \omega(z^n)$ appears in both the numerator and denominator thus canceling each other out. Thus the second summand of the logarithm of \Cref{eq:thm_LHL_blast} simplifies immediately to:
\begin{align*}
&2^k2^k \sum\limits_{u' \in \Mprime} \omega(u'|z^n) \sum\limits_{s \in \seed} \frac{P_M(f_s(m'))}{|\seed|}   \cdots \\
&\qquad \qquad\qquad \cdots \indicator{(f_s(m') = f_s(u'))} \indicator{(u' \neq m')}. \\
&=2^k 2^k \sum\limits_{u' \in \Mprime} \omega(u'|z^n) \indicator{(u' \neq m')}  \sum\limits_{m \in \M} P_M(m)  \cdots \\
&\qquad \cdots\frac{1}{|\seed|} \sum\limits_{s \in \seed}\indicator{(m = f_s(m'))} \indicator{(f_s(m') = f_s(u'))} \numberthis \label{eq:thm_LHL_water}.
\end{align*}
Now note that 
\begin{align*}
&\sum\limits_{m \in \M}\sum\limits_{s \in \seed} P_M(m) \frac{1}{|\seed|} \indicator{(m = f_s(m'))} \indicator{(f_s(m') = f_s(u'))}\\
&=\P_{MS}\left[M = f_S(m') \text{ and } f_S(m') = f_S(u') \right]\\
&= \P_{MS}\left[f_S(m') = f_S(u')\;|\; M = f_S(m') \right] \cdots \\
&\hspace*{3cm} \cdots \P_{MS}\left[ M = f_S(m') \right]\\
&= \P_{MS}\left[M = f_S(u')\right] \cdot \P_{MS}\left[ M = f_S(m') \right] \numberthis \label{eq:thm_LHL_bowl1}.
\end{align*}

However, 
\begin{align*}
\P_{MS}\left[M = f_S(u')\right] &= \sum_{m \in \M} \sum_{s \in \seed} P_M(m) \frac{1}{|\seed|} \indicator\left( m = f_s(u')\right)\\
&= \sum_{m \in \M} P_M(m) \frac{1}{|\seed|} \sum_{s \in \seed} \indicator\left( m = f_s(u')\right)\\
&= 2^{-k} \sum_{m \in \M} P_M(m) \numberthis \label{eq:thm_LHL_penultimate}\\
&= 2^{-k} \numberthis \label{eq:thm_LHL_bowl2},
\end{align*}
where \Cref{eq:thm_LHL_penultimate} follows immediately from the uniform property of our SS-UHF. From this we also have: \[\P_{MS}\left[ M = f_S(m') \right] = 2^{-k}.\]

Thus, combining \Cref{eq:thm_LHL_bowl1} and \Cref{eq:thm_LHL_bowl2} together with \Cref{eq:thm_LHL_water} simplifies the entire second summand of the logarithm in \Cref{eq:thm_LHL_blast} to 
\[\sum\limits_{u' \in \Mprime} \omega(u'|z^n) \indicator{(u' \neq m')} \leq \sum\limits_{u' \in \Mprime} \omega(u'|z^n) = 1.\]

Then it follows, (continuing on from \Cref{eq:thm_LHL_blast}):
\begin{align*}
&I(M \sep Z^n,S)\\
& \leq 2^{-b-k}  \int\limits_{\Z^n} \sum_{m' \in \Mprime} \omega(z^n|m')\log\left( 2^k\omega(m'|z^n) + 1 \right) \mu(dz^n)\\
& \explainequals{\up} \int\limits_{\Z^n}\omega(z^n) \sum_{m' \in \Mprime} \omega(m'|z^n)\log\left( 2^k\omega(m'|z^n) + 1 \right) \mu(dz^n)\\
& \explainlessthanequals{\up} \int\limits_{\Z^n}\omega(z^n)\log\left(  \sum_{m' \in \Mprime} \omega(m'|z^n)(2^k\omega(m'|z^n) + 1) \right) \mu(dz^n)\\
&= \int\limits_{\Z^n}\omega(z^n)\log\left(1+ 2^k\sum_{m' \in \Mprime} \omega(m'|z^n)^2  \right)\mu(dz^n)\\
&\explainlessthanequals{\up} \frac{2^{\frac{1}{2}k}}{\ln 2} \int\limits_{\Z^n}\omega(z^n) \left(\sum_{m' \in \Mprime} \omega(m'|z^n)^2 \right)^{\frac{1}{2}} \mu(dz^n)\\
&= \frac{1}{\ln 2}2^{ \frac{1}{2} \left(k + 2\log  \int\limits_{\Z^n}\omega(z^n) \left(\sum_{m' \in \Mprime} \omega(m'|z^n)^2 \right)^{\frac{1}{2}} \mu(dz^n) \right)}\\
&=  \frac{1}{\ln 2}2^{ \frac{1}{2} \left(k - H_2(M'|Z^n) \right)}\\
&= \frac{1}{\ln 2}2^{ \frac{1}{2} \left(k - l + l -H_2(M'|Z^n) \right)}\\
&\explainequals{\up} \frac{1}{\ln 2} 2^{ \frac{1}{2} \left(-b + I_2(M'\sep Z^n) \right)}\\
&\explainlessthanequals{\up}\frac{1}{\ln 2} 2^{ \frac{1}{2} \left(-b + I_\alpha(M'\sep Z^n) \right)} \quad \text{ for any } \alpha \in [2,\infty].
\end{align*}

\begin{explain}\
\begin{enumerate}\setcounter{enumi}{9}
\item \Cref{eq:thm_LHL_bayes} and $b=l-k$.
\item Jensen's inequality on the sum over $m'$.
\item Use the bound $\log( 1 + x) \leq \frac{1}{\ln 2}\sqrt{x}$ for all $x \geq 0$.
\item By \Cref{lem:mUniform},  $M'$ is uniform so that $H_\alpha(M') = l$ for any $\alpha$ and $I_\alpha(M' \sep Z^n) = H_\alpha(M') - H_\alpha(M'|Z^n)$. Also recall $b = l-k$.
\item By \Cref{fact:renyiOrdering1}, $I_2(M' \sep Z^n) \leq I_\alpha(M' \sep Z^n)$ for any $\alpha \in [2,\infty]$. In particular, $\alpha = \infty$ here proves the second part of our claim for the $\epsilon = 0$ case.
\end{enumerate}
\end{explain}

With this, we have constructed an upper bound to $I(M \sep Z^n, S)$ for an arbitrary message distribution $P_M$. However, since the bound did not depend on the specific choice of $P_M$, the bound also holds for $\max_{P_M} I(M \sep Z^n, S)$. Therefore, we have concluded the $\epsilon = 0$ case.

\newcommand{\Tpre}{{\T_{*}}}
\newcommand{\TpreC}{{\T^\complement_{*}}}
Let us move onto the $\epsilon >0$ case. Fix some $\epsilon > 0$ and consider some $(1-\epsilon)$ typical set $\T \subset \M' \times \Z^n$.

Now consider \Cref{eq:thm_LHL_star} in the previous string of inequalities written as:
\begin{align*}
& \frac{2^{-b}}{|\seed|} \int\limits_{\Z^n} \sum\limits_{\substack{m \in \M \\ s \in \seed}}P_M(m) \left[ \sum\limits_{m' \in \M'} \omega(z^n|m')  \indicator{(f_s(m') = m)}\right. \cdots\\
&\qquad\cdots\left.\log \left(\frac{ \sum\limits_{u' \in \Mprime} \omega(z^n|u')\indicator{(f_s(u') = m)}}{ \sum\limits_{u'' \in \Mprime} \omega(z^n|u'')  P_M(f_s(u'')) }\right)\right] \mu(dz^n) \numberthis \label{eq:thm_LHL_pumpernickel}.
\end{align*}
Inside of the square bracket of \Cref{eq:thm_LHL_pumpernickel}, $z^n$ and $m$ can be considered \textit{fixed}, and thus, each of the 3 sums over $\M'$ can be considered as a sum over two other sets: \[\M_1' = \{m'\in \M':(m',z^n) \in \T\} \text{ and}\] \[\M_2' = \{ m'\in \M':(m',z^n) \in \T^\complement\},\] where $\T^\complement$ denotes the complement of $\T$ in $\M' \times \Z^n$. 

With this, we can then apply the log-sum inequality to \Cref{eq:thm_LHL_pumpernickel} to yield the following:
\begin{align*}
&(\ref{eq:thm_LHL_pumpernickel})\leq \frac{2^{-b}}{|\seed|} \int\limits_{\Z^n} \sum\limits_{\substack{m \in \M \\ s \in \seed}}P_M(m)\cdots \\
&\cdots \left[\left( \sum\limits_{m' \in \M_1'} \omega(z^n|m')  \indicator{(f_s(m') = m)} \right)\right.\cdots\\
&\qquad\quad\cdots\left.\log \left(\frac{ \sum\limits_{u' \in \M_1'} \omega(z^n|u')\indicator{(f_s(u') = m)}}{ \sum\limits_{u'' \in \M_1'} \omega(z^n|u'')  P_M(f_s(u'')) }\right)\right. +\cdots \\
&\cdots +\left.
\left( \sum\limits_{m' \in \M_2'} \omega(z^n|m')  \indicator{(f_s(m') = m)} \right) \right.\cdots\\
&\qquad\quad\cdots\left.\log \left(\frac{ \sum\limits_{u' \in \M_2'} \omega(z^n|u')\indicator{(f_s(u') = m)}}{ \sum\limits_{u'' \in \M_2'} \omega(z^n|u'')  P_M(f_s(u'')) }\right)\right] \mu(dz^n).
\end{align*}

Now define $\mathcal{Q}_\T$ by
\begin{align*}
&\frac{2^{-b}}{|\seed|} \int\limits_{\Z^n} \sum\limits_{\substack{m \in \M \\ s \in \seed}}P_M(m)  \sum\limits_{m' \in \M'} \omega_\T(z^n|m')  \indicator{(f_s(m') = m)}\cdots \\
&\qquad \cdots\log \left(\frac{ \sum\limits_{u' \in \M'} \omega_\T(z^n|u')\indicator{(f_s(u') = m)}}{ \sum\limits_{u'' \in \M'} \omega_\T(z^n|u'')  P_M(f_s(u'')) }\right)\mu(dz^n), 
\end{align*}
so that \Cref{eq:thm_LHL_pumpernickel} yields: \[ I(M \sep Z^n,S) \leq \mathcal{Q}_\T + \mathcal{Q}_{\T^\complement}.\]

When considering just $\mathcal{Q}_\T$ we can continue where we left off from \Cref{eq:thm_LHL_star} of the previous proof ($\epsilon = 0$ case). In fact, it is not hard to see that almost nothing changes and we end up with
\begin{align*}
\mathcal{Q}_\T &\leq \frac{1}{\ln 2} 2^{\frac{1}{2}(-b+I_\alpha^\T(M' \sep Z^n))},
\end{align*}
for $\alpha \in [2,\infty]$.

Now let's focus on $\mathcal{Q}_{\T^\complement}$. It follows that:
\begin{flalign*}
&\mathcal{Q}_{\T^\complement}\\
&\explainlessthanequals{\up}  \frac{2^{-b}}{|\seed|} \int\limits_{\Z^n} \sum\limits_{s \in \seed} \left( \sum\limits_{m' \in \M'} \omega_{\T^\complement}(z^n|m') P_M(f_s(m')) \right) \cdots \\
&\qquad\quad\cdots\log \left(\frac{ \sum_{u' \in \Mprime} \omega_{\T^\complement}(z^n|u')}{ \sum\limits_{u'' \in \Mprime} \omega_{\T^\complement}(z^n|u'')  P_M(f_s(u'')) }\right) \mu(dz^n) \\
&\explainlessthanequals{\up}  \frac{2^{-b}}{|\seed|} \int\limits_{\Z^n} \left(  \sum\limits_{s \in \seed} \sum\limits_{m' \in \M'} \omega_{\T^\complement}(z^n|m') P_M(f_s(m')) \right)  \cdots \\
&\qquad\cdots\log \left(  \frac{|\seed| \sum_{u' \in \Mprime} \omega_{\T^\complement}(z^n|u')}{ \sum\limits_{s' \in \seed} \sum\limits_{u'' \in \Mprime} \omega_{\T^\complement}(z^n|u'')  P_M(f_{s'}(u'')) }\right) \mu(dz^n)\\
&\explainequals{\up}  2^{-l} \int\limits_{\Z^n}\sum\limits_{m' \in \M'} \omega_{\T^\complement}(z^n|m') \cdots \\
&\qquad\quad\cdots\log \left(  \frac{2^k \sum_{u' \in \Mprime} \omega_{\T^\complement}(z^n|u')}{ \sum\limits_{u'' \in \Mprime} \omega_{\T^\complement}(z^n|u'') }\right) \mu(dz^n) \\
&= k 2^{-l} \sum\limits_{m' \in \M'} \int\limits_{\Z^n} \omega(z^n|m')\indicator{\left( (m',z^n) \in \T^\complement \right)}  \mu(dz^n) \\
&=  k 2^{-l} \sum\limits_{m' \in \M'} \P\left[(M',Z^n) \in \T^\complement\,|\,M' = m' \right]\\
&\explainlessthanequals{\up} k \epsilon.
\end{flalign*}

\begin{explain}\
\begin{enumerate}\setcounter{enumi}{14}
\item In the numerator of the logarithm, we have used the trivial bound $\indicator{(f_s(u') = m)} \leq 1$ for all $s,u',m$.
\item Log-sum inequality.
\item Our preprocessor is a SS-UHF and hence it is uniform.
\item We chose $\T$ to be a $(1-\epsilon)$ typical set and there are $2^l$ pseudo-messages.
\end{enumerate}
\end{explain}
\ \\
Again, just as in the $\epsilon = 0$ case, we have provided an upper bound to $I(M \sep Z^n, S)$ for an arbitrary message distribution $P_M$ so that the upper bound also holds for $\max_{P_M} I(M \sep Z^n,S)$. This concludes the $\epsilon > 0$ case. 

Combining both cases, we have for any $\epsilon \geq 0$, $\alpha \in [2,\infty]$: \begin{align*}
\max_{P_M} I(M \sep Z^n,S)\leq \frac{1}{\ln 2} 2^{ \frac{1}{2}\left( -b+I_\alpha^\T(M' \sep Z^n) \right)} + \epsilon k. 
\end{align*}

Since this inequality was derived using an \textit{arbitrary} $(1-\epsilon)$-typical set $\T$, we may as well optimize our choice of $\T$ while keeping $\epsilon$ fixed so as to obtain the \textit{tightest} possible bound. With this we have proven the claim of \Cref{lem:LHL}. 
\end{proof} \label{endOfLemma}

\subsection*{\textbf{Proof of \Cref{thm:LHL3}}}
\begin{proof}[\unskip\nopunct]\
\begin{enumerate}
\item Consider \Cref{lem:LHL2}: we need the right hand side of the inequality to approach $0$ as $n \to \infty$ to show our wiretap coding scheme is semantically secure. We have as $n \to \infty$ that $R_n \to \Rs$ and $R_{\code_n} \to R_\code$. Since $\Rs$ is finite then $\lim_{n \to \infty} \epsilon n R_n = 0$ by the assumption that $\epsilon n \to 0$ as $n \to \infty$. Now if $\lim_{n\to\infty} (R_{\code_n} - R_n - \frac{\Iemax}{n}) > 0$ then the first term in the sum on the right hand side of \Cref{lem:LHL2} will also go to 0. But this is equivalent to \[\Rs < R_\code - \lim_{n \to \infty} \frac{\Iemax}{n}. \] If the right hand side is non-positive however, we will instead choose $\Rs = 0$ since rates must be non-negative. 
\ \\
\item Consider \Cref{lem:LHL2} again. Since $\lim_{n \to \infty} \frac{\Iemax}{n} \leq \xi$ by assumption, we can bound the asymptotic leakage as 
\begin{align*}
&\lim\limits_{n \to \infty} \max_{P_M} I(M \sep Z^n)\\
&\leq \lim\limits_{n \to \infty} \left( \frac{1}{\ln 2} 2^{-\frac{n}{2} (R_{\code_n} - R_n)} 2^{\frac{n}{2} \frac{\Iemax}{n}} + \epsilon n R_n \right)\\
&= \frac{1}{\ln 2} 2^{\lim\limits_{n \to \infty}( -\frac{n}{2} (R_{\code_n} - R_n))}2^{\lim\limits_{n \to \infty}( \frac{n}{2} ) \cdot \lim\limits_{n \to \infty} \frac{\Iemax}{n}} + \\
&\qquad\qquad\qquad\qquad \cdots+\lim\limits_{n \to \infty} \epsilon n R_n \\
&\leq \frac{1}{\ln 2} 2^{\lim\limits_{n \to \infty}( -\frac{n}{2} (R_{\code_n} - R_n))}2^{\lim\limits_{n \to \infty}( \frac{n}{2} ) \cdot \xi} + \lim\limits_{n \to \infty} \epsilon n R_n \\
&= \frac{1}{\ln 2} 2^{\lim\limits_{n \to \infty}( -\frac{n}{2} (R_{\code_n} - R_n - \xi))} + \lim\limits_{n \to \infty} \epsilon n R_n.
\end{align*}
At this point we can continue exactly as in part 1).
\ \\
\item Clearly the first of the two summands on the right hand side of the conclusion of  \Cref{lem:LHL2} is exponentially decreasing when $\Rs$ satisfies the rates given in either (1) or (2) above. Thus, if $\epsilon n R_n$ is exponentially decreasing with $n$, the semantic leakage is exponentially decreasing to 0; i.e.\ $\wiretap$ is  exceptionally semantically secure. For $\epsilon n R_n$ to be exponentially decreasing, it suffices for $\epsilon$ to be exponentially decreasing.
\end{enumerate}
This concludes the proof of \Cref{thm:LHL3}.
\end{proof}

\setcounter{lemma}{0}
\section{Removing the assumption of a public seed}\label{appendix:seedremoval}
In this appendix we shall overview a method that removes the assumption of a \textit{public} seed without rate/security/reliability loss. This method is called \textit{seed recycling} and can be found in \cite{semanticallySecure} and \cite{UHF}.

We have seen in \Cref{thm:LHL3} that our wiretap coding scheme can provide semantic security for certain achievable rates (provided that we prove a bound on the max-information rate), however, we have assumed hitherto that the seed $S$ was publicly available to all parties. This is in \textit{strict} violation of assumptions on a wiretap channel; that is, \textit{all} communication must take place over the wiretap channel. In this section, we remove this assumption and transmit the seed over the wiretap channel. We will show asymptotically that no rate, security, or reliability is lost.

As a first attempt to resolve this violation, suppose the seed is transmitted before beginning transmission of an actual message. This is a problem, however, because it leads to \textit{information rate loss} as follows. Suppose the seed can be transmitted with a probability of error less than some $p_{e,n}$ to the intended receiver in $n c$ channel uses for some constant $c > 1$. Then the transmitter sends $k$ message bits of information in another $n$ channel uses. Overall, $k$ bits of information were transferred in $n + nc = n (1 + c)$ channel uses, thus our overall secure information rate in this case is given asymptotically by \[\lim\limits_{n \to \infty} \frac{k}{n(1+c)} = \frac{1}{1+c} \Rs < \Rs,\] where $\Rs$ is the previous secure achievable rate assuming the seed was public. In other words, the possible asymptotic rates now achievable when sending the seed before message transmission are \textit{strictly} less than before. Therefore, in this case, the rates achieved using \Cref{thm:LHL3} are no longer possible.

As a better attempt to resolve this problem, suppose we use the same seed to send $\eta$ messages $M_1,M_2, \dots, M_\eta$ using $\eta$ \textit{independent} instances of the wiretap channel. First we will pick a block-length $n$ and on the first instance of the wiretap channel, we will send the seed over in $n c$ channel uses, where $c > 1$ is chosen so that the seed's probability of error at the intended receiver is less than or equal to $p_{e,n}$. Pessimistically (from the point of view at the transmitter), we will assume that the eavesdropper always receives a perfect copy of the seed. Now on each of the $\eta$ independent channel instances, we will send a corresponding message using the same scheme as outlined in \cref{scheme} except using the \textit{same seed} for each instance. Let $\textbf{M} = (M_1, M_2, \dots, M_\eta)$ be the vector consisting of the $\eta$ messages and let $\textbf{Z} = (Z^n(1),Z^n(2),\dots,Z^n(\eta))$ where $Z^n(i)$ is the $n$-letter eavesdropper output corresponding to the $i$-th message (also to the $i$-th channel instance).

Consider first the rate of this new procedure. In each of the $\eta$ channel uses, we are sending $k$ bits of information. Moreover, we will end up using the channel $\eta \cdot n$ times for the messages and $n c$ times for the seed. Overall, the asymptotic secure rate of this new procedure is thus given by \[\lim_{n \to \infty} \frac{\eta k}{\eta n + c n} = \lim_{n \to \infty} \frac{k}{n (1 + c/\eta)} = \frac{\Rs}{\lim_{n \to \infty} (1 + c/\eta)},\] where $\Rs$ is again the previous asymptotic secure achievable rate when the seed was public. Since $c$ is a constant, the only way to avoid information rate loss asymptotically is if $\eta \to \infty$ as $n \to \infty$.

Consider next the reliability of this new procedure. If each message has probability of error at the intended receiver bounded by $p_{e,n}$, then the probability that $\textbf{M}$ is in error is given in the next lemma.

\begin{lemma}[Reliability]\label{lem:seedReliablility}
The probability that $\textbf{M}$ is in error is upper bounded by \[1-(1-p_{e,n})^\eta.\]
\end{lemma}
\begin{proof}
Let $A_i$ be the event corresponding to the $i$-th message being in error. Then $A =\bigcup_{i=1}^\eta A_i$ is the event corresponding to at least one of the $\eta$ messages being in error. Hence $\P(A)$ is the probability of error of $\textbf{M}$.

Then since each instance of the wiretap channel is independent, we have the following.
\begin{align*}
\P(A) &= 1- \P(A^\complement)\\
&= 1- \P\left( \bigcap_{i=1}^\eta A_i^\complement \right) \\
&= 1- \prod_{i=1}^\eta \P(A_i^\complement)\\
&= 1- \prod_{i=1}^\eta (1-\P(A_i))\\
&\leq 1- \prod_{i=1}^\eta (1-p_{e,n} )\\
&= 1- (1-p_{e,n})^\eta.
\end{align*}
This concludes the proof of \Cref{lem:seedReliablility}.
\end{proof}

With this lemma, we see that in order to transmit reliably, we have another constraint on $\eta$, that is, we must choose $\eta$ so that $(1-p_{e,n})^\eta \to 1$ as $n \to \infty$. 

Consider last the leakage of this new procedure. 
\begin{lemma}[Security]\label{lem:seedLeakage} For some $i\in \{1,\dots, \eta\}$ the following holds:
\[\max\limits_{P_\textbf{M}} I(\textbf{M} \sep \textbf{Z} ) \leq \eta \cdot \max\limits_{P_M} I(M_i \sep Z^n(i) |S). \] 
\end{lemma}
\begin{proof}
Let $\textbf{M}$ have an arbitrary distribution $P_\textbf{M}$. By the chain rule of mutual information, \[I(\textbf{M} \sep \textbf{Z}) \leq I(\textbf{M} \sep \textbf{Z}, S).\] Since $S \bot M_i$ for each $i$, then $S \bot \textbf{M}$. Then by the chain rule of mutual information again, we have, \[I(\textbf{M} \sep \textbf{Z}) \leq I(\textbf{M} \sep \textbf{Z}| S).\] 

Now $(M_1,Z^n(1)),\ldots,(M_\eta,Z^n(\eta))$ are mutually independent once we are given $S$, thus by a standard mutual information inequality we have \[I(\textbf{M} \sep \textbf{Z}) \leq \sum\limits_{i = 1}^\eta I(M_i \sep Z^n(i)|S).\]
We want to maximize $I(\textbf{M} \sep \textbf{Z})$ over all probability distributions $P_\textbf{M}$. However, that is equivalent to maximizing over each choice of $P_{M_i}$ individually. The above becomes: \[\max\limits_{P_\textbf{M}} I(\textbf{M} \sep \textbf{Z}) \leq \sum\limits_{i = 1}^\eta \max\limits_{P_{M_i}} I(M_i \sep Z^n(i)|S).\]

Here $i$ represents an instance of the wiretap channel. Choose the channel instance $j$ that corresponds to the most leakage $\max_{P_{M_j}} I(M_j \sep Z^n(j)|S)$ leaked to the eavesdropper. The above then becomes \[\max\limits_{P_\textbf{M}} I(\textbf{M} \sep \textbf{Z}) \leq \eta \max\limits_{P_{M_j}} I(M_j \sep Z^n(j)|S).\]
This concludes the proof of \Cref{lem:seedLeakage}.
\end{proof}
This lemma intuitively says that the message leakage of all $\eta$ wiretap channel instances is no more than the number of channel instances multiplied by the leakage over the ``worst case'' wiretap channel (worst here is with respect to the transmitter). Combining this result with \Cref{lem:LHL} and \Cref{lem:LHL2} gives the following proposition.
\begin{prop} Let $i$ be the wiretap channel instance where the transmitter leaks the most information to the eavesdropper. Let $R_{\code_n}$ be the rate of the ECC and $R_n$ the secure rate of transmission for that wiretap channel instance. It follows that
\[\max\limits_{P_\textbf{M}} I(\textbf{M} \sep \textbf{Z} ) \leq  \frac{\eta}{\ln 2}2^{-\frac{n}{2} \left(R_{\code_n} - R_n - \frac{1}{n}\Iemax \right)} + \epsilon \eta n R_n.\]
\end{prop}

With this, just as in \Cref{thm:LHL3}, we see that if $R_n < R_{\code_n} - \frac{1}{n}\Iemax$ for each $n$, then so long as $\eta$ grows with $n$ strictly slower than exponential, the first term will go to 0. Furthermore, $\eta$ must be chosen slow enough so that $\epsilon \eta n \to 0$ as $n \to \infty$.

In summary, with regards to how $\eta$ must grow with $n$ we need the following as $n \to \infty$:
\begin{itemize}
\item $\eta \to \infty$ to guarantee negligible rate loss,
\item $(1-p_{e,n})^\eta \to 1$ to guarantee negligible reliability loss,
\item $\eta$ must grow slower than exponential in $n$ and $\epsilon \eta n \to 0$ to guarantee negligible security loss.
\end{itemize}

It will depend on the specific choice of $\epsilon$ and $p_{e,n}$ in each case in order to properly determine $\eta$, however, if for example $\epsilon$ is exponentially diminishing with $n$ and $p_{e,n}$ diminishes on the order of $1/n$, then picking $\eta$ on the order of $\log(n)$ will be sufficient to satisfy all of the previous requirements. Indeed, there is significant flexibility in these three parameters and finding them to satisfy the previous requirements should not be too intrusive. 

Intuitively, the previous has a nice interpretation. It says that as long as we keep on adding new independent messages when increasing the block length, we can still achieve the same rate, reliability, and security asymptotically as before when we assumed the seed to be public.

\setcounter{lemma}{0}
\section{Proofs from \Cref{sec:App1}} \label{appendix:awgnMaxInfoBound}
In this appendix we will prove two statements from the first applications section. We first prove \Cref{lem:awgnDMCmaxinfo}, which simplifies the expression of max information. Then we provide a reworked proof of \cite[Lemma 6]{UHF} (\Cref{lem:awgnmaxinfobound}) as an aid for our proof of \Cref{thm:ImaxLessCE} in \Cref{appendix:NOCSIT}.

\subsection*{\textbf{Proof of \Cref{lem:awgnDMCmaxinfo}}}
\begin{proof}[\unskip\nopunct]
Recall that $X^n = e_n(M')$ is a random variable over $\codebook_n$ with the same distribution as $M'$. By \Cref{lem:mUniform}, this means $X^n$ is uniform over $\codebook_n$. Since $\codebook_n$ has $2^l$ elements then $\omega(x^n) = 2^{-l}$. Now consider the following string of equalities.
\begin{align*}
&I_\infty^{\T}(X^n \sep Z^n)\\
&= \log|\codebook_n| - H_\infty^{\T}(X^n|Z^n)\\
&= l + \log \int_{\Z^n} \omega(z^n) \max_{x^n \in \codebook_n} \omega_\T(x^n|z^n) \mu(dz^n)\\
&=\log  \int_{\Z^n}  \max_{x^n \in \codebook_n} \frac{\omega(x^n,z^n)}{2^{-l}} \indicator\left( (x^n,z^n) \in \T\right) \mu(dz^n)\\
&= \log \int_{\Z^n} \max_{x^n \in \codebook_n} \omega_\T(z^n|x^n) \mu(dz^n).
\end{align*}
This proves the validity of \Cref{lem:awgnDMCmaxinfo}.
\end{proof}

\newcommand{\pout}{\mathcal{P}_\text{out}}
\newcommand{\pnoise}{\mathcal{P}_\text{noise}}

\begin{lemma}[\text{\cite[Lemma 6]{UHF}}]\label{lem:awgnmaxinfobound} Let $\delta > 0$ small. Then for any $(1-\epsilon)$ typical set $\T$ where $\epsilon = \exp(-n\delta^n/8)$, the asymptotic $\epsilon$-smooth average max-information of an AWGN eavesdropper channel $E$ is bounded by the point-to-point capacity:
\[ \lim_{n\to\infty} \frac{1}{n}  \log \int\limits_{\R^n} \max\limits_{x^n \in \codebook_n} \omega_\T (z^n|x^n) dz^n \leq C_E.\]
\end{lemma}
\begin{proof}
Define a set \[\pout = \{z^n \in \R^n \,|\, ||z^n||^2 \leq n(P + \sigma_E^2)(1 + \delta)\}.\] Also for each $x^n \in \codebook_n$ define a set \[\pnoise^{x^n}= \{z^n\,|\, ||z^n - x^n||^2 \geq n \sigma_E^2 (1- \delta) \}.\] 

Now let $\T_\text{out}, \T_\text{noise} \subset \codebook_n \times \R^n$ be sets defined as $\T_\text{out} = \codebook_n \times \pout$ and $\T_\text{noise} = \{(x^n,z^n)  \,|\, z^n \in \pnoise^{x^n} \text{ for each } x^n \in \codebook_n\}$. Then define a set $\T = \T_\text{out} \cap \T_\text{noise}$. 

It was shown in \cite{UHF} that $\T$ is a $(1-\epsilon)$-typical set using the given $\epsilon$. Note that $\epsilon \to 0$ exponentially fast with $n$. With this we have the following.
\begin{align*}
& \int \limits_{\Rn} \max_{x^n \in \codebook_n} \omega(z^n|x^n) \indicator((x^n,z^n) \in \T) dz^n\\
&\explainequals{1} \int \limits_{\R^n} \max_{x^n \in \codebook_n} \Biggr[ \left( \prod\limits_{i=1}^n  \frac{1}{\sqrt{2\pi\sigma_E^2}} \exp\left(-\frac{(z_i - x_i)^2}{2\sigma_E^2} \right)\right)\cdots \\
&\hspace*{5cm}  \cdots \indicator((x^n,z^n) \in \T) \Biggr] dz^n\\
&= \frac{1}{(2\pi\sigma_E^2)^{\frac{n}{2}}} \int\limits_{\R^n} \max_{x^n \in \codebook_n} \Biggr[  \exp\left(-\frac{\norm{z^n- x^n}^2}{2\sigma_E^2}\right)\cdots\\
&\hspace*{5cm} \cdots \indicator((x^n,z^n) \in \T) \Biggr] dz^n\\
&\explainlessthanequals{2} \frac{\exp\left(-\frac{n}{2} (1-\delta)\right)  }{(2\pi \sigma_E^2)^{\frac{n}{2}}}  \int \limits_{\R^n} \max_{x^n \in \codebook_n} \indicator((x^n,z^n) \in \T) dz^n \\
&\explainlessthanequals{3} \frac{\exp\left(-\frac{n}{2} (1-\delta)\right)  }{(2\pi \sigma_E^2)^{\frac{n}{2}}}  \int \limits_{\pout} dz^n\\
&\explainequals{4} \frac{\exp\left(-\frac{n}{2} (1-\delta)\right)  }{(2\pi \sigma_E^2)^{\frac{n}{2}}}  \text{Vol}(\pout)\\
&\explainequals{5}  \frac{\exp\left(-\frac{n}{2} (1-\delta)\right)  }{(2\pi \sigma_E^2)^{\frac{n}{2}}}  \frac{(\pi n (P+ \sigma_E^2)(1+\delta))^{\frac{n}{2}}}{\Gamma(n/2 +1)}.
\end{align*}

\begin{explain}\
\begin{enumerate}
\item On an AWGN channel, given that $x_i$ was sent, we know that each output is a normal random variable with mean $x_i$ and variance $\sigma_E^2$. Since we assume the channel is memoryless, we can split this density simply into a product.
\item We are working on $\T$ in the integral and thus $\pnoise$. Thus, $\norm{z^n- x^n}^2 \geq n \sigma_E^2 (1-\delta)$.
\item The indicator function returns either 0 or 1 in the area of interest $\pout \cap \pnoise$ and 0 elsewhere. Thus, we can simply upper bound the indicator by 1 everywhere inside of $\pout$.
\item Consider the following equalities: \[\int \limits_{\pout} dz^n = \int_{\R^n} \indicator(z^n \in \pout)dz^n = \mu(\pout) = \text{Vol}(\pout).\]
\item $\pout$ is clearly a ball in real $n$ space of radius $n (P + \sigma_E^2)(1 + \delta)$. The volume of an $n$ ball of radius $r$ is given by \[\frac{\pi^{n/2}}{\Gamma(n/2 +1)} r^n,\] where here $\Gamma$ is the gamma function (generalized factorial) from analysis.
\end{enumerate}
\end{explain}
Taking the logarithm of both sides of the preceding and dividing by $n$ yields:
\begin{align*}
& \frac{1}{n}  \log \int\limits_{\R^n} \max\limits_{x^n \in \codebook_n} \omega_\T (z^n|x^n) dz^n\\
&\leq \frac{1}{n} \log \left( \frac{\exp\left(- (1-\delta)\right)  }{2}  \frac{( n (1+ P/\sigma_E^2)(1+\delta))}{\Gamma(n/2 +1)^{2/n} } \right)^{n/2}\\
&= \frac{1}{2}\biggr(\log\left(1 + \frac{P}{\sigma_E^2}\right) + \log\left((1+\delta)e^\delta\right) + \cdots\\
&\hspace*{3cm} \cdots + \log\left(\frac{1}{2e}\cdot \frac{n}{\Gamma(n/2 +1)^{2/n}} \right)  \biggr)\\
&= C_E + \frac{1}{2}\log\left((1+\delta)e^\delta\right) + \frac{1}{2}\log\left(\frac{1}{2e}\cdot \frac{n}{\Gamma(n/2 +1)^{2/n}} \right).
\end{align*}

Fortunately, $\frac{n}{\Gamma(n/2 +1)^{2/n}} \to 2e$ as $n \to \infty$. Moreover, our choice of $\delta$ is not restricted and can be made arbitrarily small. This completes the proof of \Cref{lem:awgnmaxinfobound}.
\end{proof}

\setcounter{lemma}{0}
\section{Sphere Packing Argument for No-CSIT Channels}\label{appendix:SpherePacking}
In this appendix we provide motivation for how we constructed the typical set in the No-CSIT scenario. We provide sphere packing bounds in this case that are analogous to their AWGN counterparts (cf. \cite{cover_thomas,tse}). 

The capacity expression for an additive white Gaussian noise channel (AWGN) is motivated by an intuitive argument called \textit{sphere packing}. The argument asserts that due to properties of Gaussian random variables, a received output vector should be contained in some small $n$-dimensional ball around the transmitted codeword with high probability. In other words, the noise of the channel will only disturb the input vector by a certain amount (the radius of the small ball) with high probability. Furthermore, all received outputs should be contained in some larger ball with high probability since we are assuming that all the codewords are being transmitted while obeying the power constraint. If we use maximum likelihood decoding, given an output that resides in one of the small balls, the receiver assumes it came from the codeword that generated said ball. Therefore, the maximum number of small spheres we can pack into the larger ball roughly corresponds to how many codewords we can transmit reliably. This technique is called sphere packing since we are attempting to \textit{pack} the larger ball with smaller spheres. Exact calculation is quite challenging; however, simply dividing the volume of the large ball by the volume in a small sphere gives an upper bound. What is perhaps surprising is that as the block length approaches infinity, this upper bound is actually achievable and is exactly the capacity of the AWGN channel. 

We will provide a symmetric argument for the fast fading channel as justification for how and why we choose our typical sets the way we do in the No-CSIT case. Given an input $x^n$ and channel coefficient $h^n$, we know the output $z^n$ will reside in some small ball about the point $h^nx^n$ with high probability since we assume the noise follows a Gaussian distribution. In fact, such a ball will have radius $\sqrt{n \sigma_E^2 (1 + \delta)}$ for $\delta >0$ small. 

In the case of the AWGN channel, the larger ball's dimensions were derived using the fact that we expect our channel to obey the law of conservation of energy; that is, the maximum output energy should be equal to the summation of the maximum input energy and noise energy. We expect a similar phenomenon to hold on the fast fading channel; however, the input energy will also depend on the channel coefficient realization. During the $i$th symbol transmission, suppose $h_i$ is the realized channel coefficient; then the effective maximum input power is given by $h_i^2 P$ so that the effective maximum average output power $\frac{1}{n} Z_i^2$ is given by $h_i^2 P + \sigma_E^2$. Therefore we expect the realization $z_i^2$ to be less than $n (h_i^2 P + \sigma_E^2)(1+\delta)$.

Since $i$ is a \textit{coordinate} of the vector $z^n$, we should then expect $z^n$ to be found in some volume where each component $z_i$ is bounded by $\pm \sqrt{n (h_i^2 P + \sigma_E^2)(1+\delta)}$. Because $h_i$ is \textit{changing} for each use of the channel, each of these bounds will be different. Therefore, in contrast to the AWGN channel where each upper bound was constant with respect to each component, the volume in this case is actually an $n$-dimensional ellipsoid with radii $\sqrt{n (h_i^2 P + \sigma_E^2)(1+\delta)}$. Thus, if we try to pack as many spheres into this ellipsoid as possible as illustrated in the (2-dimensional) \Cref{fig:spherepack}, we should come up with the maximum number of codewords we can transmit reliably, i.e., an expression for capacity. 

Using the same technique as \cite{cover_thomas}, we simply divide the volume of the ellipsoid by the volume of the small balls. That is, since the volume of an ellipsoid with radii $r_i$ is given by $c \cdot \prod_{i=1}^n r_i$ where $c$ is the same constant factor used to calculate the volume of an $n$-dimensional ball, it follows that an upper bound to the max number of codewords is given by:
\begin{align*}
&\frac{c \cdot \prod\limits_{i=1}^n \sqrt{n (h_i^2 P + \sigma_E^2)(1+ \delta)}}{c \cdot \sqrt{n \sigma_E^2 (1 + \delta)}^n}\\
&= \frac{\prod\limits_{i=1}^n \sqrt{n \sigma_E^2 (1 + h_i^2 \SNR )(1+ \delta)}}{ \sqrt{n \sigma_E^2 (1 + \delta)}^n}\\
&= \prod\limits_{i=1}^n \sqrt{1+h_i^2 \SNR}.
\end{align*}

Since \text{rate} is usually defined as the logarithm of the number of codewords normalized by $n$, an upper bound to the max achievable rate is given by:
\begin{align*}
\frac{1}{n} \log \prod\limits_{i=1}^n \sqrt{1+h_i^2 \SNR} &= \frac{1}{2}\left(\frac{1}{n}\log\prod\limits_{i=1}^n (1+h_i^2 \SNR)
\right)\\
&= \frac{1}{2}\left(\frac{1}{n} \sum\limits_{i=1}^n \log(1+h_i^2 \SNR)\right) \\
&\xrightarrow{n\to\infty} \frac{1}{2} \E\left[\log(1+H_E^2 \SNR) \right]\\ 
&= C_E,
\end{align*}
where the convergence follows from the law of large numbers.

Since the above characterizations correctly estimated the asymptotic upper bound for the fast fading channel using the same sphere packing argument as in the AWGN case, we are confident moving forward that these bounds will produce sets that are typical in the proper sense. 

\begin{figure}[h]
  \begin{center}
  	\def\Scale{1.25}
\begin{tikzpicture}[scale=0.63, every node/.style={scale=0.63}]
\draw circle (.2959118*\Scale);
\draw ellipse (\Scale*4 and 2*\Scale);
\foreach \i in {0,1,...,6}
    \draw (60*\i:.5918236*\Scale) ellipse (.2959118*\Scale); 
\foreach \i in {0,1,...,12}
    \draw ({80+30*\i}:1.1436472*\Scale) circle(.2959118*\Scale);
\foreach \i in {0,1,...,18}
    \draw ({10+20*\i}:1.722045*\Scale) circle(.2959118*\Scale);
\draw (1.7:2.25*\Scale) circle(.2959118*\Scale);
\draw (17.5:2.22*\Scale) circle(.2959118*\Scale);
\draw (42.5:2.26*\Scale) circle(.2959118*\Scale);
\draw (30:2.55*\Scale) circle(.2959118*\Scale);
\draw (22:3.0*\Scale) circle(.2959118*\Scale);
\draw (12:2.7*\Scale) circle(.2959118*\Scale);
\draw (13.2:3.35*\Scale) circle(.2959118*\Scale);
\draw (2.9:3.1*\Scale) circle(.2959118*\Scale);
\draw (5:3.65*\Scale) circle(.2959118*\Scale);
\draw (-3:3.7*\Scale) circle(.2959118*\Scale);
\draw (-11:3.5*\Scale) circle(.2959118*\Scale);
\draw (-2:2.7*\Scale) circle(.2959118*\Scale);
\draw (-4:3.2*\Scale) circle(.2959118*\Scale);
\draw (-12:2.3*\Scale) circle(.2959118*\Scale);
\draw (-13:2.9*\Scale) circle(.2959118*\Scale);
\draw (-22:3.0*\Scale) circle(.2959118*\Scale);
\draw (-25:2.3*\Scale) circle(.2959118*\Scale);
\draw (-40:2.25*\Scale) circle(.2959118*\Scale);
\draw (17.5:-2.22*\Scale) circle(.2959118*\Scale);
\draw (42.5:-2.26*\Scale) circle(.2959118*\Scale);
\draw (30:-2.55*\Scale) circle(.2959118*\Scale);
\draw (22:-3.0*\Scale) circle(.2959118*\Scale);
\draw (10:-2.8*\Scale) circle(.2959118*\Scale);
\draw (13.2:-3.35*\Scale) circle(.2959118*\Scale);
\draw (5:-3.65*\Scale) circle(.2959118*\Scale);
\draw (-4:-3.6*\Scale) circle(.2959118*\Scale);
\draw (-13:-3.5*\Scale) circle(.2959118*\Scale);
\draw (1.5:-2.3*\Scale) circle(.2959118*\Scale);
\draw (-1.7:-3*\Scale) circle(.2959118*\Scale);
\draw (-13.5:-2.3*\Scale) circle(.2959118*\Scale);
\draw (-13:-2.9*\Scale) circle(.2959118*\Scale);
\draw (-25:-2.88*\Scale) circle(.2959118*\Scale);
\draw (-28:-2.3*\Scale) circle(.2959118*\Scale);
\draw (-43:-2.28*\Scale) circle(.2959118*\Scale);

\draw[thick,->] (0,0) -- (-4*\Scale,0) node[anchor=east] {$\sqrt{n (h_1^2 P + \sigma_E^2)(1+ \delta)}$};
\draw[thick,->] (0,0) -- (0,2*\Scale) node[anchor=south] {$\sqrt{n (h_2^2 P + \sigma_E^2)(1+ \delta)}$};
\draw[thick,->] (2.75,1.59) -- (3.05, 1.85) node[anchor=east] {};
\node[text width=2cm] at (3.85,2.5) {$\sqrt{n \sigma_E^2 (1 + \delta)}$};
\node[anchor=east] at (3.1,1.81) (arrow) {};
\node[anchor=west] at (2.91,2.5) (label) {};
\draw (label) edge[out=180,in=180,->] (arrow);
\end{tikzpicture}
  \end{center}
  \caption{Sphere packing for the fading channel.}
  \label{fig:spherepack}
\end{figure}
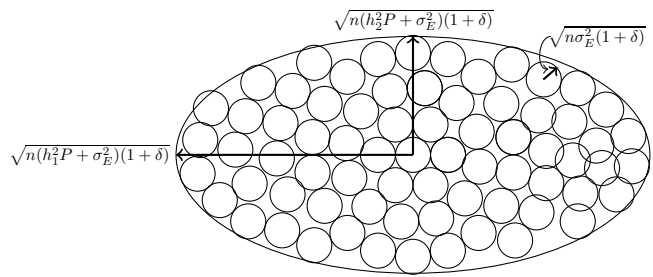

\setcounter{lemma}{0}
\section{Fading: Proof of \Cref{lem:maxinfoFading}, \Cref{fact:ComplexFadingtoRealFading}, and \Cref{lem:noCSIT}} \label{appendix:fading}
In this appendix we prove results related to fading. In particular we prove \Cref{lem:maxinfoFading} (a simplification of max information on fast fading channels), \Cref{fact:ComplexFadingtoRealFading} (the usual result converting complex fast fading channels into real parallel fading channels when CSIR is available), and \Cref{lem:noCSIT} (further simplification of max-information in the No-CSIT case).

\ \\ \\
\subsection*{\textbf{Proof of \Cref{lem:maxinfoFading}}}
\begin{proof}[\unskip\nopunct]
The proof follows directly.
\begin{align*}
&2^{I_\infty^\T(X^n \sep Z^n,\Lambda^n)}\\
&=|\codebook_n| \int\limits_{\Z^n \times \H^n} \omega(z^n,\lambda^n) \max_{x^n \in \codebook_n} \omega_\T(x^n|z^n,\lambda^n)\mu(dz^n,d\lambda^n)\\
&\explainequals{1}  \int\limits_{\H^n} \int\limits_{\Z^n} \max_{x^n \in \codebook_n} \frac{\omega_\T(z^n,x^n,\lambda^n)}{\omega(x^n)} \mu(dz^n) \mu(d\lambda^n)\\
&\explainequals{2} \E_{\Lambda^n} \int\limits_{\Z^n} \max_{x^n \in \codebook_n} \frac{\omega_\T(z^n,x^n,\Lambda^n)}{\omega(x^n)\omega(\Lambda^n)} \mu(dz^n) \\
&\explainequals{3}  \E_{\Lambda^n} \int\limits_{\Z^n} \max_{x^n \in \codebook_n} \omega_\T(z^n|x^n,\lambda^n) \mu(dz^n).
\end{align*}
\begin{explain}\
\begin{enumerate}[1)]
\item First recall that $|\codebook_n| = 2^l$. Next, $X^n = e_n(M')$ has the same distribution as $M'$; that is, $\omega(x^n) = 2^{-l}$ by \Cref{lem:mUniform}. Thus, we can move the $2^l$ inside of the maximization then convert to $\omega(x^n)$. Moreover, we can move $\omega(z^n,\lambda^n)$ inside of the maximization since it does not depend on the maximizing variable $x^n$. Lastly note that each $\mu$ here is not equivalent to each other or the measure from the previous line, it is denoted such simply for notational convenience.
\item We can multiply by the unit $\omega(\lambda^n)/\omega(\lambda^n)$ inside of the maximization. Then we can pull out $\omega(\lambda^n)$ since it does not depend on the maximizing variable $x^n$.
\item Here we are using the assumption that $X^n \bot \Lambda^n$ so that $\omega(x^n)\omega(\lambda^n) = \omega(x^n,\lambda^n)$. Then we use the definition of a conditional probability density function.
\end{enumerate}
\end{explain}
This concludes the proof of \Cref{lem:maxinfoFading}.
\end{proof}

\subsection*{\textbf{Proof of \Cref{fact:ComplexFadingtoRealFading}.}}
\begin{proof}[\unskip\nopunct]
Without loss of generality, consider the intended receiver's channel given above and drop the index $i$ for simplicity. Therefore, we are working with the complex fading channel $Y = H_T X + U_T$. Since $H_T \in \C$ we can write $H_T = |H_T| e^{i \theta}$ and thus, the receiver will receive the random variable $Y = |H_T|e^{i \theta} X + U_T$. However, since we are assuming channel state information is available at the receiver, the receiver actually knows the realization of $H_T$ and hence knows the value $e^{i \theta}$. The receiver thus \textit{adjusts} his output $Y$ accordingly: $Y e^{-i \theta} = |H_T| X + U_T e^{-i \theta}$. Also, the additive white Gaussian noise is assumed to be circularly symmetric, so that $U_T e^{-i \theta}$ is actually distributed the same way as was $U_T$. Therefore, if we define $\tilde{Y} = Y e^{-i \theta}$ as the new output and $\tilde{U}_T = U_T e^{-i \theta}$ as the rotated noise, under the assumption of CSIR, the receiver can convert the original channel into the new channel: $\tilde{Y} = |H_T| X + \tilde{U}_T$. Now we can break up \textit{this} channel into its real and imaginary parts:
\begin{align*}
\tilde{Y}_R + i\tilde{Y}_I &= \left(|H_T|X_R + i |H_T|X_I\right) + \left((\tilde{U}_T)_R + i (\tilde{U}_T)_I\right).
\end{align*}
Combining the real and imaginary parts respectively yields two parallel channels
\begin{align*}
\tilde{Y}_R &= |H_T|X_R + (\tilde{U}_T)_R\\
\tilde{Y}_I &= |H_T|X_I +(\tilde{U}_T)_I.
\end{align*}
Here each output is identically given as \[Y' = |H_T|X' + U_T' \] where $|H_T| \in \Rp, X' \in \R$, $U_T' \sim \mathcal{N}(0, \sigma_T^2)$, and $\E\left[ (X')^2 \right] \leq P$.
This concludes the proof of \Cref{fact:ComplexFadingtoRealFading}.
\end{proof}

\subsection*{\textbf{Proof of \Cref{lem:noCSIT}.}}
\begin{proof}[\unskip\nopunct]
From \Cref{lem:maxinfoFading} and recalling that $\Lambda^n = (H_T^n,H_E^n)$, we have:
\begin{align*}
&2^{I_\infty^\T(X^n \sep Z^n, H_T^n,H_E^n)} \\
&= \E_{H_T^n H_E^n}  \int\limits_{\Rn} \max_{x^n \in \codebook_n}\omega_{\T}(z^n|x^n,H_T^n, H_E^n) dz^n \\
&\explainequals{1} \int\limits_{\Rpn} \omega(h_E^n) \int\limits_{\Rpn} \omega(h_T^n) \int \limits_{\Rn} \max_{x^n \in \codebook_n}\omega_{\T}(z^n|x^n,h_E^n) dz^n dh_T^n dh_E^n\\
&= \int\limits_{\Rpn} \omega(h_E^n) \int \limits_{\Rn} \max_{x^n \in \codebook_n}\omega_{\T}(z^n|x^n,h_E^n) dz^n  dh_E^n\\
&=\E_{H_E^n}  \int\limits_{\Rn} \max_{x^n \in \codebook_n}\omega_{\T}(z^n|x^n,H_E^n) dz^n.
\end{align*}
\begin{explain} \ 
\begin{enumerate}
\item Independence of $H_E^n$ and $H_T^n$. Also, $Z_i = H_{E,i}X_i + U_{E,i}$ and $X_i$ is \textit{not} a function of the channel coefficients since we have No-CSIT; therefore, $Z^n$ is independent of $H_T^n$.
\end{enumerate}
\end{explain}
This concludes the proof of \Cref{lem:noCSIT}.
\end{proof}

\setcounter{lemma}{0}
\section{No-CSIT: Proof of \Cref{lem:Psets}, \Cref{lem:typicaldefinition}, and \Cref{thm:ImaxLessCE}} \label{appendix:NOCSIT}
In this appendix we prove the main results related to the No-CSIT fast fading wiretap channel. In particular, we prove \Cref{lem:Psets}, \Cref{lem:typicaldefinition} (proves that the sets we defined for No-CSIT are actually typical), and \Cref{thm:ImaxLessCE} (one of our main results that proves a bound on max-information in the No-CSIT scenario).

\subsection*{\textbf{Proof of \Cref{lem:Psets}}}
To prove \Cref{lem:Psets}, we will first need a fact and a lemma. The fact is due to \cite{sung} where we have modified its form so as to be easily utilized in the following proofs. It can be considered a generalization of Hoeffding's inequality \cite{hoeffding} to the case of \textit{unbounded} random variables.

\begin{fact}\label{lem:sunglemma}\cite[Theorem 2.1]{sung}
Let $\{W_i\}_{i=1}^n$ be a sequence of independent random variables. Suppose  for all $i$ there exists a $\gamma_i > 0$ such that $\E\left[e^{\gamma_i |W_i|}\right] < \infty.$
Then for any sufficiently small $a > 0$,
\[ \P\left[\left|\frac{1}{n}\sum_{i=1}^{n} (W_i-\E[W_i])\right|\leq a \right] \geq 1 - 2e^{-\frac{n a^2}{4K^*}}\]
where $K_i =2(\E\left[W_i^{4}\right])^{\frac{1}{2}}\E\left[e^{a|W_i|}\right]$ and $K^*=\max\limits_i K_i$.
\end{fact}

The second item that will be needed for the proof of \Cref{lem:Psets} is the following.
\begin{lemma} \label{lem:expectedvaluebound} The following inequality holds:
\[\E\left[ \frac{1}{n} \sum_{i=1}^n \frac{Z_i^2}{\sigma_E^2 + H_{E,i}^2 P} \, \biggr| \, X^n = x^n\right] \leq 1.\]
\end{lemma}
\begin{proof}
\begin{align*}
&\E\left[ \frac{1}{n} \sum_{i=1}^n \frac{Z_i^2}{\sigma_E^2 + H_{E,i}^2 P} \, \biggr| \, X^n = x^n \right]\\
&= \frac{1}{n} \sum_{i=1}^n \E\left[\dfrac{H_{E,i}^2 x_i^2 + U_{E,i}^2 + 2H_{E,i} x_i U_{E,i}}{\sigma_E^2 + H_{E,i}^2 P} \right] \\
&\explainequals{1} \frac{1}{n} \sum_{i=1}^n x_i^2 \cdot \E \left[\dfrac{H_{E,i}^2}{\sigma_E^2 + H_{E,i}^2 P} \right]+\cdots\\
& \cdots+\E U_{E,i}^2 \cdot \E \left[\dfrac{1}{\sigma_E^2 + H_{E,i}^2 P} \right] + \E U_{E,i} \cdot \E\left[\dfrac{2 x_i^2 H_{E,i} }{\sigma_E^2 + H_{E,i}^2 P} \right] \\
&\explainequals{2}  \left(\frac{1}{n} \sum_{i=1}^n  x_i^2 \right)  \cdot \E \left[\dfrac{H_{E}^2}{\sigma_E^2 + H_{E}^2 P} \right] + \E \left[\dfrac{\sigma_E^2}{\sigma_E^2 + H_{E}^2 P} \right] \\
&\explainlessthanequals{3} \E \left[\dfrac{H_{E}^2 P}{\sigma_E^2 + H_{E}^2 P} \right] + \E \left[\dfrac{\sigma_E^2}{\sigma_E^2 + H_{E}^2 P} \right] \\
&= 1.
\end{align*}
\begin{explain}\ 
\begin{enumerate}[1)]
\item Follows from independence of $H_E$, $U_E$.
\item $U_E$ is i.i.d.\ and $ \sim \mathcal{N}(0,\sigma_E^2)$.
\item Follows from the power constraint on all codewords.
\end{enumerate}
\end{explain}
This completes the proof of \Cref{lem:expectedvaluebound}.
\end{proof}

With these tools in hand, we now give the proof of \Cref{lem:Psets}.
\begin{enumerate}
\item \begin{proof}[\unskip\nopunct] \textit{Proof of \Cref{lem:Psets}.1.}\\ 
Let\footnote{Note that $\hat{\mu}$ is the mean here, i.e. it is a number, and is not related to the measure $\mu$.} $\hat{\mu} = \E\left[ \frac{1}{n} \sum_{i=1}^n \frac{Z_i^2}{\sigma_E^2 + H_{E,i}^2 P} \, \biggr| \, X^n = x^n \right]$. Then,
\begin{align*}
&\P \left[\left( H_E^n,Z^n\right) \in \Pout_n \, \biggr|\, X^n = x^n \right]\\
&= \P\left[\frac{1}{n} \sum\limits_{i=1}^n \frac{Z_i^2}{\sigma_E^2 + H_{E,i}^2 P} - 1 \leq \delta_n \, \biggr|\, X^n = x^n \right]\\
&\geq \P\left[\frac{1}{n} \sum\limits_{i=1}^n \frac{Z_i^2}{\sigma_E^2 + H_{E,i}^2 P} - \hat{\mu} \leq \delta_n \, \biggr|\, X^n = x^n \right]\\
&= \P\left[\frac{1}{n} \sum\limits_{i=1}^n \frac{x_i^2 H_{E,i}^2 + U_{E,i}^2 + 2x_i H_{E,i} U_{E,i}}{\sigma_E^2 + H_{E,i}^2 P} - \hat{\mu} \leq \delta_n \right] \numberthis \label{eq:sidewaystriangle},
\end{align*}
where the inequality follows from \Cref{lem:expectedvaluebound}.

Since $x^n$ is a constant and $\{H_{E,i}\}$ and $\{U_{E,i}\}$ are each mutually independent, the term \[\frac{x_i^2 H_{E,i}^2 + U_{E,i}^2 + 2x_i H_{E,i} U_{E,i}}{\sigma_E^2 + H_{E,i}^2 P}\] is an independent random variable. Let us show that it also satisfies the main condition of \Cref{lem:sunglemma} (dropping the subscript $E$ on $H_{E,i}$ and $U_{E,i}$ to reduce clutter).

\begin{align*}
&\E \exp\left(\gamma \frac{x_i^2 H_i^2 + U_i^2 + 2x_i H_i U_i}{\sigma^2 + H_i^2 P}\right) \\
&= \E \exp\left(\gamma \frac{x_i^2 H_i^2 + U_i^2 + 2x_i H_i U_i}{\sigma^2 + H_i^2 P}\right)\indicator\left(H_i > 1 \right)  + \\
&\;\cdots+\E \exp\left(\gamma\frac{x_i^2 H_i^2 + U_i^2 + 2x_i H_i U_i}{\sigma^2 + H_i^2 P}\right)\indicator\left(H_i \leq 1 \right) \\
&\explainlessthanequals{1} \E \exp\left(\gamma\left( \frac{x_i^2 H_i^2}{ H_i^2 P} + \frac{U_i^2}{\sigma^2} +\frac{2x_i H_i^2 U_i}{ H_i^2 P} \right)\right)\indicator\left(H_i > 1 \right) +\\
&\;\cdots+\E \exp\left(\gamma\left( \frac{x_i^2 H_i^2}{H_i^2 P} + \frac{U_i^2}{\sigma^2} +\frac{2x_i U_i}{\sigma^2} \right)\right)\indicator\left(H_i \leq 1 \right) \\
&= \E \exp\left(\gamma\left( \frac{x_i^2}{ P} + \frac{U_i^2}{\sigma^2} +\frac{2x_i U_i}{ P} \right)\right)\\
&\quad+\E \exp\left(\gamma\left( \frac{x_i^2}{ P} + \frac{U_i^2}{\sigma^2}+\frac{2x_i U_i}{\sigma^2} \right)\right)\\
&\leq \E \exp\left(2\gamma\left( \frac{U_i}{\sigma} + x_i \frac{P + \sigma^2}{2 P \sigma}\right)^2\right) \\
&\explainlessthanequals{2} \E \exp\left( 2\gamma\left( G_i \right)^2\right) \\
&\explainlessthan{3} \infty.
\end{align*}
\begin{explain}\
\begin{enumerate}[1)]
\item $H_{i} > 1$ implies $H_{i} \leq H_{i}^2$.
\item $G_i \sim \mathcal{N}(x_i \frac{P+\sigma^2}{2 P \sigma}$, 1) implies that $G_i^2$ is a non-central $\chi^2$ random variable.
\item Choosing $\gamma$ appropriately ensures the moment generating function is finite.
\end{enumerate}
\end{explain}

Since a finite moment generating function implies \textit{every} moment is finite, $K_i$ exists for all $i$ so that $K^*$ is well defined. Therefore, using \Cref{lem:sunglemma}, it follows immediately that 
\[(\ref{eq:sidewaystriangle}) \; \geq 1 - 2e^{-\frac{n\delta_n^2}{4K^*}},\] thereby completing the proof of Lemma 8.1.
\end{proof}
\ \\
\item \begin{proof}[\unskip\nopunct] \textit{Proof of \Cref{lem:Psets}.2.}
\begin{align*}
&\P\left[ Z^n \in \Pnoise_n \, \biggr|\, X^n = x^n, H_E^n = h^n \right]\\
&= \P\left[ \norm{Z^n - x^nh^n}^2 \geq n \sigma_E^2 (1 - \delta'_n)\right. \\
&\hspace{4.25cm} \left.\biggr| X^n = x^n, H_E^n = h^n  \right]\\
&= \P\left[\sum\limits_{i=1}^n \left(Z_i-x_ih_i\right)^2 \geq n \sigma_E^2(1-\delta'_n)\right.\\
&\hspace{4.25cm}\left.\biggr| X^n = x^n, H_E^n = h^n \right] \\
&= \P\left[\sum\limits_{i=1}^n \left(U_i +x_ih_i -x_ih_i\right)^2 \geq n \sigma_E^2(1-\delta'_n)  \right] \\
&= \P\left[\frac{1}{\sigma_E^{2}}\sum\limits_{i=1}^n U_i^2 \geq n(1-\delta'_n)  \right] \\
&\explaingreaterthanequals{1} 1 - e^{-\frac{n {\delta'_n}^{2}}{4}}.
\end{align*}

\begin{explain} \ 
\begin{enumerate}[1)]
\item Chi-squared tail bounds \cite[Lemma 1]{laurent}.
\end{enumerate}
\end{explain}
\end{proof}

\item \begin{proof}[\unskip\nopunct] \textit{Proof of \Cref{lem:Psets}.3.}\\ 
To prove this, we will use \Cref{lem:sunglemma} reduced to the i.i.d.\ case. We have that $\{\log(1+H_{E,i}^2\SNR)\}$ is a sequence of i.i.d. random variables; to employ \Cref{lem:sunglemma} it remains to prove that $\E\left[e^{\gamma |\log(1+H_E^2SNR)|}\right] < \infty$ for some $\gamma > 0$.
\begin{align*}
\E\left[e^{\gamma |\log(1+H_E^2\SNR)|}\right] &= \E\left[e^{\gamma \frac{\ln(1+H_E^2 \SNR)}{\ln(2)}}\right]\\
&= \E\left[(1+H_E^2 \SNR)^{\frac{\gamma}{\ln(2)}}\right]\\
&\text{Letting $\gamma = \ln 2$:}\\
&= \E\left[(1+H_E^2 \SNR)\right] \\
&= 1 + \E[H_E^2]\SNR\\
&< \infty.
\end{align*}
Then \Cref{lem:sunglemma} gives us:
\begin{align*}
&\P \left[H_E^n \in \Perg_n \right]\\
&= \P\left[ \left| \frac{1}{n} \sum\limits_{i=1}^n \log\left(1+H_{E,i}^2 \SNR \right)\right.\right.\cdots \\
&\qquad\qquad\qquad\cdots \left.\left.\vphantom{\sum\limits_{i=1}^n}- \E\left[1 + H_E^2 \SNR \right] \right| \leq \delta''_n\right]\\
&\geq 1 - 2e^{-\frac{n{\delta''_n}^2}{4K}},
\end{align*}
where \[K=2\left(\E\left[\log(1+H_E^2 \SNR)^{4}\right]\right)^{\frac{1}{2}} \E\left[e^{\gamma \log(1+H_E^2 \SNR)}\right].\]
\end{proof}
\end{enumerate}
At this point we have finished the proof of \Cref{lem:Psets}.

\subsection*{\textbf{Proof of \Cref{lem:typicaldefinition}.}}
\begin{proof}[\unskip\nopunct]
\begin{align*}
&\P\left[(H_T^n,X^n,H_E^n,Z^n) \in \Rp^n\times\T_n |X^n = x^n \right]\\
&\explaingreaterthanequals{1}\P\left[(H_T^n) \in \Rp^n |X^n = x^n \right]+\cdots \\
&\qquad \cdots +\P\left[(X^n,H_E^n,Z^n) \in \T_n |X^n = x^n \right] - 1\\ 
&=\P\left[(X^n,H_E^n,Z^n) \in \T_n |X^n = x^n \right]\\
&=\P\left[(X^n,H_E^n,Z^n) \in \Tout_n \cap \Tnoise_n \cap \Terg_n |X^n = x^n \right]\\
&\explaingreaterthanequals{2} \P\left[(X^n,H_E^n,Z^n) \in \Tout_n|X^n = x^n \right] + \cdots \\
&\qquad\cdots +\P\left[(X^n,H_E^n,Z^n) \in \Tnoise_n |X^n = x^n \right]+\cdots \\
&\qquad \cdots +\P\left[(X^n,H_E^n,Z^n) \in \Terg_n |X^n = x^n \right] -2 \\
&\explainequals{3} \P\left[(H_E^n,Z^n) \in \Pout_n \,\biggr|\,X^n = x^n \right]  +\cdots \\
&\quad \cdots +\E_{H_E^n} \left( \P\left[Z^n \in \Pnoise_n \biggr|H_E^n = h^n, X^n = x^n \right]\right) +\cdots\\
&\quad \cdots +\P\left[ H_E^n \in \Perg_n \right] -2\\
&\explaingreaterthanequals{4} (1-\epsilon_n^1) + (1-\epsilon_n^2) + (1-\epsilon_n^3) - 2\\
&= 1- (\epsilon_n^1 + \epsilon_n^2 + \epsilon_n^3)\\
&= 1- \epsilon_n.
\end{align*}

\begin{explain}\
\begin{enumerate}
\item Fr\'echet inequality for Cartesian products.
\item Fr\'echet inequality for intersections.
\item The second term of the sum is explained here:
\begin{align*}
&\P\left[(X^n,H_E^n,Z^n) \in \Tnoise_n |X^n = x^n \right] \\
&= \int\limits_{\H^n}\int\limits_{\Z^n}\omega(z^n, h_E^n | x^n) \indicator((x^n,h_E^n,z^n)\in \T_n^2) dz^ndh_E^n\\
&=\int\limits_{\H^n}\int\limits_{\Z^n}\frac{\omega_{\T_n^2}(z^n| h_E^n, x^n)\omega(h_E^n, x^n)}{\omega(x^n)}  dz^ndh_E^n\\
&=\int_{\H^n}\omega(h_E^n) \int_{\Z^n}\omega_{\T_n^2}(z^n| h_E^n, x^n)dz^ndh_E^n\\
&=\E_{H_E^n} \left( \P\left[Z^n \in \Pnoise_n \biggr|H_E^n = h^n, X^n = x^n \right]\right).
\end{align*}
\item This line follows immediately from \Cref{lem:Psets}.
\end{enumerate}
\end{explain}
This completes the proof of \Cref{lem:typicaldefinition}.
\end{proof}

\subsection*{\textbf{Proof of \Cref{thm:ImaxLessCE}}}
\begin{proof}[\unskip\nopunct]
The proof of \Cref{thm:ImaxLessCE} follows directly and is analogous to our proof of \Cref{lem:awgnmaxinfobound} found in \Cref{appendix:awgnMaxInfoBound}.
\begin{align*}
&2^{\Iemaxf}\\
&\explainlessthanequals{1} 2^{I_\infty^{\T_n}(X^n \sep Z^n, H_T^n,H_E^n)}\\
&\explainequals{2} \E_{H_E^n} \int \limits_{\Rn} \max_{x^n \in \codebook_n} \omega_{\T_n}(z^n|x^n,H_E^n) dz^n\\
&\explainequals{3}\E_{H_E^n} \int \limits_{\R^n} \max_{x^n \in \codebook_n} \left[ \left( \prod\limits_{i=1}^n  \frac{1}{\sqrt{2\pi\sigma_{E}^2}} e^{-\frac{1}{2\sigma_E^2}(z_i -H_{E,i} x_i)^2}\right)\right.\cdots \\
&\hspace{3.5cm}\cdots\left.\indicator((x^n,H_E^n,z^n) \in \T_n) \vphantom{\prod\limits_{i=1}^n} \right] dz^n\\
 &=\frac{1}{(2\pi\sigma_{E}^2)^{\frac{n}{2}}} \E_{H_E^n} \int \limits_{\R^n} \max_{x^n \in \codebook_n} \left(  e^{-\frac{1}{2\sigma_E^2}\norm{z^n-H_{E}^n x^n}^2}\right. \cdots\\
&\hspace{3.5cm}\cdots\left.\indicator((x^n,H_E^n,z^n) \in \T_n)\vphantom{e^{-\frac{1}{2\sigma_E^2}\norm{z^n-H_{E}^n x^n}^2}} \right) dz^n\\
&\explainlessthanequals{4} \frac{e^{-\frac{n}{2} (1-\delta_n')}}{(2\pi \sigma_{E}^2)^{\frac{n}{2}}}  \E_{H_E^n}  \int \limits_{\R^n} \max_{x^n \in \codebook_n} \indicator((x^n,H_E^n,z^n) \in \T_n) dz^n \\
&\explainequals{5}  \frac{e^{-\frac{n}{2} (1-\delta_n')}}{(2\pi \sigma_{E}^2)^{\frac{n}{2}}}  \E_{H_E^n} \int \limits_{\R^n} \max_{x^n \in \codebook_n} \indicator((x^n,H_E^n,z^n) \in \Tout_n \cap \Tnoise_n)\cdots\\
&\hspace{3.5cm}\cdots\indicator((x^n,H_E^n,z^n) \in \Terg_n)dz^n  \\
&=  \frac{e^{-\frac{n}{2} (1-\delta_n')}}{(2\pi \sigma_{E}^2)^{\frac{n}{2}}} \E_{H_E^n} \left[\vphantom{\int \limits_{\R^n}}\indicator(H_E^n \in \Perg_n)\right.\cdots\\
&\qquad\quad\cdots\left.\left(\, \int \limits_{\R^n} \max_{x^n \in \codebook_n} \indicator((x^n,H_E^n,z^n) \in \Tout_n \cap \Tnoise_n) dz^n \right) \right] \numberthis \label{eq:diamond}.
\end{align*}
\begin{explain} \ 
\begin{enumerate}
\item $\Rp^n\times\T_n$ is a $(1-\epsilon_n)$ typical set; however, it may not be the set corresponding to the ``smallest'' $\epsilon_n$ smooth max-information. Note that here we are labeling our typical set as just $\T_n$ for ease and dropping the subscript on $\epsilon_n$.
\item \Cref{lem:noCSIT}. Since we no longer have any dependencies on $H_T$, we will henceforth write our typical set as just $\T_n$.
\item Each output, given $X_i =x_i$ and $H_{E,i}=h_{E,i}$, is $Z_i = h_{E,i}x_i + U_{E,i}$. This is simply a normal random variable that is shifted in mean by $h_{E,i}x_i$ with variance $\sigma_E^2$. Thus, the density for each transmission is given as $$\omega(z_i|x_i,h_{E,i}) = \frac{1}{\sqrt{2 \pi \sigma_E^2}}e^{-\frac{1}{2\sigma_E^2}(z_i-h_{E,i}x_i)^2}.$$
Since we assume the channel is memoryless, we can split this density simply into a product.
\item We are working on $\T_n$ and thus $\Pnoise_n$; thus, $\norm{z^n-h_{E}^n x^n}^2 \geq n \sigma_E^2 (1-\delta_n')$.
\item $\Tout_n,\Tnoise_n,\Terg_n$ are defined in \Cref{sub:NoCSIT}.
\end{enumerate}
\end{explain}

Let us gain some intuition of what is happening at this point. In \Cref{eq:diamond}, suppose $\T_n^{\star} = \Tout_n \cap \Tnoise_n$ and let us understand the term \[\max_{x^n \in \codebook_n} \indicator((x^n,h_E^n,z^n) \in \T_n^{\star}).\] If we temporarily \textit{fix} $z^n$ and $h_E^n$, then this maximization is simply asking if there exists some codeword $x^n \in \codebook_n$ that makes the sequence $(x^n,h_E^n,z^n)$ an element of the set $\T_n^{\star}$. If there does exist such an $x^n$ then this function returns 1; otherwise, it returns 0. If we now relax $z^n$ and only fix $h_E^n$, $\T_n^{\star}$ can be thought of as a typical set as well: it is the set of typical input-output pairs. Thus the above function takes some output $z^n$ and asks if there is possibly any codewords that could have generated such an output knowing the channel coefficient is $h_E^n$. It follows then, that the integral \[\int \limits_{\R^n} \max_{x^n \in \codebook_n} \indicator((x^n,h_E^n,z^n) \in \T_n^{\star}) dz^n,\] roughly ``counts'' the number of valid input-output pairs given some $h_E^n$. 

To calculate such an integral, we need to know the \textit{shape} of $\T_n^{\star}$ and it is clear that $\T_n^{\star} = \Tout_n \cap \Tnoise_n \subset \Tout_n$ so that we can replace the $\T_n^{\star}$ with a $\Tout_n$ in the above integral at the expense of an inequality. However, this has \textit{removed} the maximization since $\Tout_n$ has no dependence on codewords. Therefore the above integration is less than or equal to \[\int \limits_{\R^n}  \indicator((h_E^n,z^n) \in \Pout_n) dz^n.\] 
Given some $h_E^n$, by definition this integral is equal to the Lebesgue measure of $\Pout_n$ which is precisely the volume of $\Pout_n$. Since $\Pout_n$ is actually an ellipsoid with radii, $\sqrt{n\sigma_E^2(1+h_{E,i}^2 \SNR) (1+\delta_n)}$, then this integration is actually calculating the \textit{volume} of said ellipsoid, which is calculated to be \[\dfrac{\pi^{\frac{n}{2}}}{\Gamma(\frac{n}{2}+1)} \prod_{i=1}^n\sqrt{n\sigma_E^2(1+h_{E,i}^2 \SNR) (1+\delta_n)},\] where $\Gamma$ is the usual gamma function of analysis.

Let us return to \Cref{eq:diamond}; using the aforementioned reasoning above we have:
\begin{align*}
(\ref{eq:diamond}) &\leq \frac{e^{-\frac{n}{2} (1-\delta'_n)}}{(2\pi \sigma_{E}^2)^{\frac{n}{2}}}  \E_{H_E^n} \left[\left(\vphantom{\prod_{i=1}^n}\,\, \frac{\pi^{\frac{n}{2}}}{\Gamma(\frac{n}{2}+1)} \indicator(H_E^n \in \Perg_n)\right.\right.\cdots\\
&\hspace{2cm}\cdots\left.\left.\prod_{i=1}^n\sqrt{n\sigma_e^2(1+H_{E,i}^2 \SNR) (1+\delta_n)} \right) \right]\\
&= \frac{e^{-\frac{n}{2} (1-\delta'_n)}}{(2\pi \sigma_{E}^2)^{\frac{n}{2}}} \frac{\pi^{\frac{n}{2}}}{\Gamma(\frac{n}{2}+1)} (n \sigma_E^2 (1+\delta_n))^{\frac{n}{2}} \cdots\\
&\qquad\cdots\E_{H_E^n}\left[\left(\,\, \indicator(H_E^n \in \Perg_n)\prod_{i=1}^n\sqrt{(1+H_{E,i}^2 \SNR)} \right) \right] \\
&= \left( (1+\delta_n) e^{\delta'_n} \frac{n}{2e \cdot \Gamma(\frac{n}{2}+1)^{\frac{2}{n}}} \right)^{\frac{n}{2}}\cdots\\
&\hspace{2.5cm}\cdots\int\limits_{\Perg_n} \omega(h_E^n)\prod_{i=1}^n\sqrt{(1+h_{E,i}^2 \SNR)} dh_E^n\\
&\explainlessthanequals{6} \left( (1+\delta_n) e^{\delta'_n} \frac{n}{2e \cdot \Gamma(\frac{n}{2}+1)^{\frac{2}{n}}} \right)^{\frac{n}{2}} \cdots\\
&\hspace{2.5cm}\cdots\int\limits_{\Perg_n} \omega(h_E^n)2^{\frac{n}{2}\left(\delta_n'' + \E_{H_E}[1+H_E^2 \SNR]\right)} dh_E^n\\
&= \left( (1+\delta_n) e^{\delta'_n} \frac{n}{2e \cdot \Gamma(\frac{n}{2}+1)^{\frac{2}{n}}} \right)^{\frac{n}{2}}\cdots\\
&\hspace{2.5cm}\cdots 2^{\frac{n}{2}\left(\delta_n'' + \E_{H_E}[1+H_E^2 \SNR]\right)}
\int\limits_{\Perg_n} \omega(h_E^n) dh_E^n\\
&= \left( \frac{n(1+\delta_n) e^{\delta'_n}}{2e \cdot \Gamma(\frac{n}{2}+1)^{\frac{2}{n}}} \right)^{\frac{n}{2}} 2^{\frac{n}{2}\left(\delta_n'' + \E_{H_E}[1+H_E^2 \SNR]\right)}\\
\end{align*}

\begin{explain}\ 
\begin{enumerate}\setcounter{enumi}{5}
\item Due to the bounds of integration we know that every value of $h_e$ will satisfy the definition of $\Perg_n$, thus it satisfies:
\begin{align*}
\frac{1}{n} \sum_{i=1}^n \log(1+h_{E,i}^2 \SNR) &- \E_{H_E}[1+H_E^2 \SNR] \leq \delta_n''\\
\Rightarrow \frac{1}{n}\log\left(\prod_{i=1}^n (1+h_{E,i}^2 \SNR)\right) &\leq \delta_n'' + \E_{H_E}[1+H_E^2 \SNR]\\
\intertext{Multiplying by $n$ and exponentiating both sides:}
\Rightarrow \prod_{i=1}^n (1+h_{E,i}^2 \SNR) &\leq 2^{n\left(\delta_n'' + \E_{H_E}[1+H_E^2 \SNR]\right)}\\
\Rightarrow \prod_{i=1}^n \sqrt{(1+h_{E,i}^2 \SNR)} &\leq 2^{\frac{n}{2}\left(\delta_n'' + \E_{H_E}[1+H_E^2 \SNR]\right)}
\end{align*}
\end{enumerate}
\end{explain}

Continuing from the last string of inequalities and equalities, we take the logarithm of the beginning and end, and divide by $n$:
\begin{align*}
&\dfrac{\Iemaxf}{n}\\
&\leq \frac{1}{n}\log\left[ \left( (1+\delta_n) e^{\delta_n'} \frac{n}{2e \cdot \Gamma(\frac{n}{2}+1)^{\frac{2}{n}}} \right)^{\frac{n}{2}}\right.\cdots\\
&\hspace{3.5cm}\cdots\left.2^{\frac{n}{2}\left(\delta_n'' + \E_{H_E}[1+H_E^2 \SNR]\right)}\vphantom{\left( (1+\delta_n) e^{\delta_n'} \frac{n}{2e \cdot \Gamma(\frac{n}{2}+1)^{\frac{2}{n}}} \right)^{\frac{n}{2}}}\right]\\
&= \frac{1}{2}  \log \left((1+\delta_n) e^{\delta'_n} \frac{n}{2e \cdot \Gamma(\frac{n}{2}+1)^{\frac{2}{n}}} \right) +\cdots\\
& \hspace{3.5cm}\cdots+\frac{1}{2}\left(\delta_n'' + \E_{H_E}[1+H_E^2 \SNR]\right) \\
&= \underbrace{\frac{1}{2} \log \left((1+\delta_n) e^{\delta_n'}\right)}_\text{A1} + \underbrace{\frac{1}{2} \log\left( \frac{n}{2e \cdot \Gamma(\frac{n}{2}+1)^{\frac{2}{n}}} \right)}_\text{A2}+\cdots\\
&\hspace{3.5cm}\cdots+\frac{1}{2}\left(\delta_n'' + \E_{H_E}[1+H_E^2 \SNR]\right).
\end{align*}

Let us see the asymptotic behavior of these first two terms.
\begin{itemize}
    \item [\textbf{A1.}] If we choose $\delta_n \to 0$ and $\delta_n' \to 0$ as $n \to \infty$ at rates sufficiently slow (so as to allow $1-\epsilon_n^1 \to 1$ and $1-\epsilon_n^2 \to 1$ resp.), then $\textnormal{A1} \to 0$ as $n\to \infty$.
    \item [\textbf{A2.}] It can be shown that $\frac{n}{2e \cdot \Gamma(\frac{n}{2}+1)^{\frac{2}{n}}} \to 1$ as $n\to \infty$ so that $\textnormal{A2} \to 0$ as $n\to \infty$.
\end{itemize}

Since we can choose $\delta_n,\delta'_n,\delta''_n$ in such a way so that $\delta_n'' \rightarrow 0$ and $\epsilon_n^1,\epsilon_n^2,\epsilon_n^3 \to 0$ as $n \to \infty$, it follows that $\epsilon_n \to 0$ as $n \to \infty$. Combing these previous steps yields our claim:
\[\lim\limits_{\substack{n \to \infty\\\epsilon \to 0}}  \frac{\Iemaxf}{n}\leq \frac{1}{2}\E_{H_E}\left[\log(1+H_E^2 \SNR)\right].\]

Thus, we have completed the proof of \Cref{thm:ImaxLessCE}.
\end{proof}

\setcounter{lemma}{0}
\section{Partial CSIT: Proof of \Cref{lem:PartialTypical} and \Cref{thm:ImaxPartialCSIT}} \label{appendix:partial}
In this final appendix, we shall prove results related to partial CSIT. In particular, we will prove \Cref{lem:PartialTypical} (proves that the sets we defined for partial CSIT are actually typical) and \Cref{thm:ImaxPartialCSIT} (another main result of our paper that proves an upper bound on max-information in the the partial CSIT scenario).
\subsection*{\textbf{Proof of \Cref{lem:PartialTypical}.}}
\begin{proof}[\unskip\nopunct]
We first see that $\T'_{n_i}$ is $(1-\epsilon_{n_{i}})$ typical directly from \Cref{sub:NoCSIT}. Then:
\begin{align*}
&\P\left[(X^n,H_T^n,H_E^n,Z^n) \in \T'_n |X^n = x^n \right]\\
&=\P\left[(X^n,H_T^n,H_E^n,Z^n) \in \T'_{n_1} \times\cdots\times \T'_{n_d} |X^n = x^n \right]\\
&\explaingreaterthanequals{1} \P\left[(X^{n_1},H_T^{n_1},H_E^{n_1},Z^{n_1}) \in \T'_{n_1}|X^{n_1} = x^{n_1}\right]+\cdots\\
&\cdots + \P\left[(X^{n_d},H_T^{n_d},H_E^{n_d},Z^{n_d}) \in \T'_{n_d}|X^{n_d} = x^{n_d}\right]-d+1\\
&\explaingreaterthanequals{2} \sum_{i=1}^{d} (1-\epsilon_{n_{i}}) -d+1\\
&=d-d+1 - \sum_{i=1}^{d} \epsilon_{n_{i}}\\
&\geq 1 - d\epsilon^{*}\quad  \text{(define $\epsilon^{*}$ be the largest $\epsilon_{n_i}$ over all $i$)}.
\end{align*}
Since $\epsilon^{*}$ is going to $0$ with $n\rightarrow \infty$ and we are free to choose $d$, we see that $\T'_n$ is a $(1-\epsilon_n)$ typical set.

\begin{explain}\
\begin{enumerate}
\item Fr\'echet inequality for Cartesian products.
\item We know that $\T'_{n_i}$ is a $(1-\epsilon_{n_i})$ typical set for all $i$ thus:
\[\P\left[(X^{n_i},H_T^{n_i},H_E^{n_i},Z^{n_i}) \in \T'_{n_i}|X^{n_i} = x^{n_i}\right] \geq 1-\epsilon_{n_i}\]
and we sum over all $i$.
\end{enumerate}
\end{explain}
This concludes the proof of \Cref{lem:PartialTypical}.
\end{proof}

\subsection*{\textbf{Proof of \Cref{thm:ImaxPartialCSIT}}}
\begin{proof}[\unskip\nopunct]
The proof follows in a similar fashion to both \Cref{lem:awgnDMCmaxinfo} and \Cref{thm:ImaxLessCE}.
\begin{align*}
&2^{\Iemaxf}\\
&\explainlessthanequals{1} 2^{I_\infty^{\T'_n}(X^n \sep Z^n, H_T^n,H_E^n)}\\
&= \E_{H_T^n H_E^n}  \int\limits_{\Rn} \max_{x^n \in \codebook_n}\omega_{\T'_n}(z^n|x^n,H_T^n,H_E^n) dz^n \\
 &= \E_{H_T^nH_E^n}\int\limits_{\Rn} \max_{x^n \in \codebook_n}\omega(z^n|x^n,H_T^n,H_E^n)\cdots\\
&\hspace{3.5cm}\cdots\indicator((x^n,H_T^n,z^n)\in \T'_n) dz^n\\
&= \int \omega(h_T^n,h_E^n) \int\limits_{\Rn} \max_{x^n \in \codebook_n}\omega(z^n|x^n,h_T^n, h_E^n)\cdots\\
&\hspace{2.7cm}\cdots\indicator((x^n,h_T^n,h_E^n,z^n)\in \T'_n) dz^n dh_T^n dh_E^n\\
&\explainlessthanequals{2} \prod_{i} \int \omega(h_T^{n_{i}},h_E^{n_{i}}) \int\limits_{\R^{n_{i}}} \max_{x^{n_{i}} \in \codebook_{n_{i}}^{i}}\omega(z^{n_{i}}|x^{n_{i}},h_T^{n_{i}},h_E^{n_{i}})\cdots\\
&\hspace{2cm}\cdots\indicator((x^{n_{i}},h_T^{n_{i}},h_E^{n_{i}},z^{n_{i}})\in \T'_{n_{i}}) dz^{n_{i}} dh_T^{n_{i}} dh_E^{n_{i}}\\
\intertext{(now let $\mathcal{J} = \left[h_{T,i},h_{T,i+1}\right)^{n_{i}} \times \Rp^{n_{i}}$)}
&\explainequals{3} \prod_{i} \int\limits_{\mathcal{J}} \omega(h_T^{n_{i}},h_E^{n_{i}}) \int\limits_{\R^{n_{i}}} \max_{x^{n_{i}} \in \codebook_{n_{i}}^{i}}\omega(z^{n_{i}}|x^{n_{i}},h_T^{n_{i}},h_E^{n_{i}})\cdots\\
&\hspace{2cm}\cdots\indicator((x^{n_{i}},h_T^{n_{i}},h_E^{n_{i}},z^{n_{i}})\in \T'_{n_{i}}) dz^{n_{i}} dh_T^{n_{i}} dh_E^{n_{i}}. \numberthis \label{eq:partialthm1}
\end{align*}
\begin{explain}\
\begin{enumerate}
\item $\T'_n$ is a $(1-\epsilon_n)$ typical set; however, it may not be the set corresponding to the ``smallest'' $\epsilon$ smooth max-information.
\item We wish to integrate over all $n$ and to do so, we break up the integral into integrals over each $n_{i}$.
\begin{enumerate}
\item Suppose $N_{i} \geq n_{i}$. In this case, we have transmitted a full $n_{i}$ length codeword over the $i$th channel and choose to \textit{not send} information over the channel during the remaining $N_{i} - n_{i}$ channel uses. Then:
\begin{align*}
&\hspace{-0.4cm}\int\limits_{\R^{N_{i} - n_{i}}} \hspace{-0.2cm}\max_{x^{N_{i} - n_{i}}}\omega(z^{N_{i} - n_{i}}|x^{N_{i} - n_{i}},h_T^{N_{i} - n_{i}},h_E^{N_{i} - n_{i}}) dz^{N_{i} - n_{i}}\\
&=\int\limits_{\R^{N_{i} - n_{i}}} \omega(z^{N_{i} - n_{i}}|h_T^{N_{i} - n_{i}},h_E^{N_{i} - n_{i}}) dz^{N_{i} - n_{i}}\\
&=1.
\end{align*}
Thus, if $N_{i} \geq n_{i} \, \forall i$ then we obtain equality at this line.
\item Suppose $N_{i} < n_{i}$. In this case, the $i$th channel did not appear often enough for the transmitter to send an entire $n_{i}$ length codeword. By not sending the full codeword, we are inherently limiting the amount of information sent across the channel and therefore the amount of information that can be leaked to the eavesdropper. Hence, sending the full $n_{i}$ length codeword allows more information (or equal amount of information) to be leaked to the eavesdropper and therefore serves as an upper bound to the actual value. More clearly:
\begin{align*}
&\int\limits_{\R^{N_{i}}} \omega(z^{N_{i}}|x^{N_{i}},h_T^{N_{i}},h_E^{N_{i}})\cdots\\
&\qquad\qquad\quad\cdots\indicator((x^{N_{i}},h_T^{N_{i}},h_E^{N_{i}},z^{N_{i}})\in \T'_{N_{i}}) dz^{N_{i}}\\
&\leq \int\limits_{\R^{n_{i}}} \omega(z^{n_{i}}|x^{n_{i}},h_T^{n_{i}},h_E^{n_{i}})\cdots\\
&\qquad\qquad\quad\cdots\indicator((x^{n_{i}},h_T^{n_{i}},h_E^{n_{i}},z^{n_{i}})\in \T'_{n_{i}}) dz^{n_{i}}.
\end{align*}
\end{enumerate}
\item Due to the partitioning of the channel coefficients, we know that for each $i$, $H_T \in \left[h_{T,i},h_{T,i+1}\right)$.
\end{enumerate}
\end{explain}

Continuing on we have:
\begin{align*}
(\ref{eq:partialthm1}) &\explainequals{4} \prod_{i} \E_{H_e^{n_{i}}}\int\limits_{\R^{n_{i}}} \max_{x^{n_{i}} \in \codebook_{n_{i}}^{i}}\omega_{\T_{n_{i}}}(z^{n_{i}}|x^{n_{i}},H_e^{n_i}) dz^{n_{i}}\\
&\explainequals{5} \prod_{i} 2^{I_\infty^{\T'_{n_i}}(X^{n_i} \sep Z^{n_i}, H_T^{n_i},H_E^{n_i})}.
\end{align*}

Taking the logarithm of each side and dividing by $n$ we have:
\begin{align*}
&\frac{\Iemaxf}{n} \\
&\leq \frac{1}{n}\log\left(\prod_{i}2^{I_\infty^{\T'_{n_i}}(X^{n_i} \sep Z^{n_i}, H_T^{n_i},H_E^{n_i})}\right) \\
&= \frac{1}{n} \sum_{i} \log\left(2^{I_\infty^{\T'_{n_i}}(X^{n_i} \sep Z^{n_i}, H_T^{n_i},H_E^{n_i})}\right)\\
&\explainlessthanequals{6} \frac{1}{n} \sum_{i} n_{i} \frac{1}{2}\E_{H_E}\left[\log\left(1+H_E^2 \frac{\gamma_{i}(h_{T,i})}{\sigma_E^2}\right)\right]\\
&=\frac{1}{2n} \sum_{i} (p_{i}n-\varepsilon_i) \E_{H_E}\left[\log\left(1+H_E^2 \frac{\gamma_{i}(h_{T,i})}{\sigma_E^2}\right)\right]\\
&=\frac{1}{2}\sum_{i} p_{i} \E_{H_E}\left[\log\left(1+H_E^2 \frac{\gamma_{i}(h_{T,i})}{\sigma_E^2}\right)\right] +\cdots\\
&\qquad \cdots- \frac{1}{2n}\sum_{i} \varepsilon_i\E_{H_E}\left[\log\left(1+H_E^2 \frac{\gamma_{i}(h_{T,i})}{\sigma_E^2}\right)\right]\\
&\explainconv{7}\frac{1}{2} \E_{H_E,H_T}\left[\log\left(1+\frac{\gamma(H_T)H_E^2}{\sigma_E^2}\right)\right],
\end{align*}
as $n \rightarrow \infty$ and $\epsilon \rightarrow 0$.

\begin{explain}\
\begin{enumerate}\setcounter{enumi}{3}
\item We can split up the conditional density as
\begin{align*}
\omega(z^{n_{i}} | x^{n_{i}}, h_T^{n_{i}}, h_E^{n_{i}}) &= \frac{\omega(z^{n_{i}}, x^{n_{i}}, h_E^{n_{i}})\omega(h_T^{n_{i}})}{\omega(x^{n_{i}},h_E^{n_{i}})\omega(h_T^{n_{i}})}\\
&=\omega(z^{n_{i}} | x^{n_{i}}, h_E^{n_{i}})
\end{align*}
where the first equality follows from the fact that $h_T^{n_{i}}$ is independent of $z^{n_{i}}, x^{n_{i}}, \text{ and } h_E^{n_{i}}$. Note that $h_T^{n_{i}}$ was indeed used to determine which codebook to use on this channel, but at this point that has been determined and we have restricted the integration of $h_T^{n_{i}}$ to take this into account, i.e. $x^{n_{i}}$ is independent of $h_T^{n_{i}}$. Thus the multiplicand becomes:
\begin{align*}
&\hspace{-0.4cm}\int\limits_{\mathcal{J}} \omega(h_E^{n_{i}})\omega(h_T^{n_{i}})\hspace{-0.12cm} \int\limits_{\R^{n_{i}}} \hspace{-0.15cm}\max_{x^{n_{i}} \in \codebook_{n_{i}}^{i}}\omega_{\T^{i}}(z^{n_{i}}|x^{n_{i}},h_E^{n_{i}})dz^{n_{i}} dh_E^{n_{i}} dh_T^{n_{i}}\\
&=\int\limits_{\Rp^{n_{i}}} \omega(h_E^{n_{i}}) \int\limits_{\R^{n_{i}}} \max_{x^{n_{i}} \in \codebook_{n_{i}}^{i}}\omega_{\T^{i}}(z^{n_{i}}|x^{n_{i}},h_E^{n_{i}})dz^{n_{i}} dh_E^{n_{i}}.
\end{align*}
The equality follows from the fact that we know $h_T \in \left[h_{T,i},h_{T,i+1}\right)$ for each component of the $n_{i}$ length vector for every $i$. Therefore integrating $\omega(h_T^i)$ over the whole space where $h_T$ is guaranteed to be will yield $1$ for each of the $\sum_{i}n_{i}$ integrals.
We then rewrite the integral over $h_E^{n_{i}}$ in the form of expected value.
\item Definition of $I_\infty^{\T'_{n_i}}(X^{n_i} \sep Z^{n_i}, H_T^{n_i},H_E^{n_i})$ and \Cref{lem:noCSIT}.
\item Upper bound as found in \Cref{thm:ImaxLessCE}.
\item $d$ can be made arbitrarily large and thus the channel coefficient intervals can be made arbitrarily small, hence the convergence of the first term to the expected value. For the second term: \[\lim\limits_{n\rightarrow \infty} \frac{1}{2n}\sum_{i} \varepsilon_i\E_{H_E}\left[\log\left(1+H_E^2 \frac{\gamma_{i}(h_{T,i})}{\sigma_E^2}\right)\right] = 0.\]
Since $\E_{H_E}\left[\log\left(1+H_E^2 \frac{\gamma_{i}(h_{T,i})}{\sigma_E^2}\right)\right]$ is constant with respect to $n$ and $\varepsilon_i \rightarrow 0$. Also, $\delta_n,\delta'_n,\delta''_n$ (from Theorem 1) can be chosen in such a way that $\epsilon\rightarrow 0$ as $n\rightarrow \infty$.
\end{enumerate}
\end{explain}
This concludes the proof of \Cref{thm:ImaxPartialCSIT}.
\end{proof}

\bibliographystyle{IEEEtran}
\bibliography{bibtex/biblio.bib}{}

\end{document}